\documentclass[12pt]{article}
\usepackage[UKenglish]{babel} % setting language
\usepackage[T1]{fontenc} % for correct display of accented characters
\usepackage[utf8]{inputenc} % for correct input of accented characters
\usepackage{graphicx} % for importing images
\usepackage{amsmath}
\usepackage{amsthm}
\usepackage{amssymb}
\usepackage{cite} % configuring citation
\usepackage{float} % fixing location of tables and figures
\usepackage{bm} % bold 3D vectors
\usepackage{indentfirst} % indent first paragraphs of sections
\usepackage{caption} % configuring captions for images
\usepackage{subcaption} % captioning subfigures
\usepackage{enumitem} % to specify enumerate styles
\usepackage{comment} % allows comment sections
\usepackage{csquotes}
\usepackage{hyperref} % hyperref, load at the end
\newtheorem{theorem}{Theorem} % Theorem environment
\newtheorem{lemma}{Lemma} % Theorem environment
\newtheorem{corollary}{Corollary}
\newtheorem{assumption}{Assumption}
\makeatletter
 \DeclareFontFamily{OMX}{MnSymbolE}{}
 \DeclareSymbolFont{MnLargeSymbols}{OMX}{MnSymbolE}{m}{n}
 \SetSymbolFont{MnLargeSymbols}{bold}{OMX}{MnSymbolE}{b}{n}
 \DeclareFontShape{OMX}{MnSymbolE}{m}{n}{
	 <-6>  MnSymbolE5
	<6-7>  MnSymbolE6
	<7-8>  MnSymbolE7
	<8-9>  MnSymbolE8
	<9-10> MnSymbolE9
   <10-12> MnSymbolE10
   <12->   MnSymbolE12
 }{}
 \DeclareFontShape{OMX}{MnSymbolE}{b}{n}{
	 <-6>  MnSymbolE-Bold5
	<6-7>  MnSymbolE-Bold6
	<7-8>  MnSymbolE-Bold7
	<8-9>  MnSymbolE-Bold8
	<9-10> MnSymbolE-Bold9
   <10-12> MnSymbolE-Bold10
   <12->   MnSymbolE-Bold12
 }{}
 
 \let\llangle\@undefined
 \let\rrangle\@undefined
 \DeclareMathDelimiter{\llangle}{\mathopen}%
					  {MnLargeSymbols}{'164}{MnLargeSymbols}{'164}
 \DeclareMathDelimiter{\rrangle}{\mathclose}%
					  {MnLargeSymbols}{'171}{MnLargeSymbols}{'171}
\makeatother

\newcommand{\const}{\,{\rm const}\,}
\newcommand{\diag}{\,{\rm diag}\,}

\newcommand{\td}{\text{d}}
\newcommand{\ord}{\mathcal{O}}

\title{\textbf{Supersymmetric black holes with a single axial symmetry in five dimensions}}
\author{David Katona\footnote{d.katona@sms.ed.ac.uk}\,\, and \,James Lucietti\footnote{j.lucietti@ed.ac.uk}
\\ \\ 
\small \sl School of Mathematics and Maxwell Institute for Mathematical Sciences, 
\\ 
\small \sl University of Edinburgh, King's Buildings, Edinburgh, EH9 3JZ, UK }

\date{}

\begin{document}
%\begin{titlepage}
	\maketitle
	\begin{abstract}
		We present a classification of asymptotically flat, supersymmetric black hole and soliton solutions of 
		five-dimensional minimal supergravity that admit a single axial symmetry which `commutes' with the supersymmetry. 
		This includes the first examples of five-dimensional black hole solutions with exactly one axial Killing field 
		that are smooth on and outside the horizon.
		The solutions have similar properties to the previously studied class with biaxial symmetry, in 
		particular, they have a Gibbons-Hawking base and the harmonic functions must be of multi-centred type with the centres corresponding to the connected components of the horizon or fixed 
		points of the axial symmetry. We find a large moduli space of black hole and soliton spacetimes with non-contractible 2-cycles and the  horizon topologies are $S^3$,  $S^1\times S^2$ and lens spaces $L(p,1)$.
	\end{abstract}	
%\end{titlepage}
\newpage
\tableofcontents

\newpage

\section{Introduction} \label{sec_Introduction}

The classification of higher dimensional stationary black hole solutions of general relativity remains a major open 
problem~\cite{emparan_black_2008}. The black hole uniqueness theorems of four dimensions do not generalise to higher 
dimensions, even for asymptotically flat vacuum spacetimes. This was first demonstrated by the discovery of an explicit 
counterexample known as the black ring~\cite{emparan_rotating_2002}, which is an asymptotically flat, five-dimensional, 
vacuum black hole solution with horizon topology $S^1\times S^2$. Furthermore, for a range of asymptotic charges, 
there exist two different black ring solutions, as well as a spherical Myers-Perry black hole solution.

A number of general results are known which constrain the topology and symmetry of higher-dimensional black hole spacetimes~\cite{hollands_black_2012}. 
These are particularly restrictive for asymptotically flat five-dimensional stationary spacetimes, which will be the focus of this paper. Topological 
censorship guarantees that the domain of outer communication (DOC) is simply connected~\cite{friedman_topological_1995}. The horizon topology theorem 
states  that cross-sections of the horizon must be $S^3, S^1\times S^2$, $S^3/\Gamma$ where $\Gamma$ is a discrete group (or connected sums 
thereof)~\cite{galloway_generalization_2006}. The rigidity theorem guarantees (under the assumption of analyticity) that rotating black holes 
must have an axial $U(1)$ symmetry that commutes with the stationary symmetry (for the non-extremal case see~\cite{hollands_higher_2007,Moncrief:2008mr}, and for generic extremal black holes see~\cite{Hollands:2008wn}). Motivated by this, further constraints on the topology of the horizon and the DOC have been derived 
for stationary black holes with a $U(1)$ axial symmetry~\cite{hollands_further_2011}. Furthermore, classification 
theorems have been proven for stationary black holes with a $U(1)^2$ biaxial symmetry~\cite{Hollands:2007aj,Hollands:2007qf,Hollands:2008fm}.

Black hole non-uniqueness is also present in five-dimensional (minimal) supergravity theory, even for asymptotically flat, supersymmetric 
black hole solutions. First, the BMPV solution was found~\cite{breckenridge_d--branes_1997}, which is a charged, rotating black hole with 
$S^3$ horizon topology and equal angular momenta in the two orthogonal 2-planes. A uniqueness theorem for the BMPV black hole  
has been proven for locally $S^3$ horizons under the assumption that the stationary Killing field (the existence of which is 
necessary by supersymmetry) is timelike outside the black hole~\cite{reall_higher_2004}. Later, supersymmetric black rings were 
constructed~\cite{elvang_supersymmetric_2004}, moreover, it was found that concentric black ring solutions can possess the same asymptotic 
charges as the BMPV black hole~\cite{gauntlett_concentric_2005}.

More recently, new classes of supersymmetric black holes with non-trivial topology have been found in this theory.  Black holes with lens 
space horizon topology were first found for $L(2,1)$ topology~\cite{kunduri_supersymmetric_2014} and then generalised to $L(p,1)$ 
topology~\cite{tomizawa_supersymmetric_2016, breunholder_moduli_2019, breunholder_supersymmetric_2019}. Black holes with $S^3$ horizons 
with non-trivial spacetime topology have also been constructed~\cite{kunduri_black_2014,  breunholder_moduli_2019, breunholder_supersymmetric_2019}. 
These solutions have a DOC with non-trivial topology due to the presence of non-contractible 2-cycles and are similar to the `bubbling' microstate 
geometries~\cite{bena_black_2008}.  These new types of black hole evade the aforementioned uniqueness theorem for the BMPV solution because 
they possess ergosurfaces on which the stationary Killing field is null. Interestingly, some of the black holes in bubbling spacetimes can have 
the same asymptotic charges as a BMPV black hole, and rather surprisingly, there even exist black holes whose horizon area exceeds that 
of the corresponding BMPV black hole with the same conserved charges~\cite{horowitz_comments_2017, breunholder_supersymmetric_2019}. This result 
is in conflict with the microscopic derivation of the BMPV black hole entropy in string theory~\cite{strominger_microscopic_1996, breckenridge_d--branes_1997}, a 
contradiction which remains to be resolved. This  highlights the importance of determining of the {\it full }moduli space of five-dimensional
supersymmetric black holes.

The general local form of supersymmetric solutions in minimal supergravity  has been known for some time~\cite{gauntlett_all_2003}. 
It is determined using Killing spinor bilinears which define a scalar function, a causal Killing field, and three 2-forms on spacetime. 
When the Killing field is timelike, the metric takes the form of a timelike fibration over a hyper-K\"ahler base space, and the  
2-forms are the complex structures of the base space. If there exists an axial symmetry of the base space that preserves the complex 
structures (i.e. the $U(1)$ action is triholomorphic), the base metric is a Gibbons-Hawking space~\cite{gibbons_hidden_1988}, 
and if the full solution is also invariant under this axial symmetry it is determined by four harmonic functions on $\mathbb{R}^3$~\cite{gauntlett_all_2003}. 
It turns out the known supersymmetric black hole solutions discussed above all belong to this class. However, despite this local solution being known 
for nearly 20 years, no general global analysis of supersymmetric solutions with a Gibbons-Hawking base has been performed\footnote{A global analysis of a subclass of supersymmetric solutions with a Gibbons–Hawking base which reduce to four-dimensional euclidean Einstein–Maxwell solutions was performed in~\cite{dunajski_einstein-maxwell_2007}.}. One of the purposes of 
this paper is to perform such an analysis to determine {\it all} asymptotically flat black hole and soliton spacetimes in this class.

In fact, a number of classification theorems for supersymmetric solutions of minimal supergravity are already known.
The near-horizon geometries of supersymmetric black holes in this theory were completely determined~\cite{reall_higher_2004}, 
and it was found that the horizon geometries are a (squashed) three-sphere $S^3$, lens space $L(p,q)$, and $S^1\times S^2$ 
(the $T^3$ geometry is excluded by~\cite{galloway_generalization_2006}). It turns out that in all cases the near-horizon geometry 
must have  a $U(1)\times U(1)$ biaxial symmetry group in addition to the stationary symmetry.  In fact,  all known regular black hole 
solutions in five dimensions possess such a biaxial symmetry.  Under the assumption of biaxial symmetry, a classification of asymptotically 
flat, supersymmetric black hole and soliton solutions in minimal supergravity has been achieved~\cite{breunholder_moduli_2019}. These solutions 
all have a Gibbons-Hawking base and the associated harmonic functions on $\mathbb{R}^3$ have collinear simple poles which correspond  to 
either horizon components or fixed points of the triholomorphic axial Killing field. This reveals a large moduli space of black holes with $S^3, S^1\times S^2$ 
and lens space $L(p,1)$ horizons and non-contractible 2-cycles in the DOC, which contains all the above examples. In particular, this rules out 
lens space horizons $L(p,q)$ for $q\neq 1$, at least in this symmetry class.

Reall conjectured that higher dimensional rotating black holes with exactly one axial symmetry should exist~\cite{reall_higher_2004}. Evidence 
supporting this has been obtained from approximate solutions~\cite{Emparan:2009vd} and the analysis of linearised 
perturbations of odd dimensional Myers-Perry black holes for $D\ge 9$~\cite{Dias:2010eu}\footnote{It is possible to construct near-horizon 
geometries with a single axial symmetry in $D\ge 6$~\cite{Kunduri:2010vg}.}. This conjecture was 
motivated by the rigidity theorem, which only applies to black holes that are rotating in the sense that the stationary Killing field is not null 
on the horizon. Even though supersymmetric black holes are  non-rotating in this sense (the stationary Killing field is null on the horizon), 
one may also expect supersymmetric black holes with a single axial symmetry to exist. There have been a number of constructions of such solutions 
in the literature~\cite{bena_one_2004, bena_black_2006, bena_sliding_2006, bena_coiffured_2014}, however, these have all resulted in solutions for which 
the metric or matter fields are not smooth at the horizon~\cite{horowitz_how_2005, candlish_smoothness_2010}. In this work we will show that one can easily 
construct examples of five-dimensional supersymmetric black hole solutions with a single axial symmetry, that are smooth on and outside the horizon, 
by working within the class of supersymmetric solutions with a Gibbons-Hawking base.

The main goal of this paper is to obtain a classification of asymptotically flat, supersymmetric black hole and soliton solutions to five-dimensional 
minimal supergravity, that possess a single axial symmetry and are smooth on and outside a horizon (if there is one). This generalises the classification 
derived under the stronger assumption of a biaxial symmetry~\cite{breunholder_moduli_2019}. Our main assumption is that the axial symmetry `commutes' 
with the  supersymmetry in the sense that it preserves the Killing spinor. It then easily follows that the $U(1)$ action is triholomorphic and commutes 
with the stationary Killing field. Our main result can be summarised in the following theorem (the full statement is given in Theorem \ref{thm_classification}).

\begin{theorem}
	Consider an asymptotically flat, supersymmetric  black hole or soliton solution of $D=5$ minimal supergravity, 
	with an axial symmetry that preserves the Killing spinor. In addition, assume that the domain of outer communication 
	is globally hyperbolic, on which the span of Killing fields is timelike. Then, the solution must have Gibbons-Hawking 
	base and the associated harmonic functions are of multi-centred type, where the poles correspond to connected components 
	of the horizon or fixed points of the axial symmetry, and the parameters must satisfy a complicated set of algebraic 
	equations and inequalities.  Furthermore, the cross-section of each horizon component must have  $S^3, S^1\times S^2$ 
	or lens space $L(p,1)$ topology.
	\label{thm_summ}
\end{theorem}

This theorem is completely analogous to the case with biaxial symmetry~\cite{breunholder_moduli_2019}. However, the method of proof is rather different since it requires an analysis of the possible three-dimensional orbit spaces. These have been analysed in detail in~\cite{hollands_further_2011}. We find that supersymmetry strongly constrains the orbit space and that it can be identified with the $\mathbb{R}^3$ base of the Gibbons-Hawking base.  In contrast, in the biaxially symmetric case the orbit space is a two-dimensional manifold with boundaries and corners, which can be identified with a half-plane where the boundary is divided into rods (this is encoded by the rod structure). In that case it was also found that supersymmetry constrains the possible orbit spaces (that is, the rod structures are constrained).  It is interesting that supersymmetry leads to such constraints on the spacetime topology.

 It turns out that the constraints 
on the parameters in Theorem \ref{thm_summ} are exactly the same as for the biaxisymmetric case, but there is no requirement for the centres to be collinear on $\mathbb{R}^3$.
As we will show, this means that generically these spacetimes have $\mathbb{R}\times U(1)$ symmetry. The existence of these 
solutions depends on whether a complicated set of  constraint equations and inequalities on the parameters can be simultaneously satisfied. 
Unfortunately, in general we do not have analytic control over these constraints.  However, for three-centred solutions, we present numerical 
evidence that solutions do exist in the case of non-collinear centres, which correspond to black holes with exactly one axial Killing field. 
This extends  the systematic study of three-centred solutions with biaxial symmetry (collinear centres)~\cite{breunholder_supersymmetric_2019}. 
To our knowledge, this is the first explicit construction of higher-dimensional black holes with $\mathbb{R}\times U(1)$ symmetry (that are smooth 
on and outside the horizon), which confirms Reall's conjecture for supersymmetric black holes.

The outline of this paper and of the proof of Theorem \ref{thm_summ} is as follows. In section \ref{ssec_susy_assumption}-\ref{ssec_axial_assumption} 
we review the local form of a supersymmetric solution, state our assumptions, and show that the  general solution must have a Gibbons-Hawking 
base. In section \ref{ssec_AF} we derive the constraints imposed by asymptotic flatness.  In section \ref{ssec_horizon} we combine the classification 
of near-horizon geometries~\cite{reall_higher_2004} with our assumptions, to prove that the horizon corresponds to a single point in the $\mathbb{R}^3$ 
cartesian coordinates of the Gibbons-Hawking base, and that the associated harmonic functions have at most simple poles at these points. In section 
\ref{ssec_orbit} we show that the $\mathbb{R}^3$ cartesian coordinates provide a global chart on the orbit space, and the associated harmonic functions 
have at most simple poles at fixed points of the $U(1)$ Killing field. This implies that the  general form of such solutions is of multi-centred type with 
simple poles in the harmonic functions, see Theorem \ref{thm_necessary}. In section \ref{sec_Sufficient} we perform a general regularity analysis of 
multi-centred solutions, in particular we derive  necessary and sufficient conditions for the solution to be smooth at a horizon and at a fixed point.
In section \ref{sec_Summary} we prove our main classification result which is stated in Theorem \ref{thm_classification}, give the asymptotic charges 
and show the symmetries of these spacetimes are generically $\mathbb{R}\times U(1)$. We conclude the paper with a discussion of the results in section 
\ref{sec_Discussion}.  A number of technical details are relegated to several Appendices. This includes a derivation of the general cohomogeneity-1 hyper-K\"ahler metric with triholomorphic euclidean $E(2)$ symmetry in Appendix \ref{app_symmetries}.

\section{Supersymmetric solutions with  axial symmetry}\label{sec_Necessary}

\subsection{Supersymmetric solutions and global assumptions}
\label{ssec_susy_assumption}

We will consider supersymmetric solutions to $D=5$ ungauged minimal supergravity. The bosonic field content of this theory consists of a spacetime metric $g$ and 
a Maxwell field $F$,  defined on a 5-dimensional spacetime manifold $\mathcal{M}$. The action is given by that of Einstein-Maxwell theory coupled to a Chern-Simons term for the 
Maxwell field. A solution is supersymmetric if it admits a supercovariantly constant spinor $\epsilon$ (Killing spinor). This condition is highly restrictive; Gauntlett et al. in \cite{gauntlett_all_2003}  derived the general {\it local} form of all supersymmetric solutions using Killing spinor bilinears. Let us now briefly summarise some of 
their results.

We will work in the conventions of~\cite{reall_higher_2004}, so in particular the metric signature is `mostly plus'.
From Killing spinor bilinears, one can construct a function $f$, a vector field $V$, and three 2-forms $X^{(i)}$, 
$i=1,2,3$. These satisfy certain algebraic identities, in particular, 
\begin{align}
	g(V,V) &=-f^2 \; ,\label{eq_fV}  \\
	\iota_V X^{(i)} &=0 \; , \label{eq_VX} \\
	X_{\mu \gamma}^{(i)} X^{(j) \gamma}_{\nu} &=  \delta_{ij} (f^2 g_{\mu\nu} + V_\mu V_\nu) - f \epsilon_{ijk} X^{(k)}_{\mu\nu}  \;  ,\label{eq_XX}
\end{align}
where $\epsilon_{ijk}$ is the alternating symbol with $\epsilon_{123}=1$\footnote{Greek indices $\mu, \nu, \dots$ denote spacetime indices and are raised and lowered with the spacetime metric $g_{\mu\nu}$}.
Importantly, (\ref{eq_fV}) shows that $V$ is causal everywhere. Furthermore, it was also shown that $V$ is non-vanishing on any region where the Killing spinor  is non-vanishing. These quantities  must also satisfy certain differential identities, in particular, $V$ is a Killing vector field on $(\mathcal{M}, g)$,  the 2-forms $X^{(i)}$ are closed and
\begin{equation}
\iota_V F= - \frac{\sqrt{3}}{2} \td f  \; . \label{eq_VF}
\end{equation}
This last relation implies that $\mathcal{L}_VF=0$, that is, the Maxwell field is preserved by the Killing vector field $V$.

On regions where $V$ is timelike, i.e. 
$f\neq 0$, the metric can be written as
\begin{equation}
	g= -f^2(\td t + \omega)^2+f^{-1}h \; , \label{eq_hdef}
\end{equation}
where  $V= \partial_t$, and $h$ is a Riemannian metric on a four-dimensional base space $B$ orthogonal to the orbits of $V$.  Note that the base metric $h$ can be invariantly defined on regions where $V$ is timelike by
\begin{equation}
h_{\mu\nu} = f \left( g_{\mu\nu} +  \frac{V_\mu V_\nu}{f^2}\right)  \; ,   \label{eq_hinv}
\end{equation}
 whereas the 1-form $\omega$ may be defined by $\iota_V \omega=0$ and $\td \omega=-\td ( f^{-2} V)$ (which fixes it up to a gradient) and hence can be regarded as a 1-form on $B$.   The constraints from supersymmetry imply that the base space $(B, h)$ is hyper-K\"ahler with complex structures given by the 2-forms $X^{(i)}$. In particular, (\ref{eq_VX})  implies that $X^{(i)}$ can be viewed as 2-forms on $B$ and (\ref{eq_XX}) implies that they obey the quaternion algebra on $(B,h)$\footnote{ Latin indices $a, b, \dots$ denote base space indices and are raised and lowered with the base metric $h_{ab}$.}
 \begin{equation}
 X^{(i)}_{ac} X^{(j) c}_{\phantom{(j) c} b} = - \delta_{ij} h_{ab} + \epsilon_{ijk} X^{(k)}_{ab} \label{eq_quaternion}  \; .
 \end{equation} 
 Furthermore, it can be shown that $X^{(i)}$ are parallel with respect to the Levi-Civita connection of $h$ and are anti-self dual with respect to the orientation $\eta$ on $B$ defined by the spacetime orientation $f(\td t+\omega)\wedge\eta$.  The Maxwell field can be written as 
\begin{equation}
	F = -\frac{\sqrt{3}}{2}\td \left(\frac{V}{f}\right) -\frac{1}{\sqrt{3}} G^+, \label{eq_Ftimelike}
\end{equation}
where $G^+$ is the self-dual part of $f\td\omega$ with respect to the base space metric $h$.

We now turn to our global assumptions.   
\begin{assumption}
\label{assumption1}
$(\mathcal{M}, g, F)$ is a solution of $D=5$ minimal supergravity such that:
\begin{enumerate}[label=(\roman*)]
\item the solution is supersymmetric in the sense that it admits a globally defined Killing spinor $\epsilon$, 
	\label{assumption_KS}

	\item the supersymmetric Killing field $V$ is complete, \label{ass_Vcomplete}
		\item the domain of outer communication (DOC), denoted by $\llangle \mathcal{M} \rrangle$, is globally hyperbolic,
		\label{assumption_DOC}
		\item $\llangle \mathcal{M} \rrangle$ is asymptotically flat, that is, it has an end diffeomorphic to $\mathbb{R}\times(\mathbb{R}^4\setminus B^4)$ where $B^4$ is a 4-ball, such that on this end, \label{assumption_AF}
		 \begin{enumerate}
	\item the metric $g=-\td u^0\td u^0 + \delta_{IJ}\td u^I\td u^J + \mathcal{O}(R^{-\tau})$ and some 
	decay rate $\tau>0$, where $u^0, (u^I)_{I=1}^4$ are the pull-back of the cartesian coordinates on $\mathbb{R}\times \mathbb{R}^4$,  
	$R :=\sqrt{u^Iu^J\delta_{IJ}}$,   and in these coordinates $\partial_\mu g_{\nu\rho} =\mathcal{O}(R^{-\tau-1})$, \label{assumption_falloff}
	\item the supersymmetric Killing field in these coordinates is $V=\partial/\partial u^0$, so we
	also refer to it as the stationary Killing field,  \label{ass_VAF}
\end{enumerate}
	\item each connected component of the event horizon $\mathcal{H}$ has a smooth cross-section $H_i$, that is, a 3-dimensional spacelike submanifold transverse to the orbits of $V$, which is compact, \label{assumption_H}
	\item   there exists a Cauchy surface $\Sigma$ that is the union of a compact set, an asymptotically flat end as in \ref{assumption_AF}, and a finite number of asymptotically cylindrical ends diffeomorphic to $\mathbb{R}\times H_i$ each corresponding to a connected component of the horizon,  \label{assumption_cauchy}
	\item the metric $g$ and the Maxwell field $F$ are smooth ($C^\infty$) on $\llangle \mathcal{M} \rrangle$ and at the horizon (if there is one). \label{assumption_smoothness}

\end{enumerate}
\end{assumption}
\noindent\textbf{Remarks.}
\begin{enumerate}
\item Assumption \ref{assumption1} \ref{assumption_KS} implies that the spinor bilinears $f$,  $V$ and $X^{(i)}$ introduced above are globally defined on spacetime.
	\item Under these assumptions it follows that $\llangle \mathcal{M}\rrangle$ simply connected, 
	by the topological censorship theorem \cite{friedman_topological_1995}.
	\item From global hyperbolicity and completeness of $V$ it follows that each integral curve of $V$ intersects a spacelike Cauchy surface $\Sigma$ 
exactly once. The flow-out from $\Sigma$ along $V$ is injective otherwise there would be closed causal curves in $\llangle \mathcal{M}\rrangle$, 
contradicting global hyperbolicity. Using the Flow-out Theorem (see e.g. \cite{lee_introduction_2012}) one can see that $\llangle \mathcal{M}\rrangle$ 
is diffeomorphic to $\mathbb{R}\times\Sigma$, and that the orbit space $\llangle \mathcal{M}\rrangle/\mathbb{R}_V$ is a manifold 
homeomorphic to $\Sigma$.   Furthermore, the base space $B$ can be identified with the open subset of the orbit space $\llangle \mathcal{M}\rrangle/\mathbb{R}_V$ corresponding to timelike orbits of $V$ and the base metric $h$ is the corresponding orbit space metric~\cite{geroch_method_1971}.  \label{remark_orbit}
	\item The assumption  that $V$ is timelike in the asymptotic region (Assumption \ref{assumption1} \ref{assumption_AF}) implies that the metric can be written as timelike fibration over a hyper-K\"ahler base space (\ref{eq_hdef}), at least in the asymptotic region.  We emphasise that we do not assume that $V$ is strictly  timelike everywhere in $\llangle \mathcal{M} \rrangle$ and therefore the hyper-K\"ahler structure is not globally defined.  Notably, the Killing field $V$ must be tangent to the event horizon and therefore tangent to the null generators of the horizon.
\item The horizon corresponds to asymptotically cylindrical ends due to the well-known fact the horizon of a supersymmetric black hole must be extremal (see discussion around equation (\ref{eq_AC}) for the argument in our context).  By Assumption \ref{assumption1} \ref{assumption_cauchy} we can compactify these ends of $\Sigma$ by adding boundaries diffeomorphic to $H_i$. Thus, we can also view $\Sigma$ as an asymptotically flat 4-manifold with boundaries $H_i$ corresponding to each connected component of the horizon.
	
\end{enumerate}

\subsection{Including axial symmetry}
\label{ssec_axial_assumption}

Let us now turn to the general analysis of a supersymmetric solution admitting a compatible axial symmetry.  In particular, we will make the following assumptions.

\begin{assumption}\label{assumption2}
The supersymmetric background $(\mathcal{M}, g, F, \epsilon)$  admits a globally defined spacelike Killing field $W$ such that:
\begin{enumerate}[label=(\roman*)]
\item its flow has periodic orbits, that is, it is an `axial' Killing field in the sense it generates a $U(1)$ isometry \label{assumption_Waxial}
	\item it preserves the Maxwell field, $\mathcal{L}_W F=0$ \label{assumption_WMaxwell}
	\item it preserves the Killing spinor $\mathcal{L}_W\epsilon=0$,  that is, commutes with the remaining supersymmetry \label{assumption_Wspinor}
	\item at each point of $\llangle \mathcal{M} \rrangle$ there exists a linear combination of $V$ and $W$ which is timelike.\label{assumption_span}
%	\item the span of the Killing fields $V, W$ is timelike at each point of $\llangle \mathcal{M} \rrangle$.\label{assumption_span}
\end{enumerate}
\end{assumption}

Assumption \ref{assumption2} \ref{assumption_Wspinor} is a supersymmetric generalisation of the usual requirement that the axial Killing field commutes with 
the stationary Killing field. This is revealed by the following lemma, which also greatly restricts the possible base space geometries.
\filbreak
\begin{lemma}
	Under Assumption \ref{assumption2} \ref{assumption_Wspinor} the following hold:
	\begin{enumerate}[label=(\alph*)]
		\item the Killing spinor bilinears $(f, V, X^{(i)})$ are preserved by the axial Killing field $W$; in particular, the supersymmetric and axial Killing fields commute, $[V, W]=0$,
		\item  wherever $V$ is timelike, the data on the base $(f, h, X^{(i)})$ and $\omega$ (in an appropriate gauge) are preserved by the axial Killing field $W$; in particular, $W$ defines a triholomorphic Killing field of the hyper-K\"ahler structure $(B, h, X^{(i)})$ which must therefore be of Gibbons-Hawking form. 
	\end{enumerate}
	\label{Lemma_GH}
\end{lemma}
\begin{proof}
	By Assumption \ref{assumption2} \ref{assumption_Wspinor} $W$ preserves the Killing spinor, so by the Leibniz rule, it also preserves Killing spinor bilinears, i.e. $W(f)=0$, $[W, V]=0$, 
	$\mathcal{L}_WX^{(i)}=0$. 
	
	Next, the base metric $h$ is invariantly defined when $f\neq 0$ by (\ref{eq_hinv}), and as $g$, $f$ and $V$ are preserved by both 
	Killing fields, so is $h$. 
	Now, $\omega$ is only defined up to a gauge transformation $\omega\to\omega+\td\lambda$ generated by $t\to t-\lambda$ where $\lambda$ is a function on $B$. 
	We may partially fix this gauge by requiring $\mathcal{L}_W t=0$, so that $\omega$ is also invariant under $W$. In this gauge it is manifest that 
	$W$ can be regarded as a vector field on the base.\footnote{
		Evidently, $W$ being well-defined on the base is not gauge-dependent. To be more precise, we can define 
		$\widetilde{W}:= \pi_*W$ on $B$, where $\pi: \mathcal{M}\to B$ is the quotient map by $V$. This is 
		a projection, in a coordinate basis $(W^t, W^a)\mapsto(W^a)$, and it is well-defined since $\mathcal{L}_V W=0$. 
		Then $\mathcal{L}_Wh=0$ on the spacetime implies $\mathcal{L}_{\widetilde{W}}h=0$ on $B$, and similarly for 
		other $U(1)$-invariant tensors well-defined on the base. In the following, we will not distinguish between 
		$W$ and $\widetilde{W}$, except for Appendix \ref{app_KS}.
	}  Since $X^{(i)}$ are the complex structures of $B$, and they are preserved by the Killing field 
	$W$ (i.e. it is triholomorphic),  the metric $h$ can always be written in Gibbons–Hawking form \cite{gibbons_hidden_1988}.

\end{proof}

\noindent\textbf{Remarks.}
\begin{enumerate}
	\item The following converse of Lemma \ref{Lemma_GH} is also true: a Killing field that is triholomorphic on the base and commutes with the stationary 
	Killing field preserves the Killing spinor. This is shown in Appendix \ref{app_KS}.
	\item Lemma   \ref{Lemma_GH} shows that Assumption \ref{assumption2} \ref{assumption_WMaxwell} is redundant, because, when $f\neq0$, $F$ 
	can be expressed in terms of $U(1)$-invariant quantities as in (\ref{eq_Ftimelike}). We will show that the region $f\neq0$ is dense in 
	$\llangle \mathcal{M} \rrangle$ (see Corollary \ref{cor_fdense}), therefore, using continuity of $\mathcal{L}_W F$, $F$ is preserved by $W$ on and outside the horizon. 

\end{enumerate} 

We have established that on regions where $f\neq 0$, the base metric $h$ has Gibbons-Hawking form. Recall the local form of this is
\begin{equation}
	h = \frac{1}{H}(\td \psi + \chi)^2 + H\td x^i\td x^i \; , \label{eq_GHmetric}
\end{equation}
where $x^i$, $i=1,2,3$ are cartesian coordinates on $\mathbb{R}^3$, $H$ and $\chi$ is a harmonic function and a 1-form on $\mathbb{R}^3$, respectively, 
satisfying
\begin{equation}
	\star_3 \td \chi = \td H \; , \label{eq_chieq}
\end{equation}
where $\star_3$ denotes the Hodge star operator on $\mathbb{R}^3$ with respect to the euclidean metric. In these coordinates, the complex structures are
\begin{equation}
	X^{(i)} = (\td \psi +\chi)\wedge \td x^i  -\frac{1}{2} H \epsilon_{ijk}\td x^j \wedge \td x^k \; , \label{eq_complexforms}
\end{equation}
and 
the triholomorphic Killing field is $W= \partial_\psi$.

Remarkably, it has been shown~\cite{gauntlett_all_2003}  that if the triholomorphic Killing field $W $ of the base is a Killing field of the five-dimensional metric, as is the case for us, then 
the general solution is completely determined by four harmonic functions, $H$, $K$, $L$, $M$ on 
$\mathbb{R}^3$ as follows. Let $\omega_\psi$ be a function and $\hat\omega$ and $\xi$ be 1-forms on $\mathbb{R}^3$ satisfying
\begin{align}
	\omega_\psi &= \frac{K^3}{H^2}+\frac{3}{2}\frac{KL}{H}+M \; , \label{eq_omegapsi} \\
	\star_3 \td\hat\omega &= H \td M - M \td H +\frac{3}{2} (K \td L-L\td K) \; , \label{eq_omegahat} \\
	\star_3 \td\xi &= -\td K \; . \label{eq_xi}
\end{align}
Then $f$ and $\omega$ can be written as
\begin{align}
	f &= \frac{H}{K^2+HL} \; , \label{eq_f}\\
	\omega &= \omega_\psi(\td \psi+\chi) + \hat\omega \; ,\label{eq_omega}
\end{align}
while the Maxwell field takes the form
\begin{equation}
	F = \td A = \frac{\sqrt{3}}{2} \td \left(f(\td t + \omega)-\frac{K}{H}(\td \psi + \chi)-\xi\right)\;.\label{eq_Maxwell}
\end{equation}
We emphasise that at this stage, the local form of the solution is now fully determined, up to four harmonic functions on $\mathbb{R}^3$.

We will now introduce several spacetime invariants that are useful for our global analysis, following \cite{breunholder_moduli_2019}. Invariance of the Maxwell field under the Killing fields $V, W$ allows us to introduce an electric and magnetic potential $\Phi$, $\Psi$ satisfying 
\begin{align}
	\frac{\sqrt{3}}{2}\td \Phi &:= \iota_V F \label{eq_electricpot} \; ,\\
	\frac{\sqrt{3}}{2}\td \Psi &:= \iota_W F \label{eq_magneticpot} \; ,
\end{align}
which are
globally defined functions (up to an additive constant) on the DOC since  it is simply connected.  In fact, by (\ref{eq_VF})  we must have $\Phi = - f +\text{const}$ so the electric potential for supersymmetric solutions does not give an independent invariant. These potentials are preserved by the Killing fields.  Indeed, using (\ref{eq_electricpot}-\ref{eq_magneticpot}), 
	the electric and magnetic potentials satisfy $\mathcal{L}_V\Phi \propto \iota_V\iota_VF=0$, $\mathcal{L}_W\Psi \propto \iota_W\iota_WF=0$, 
	$\mathcal{L}_V\Psi=-\mathcal{L}_W\Phi =\iota_W \td f =0$ where in the final step we used the relation between $\Phi$ and $f$.

It turns out that a key spacetime invariant is given by the 
determinant of the inner product matrix of Killing fields:
\begin{equation}
	N :=-\begin{vmatrix}
		g(V, V) & g(V, W)\\
		g(W, V) & g(W,W)
	\end{vmatrix} = \frac{f}{H} = \frac{1}{K^2+HL} \; , \label{eq_Ndef}
\end{equation}
where the last two equalities are valid for solutions with a Gibbons-Hawking base as above.  By Lemma \ref{Lemma_GH},  $N$ is preserved by both Killing fields, because it is defined in terms of $V$, $W$ and $g$.  The significance of $N$ is that  Assumption \ref{assumption2} \ref{assumption_span} implies that $N>0$ everywhere on $\llangle \mathcal{M} \rrangle \setminus \mathcal{F}$ where
\begin{equation}
\mathcal{F}:=   \{p\in \llangle \mathcal{M} \rrangle \;|  \; W_p= 0 \}   \label{eq_fixedpts}
\end{equation}
denotes the set of fixed points of $W$ in the DOC (see proof of Lemma \ref{lemma_harmonicwelldefined} below).

 Let us now recall a useful result proven in~\cite[Lemma 1]{breunholder_moduli_2019}.

\begin{lemma}
	A supersymmetric solution $(\mathcal{M}, g,F)$ with a Gibbons-Hawking base is smooth on the region $N>0$ if and only if the associated harmonic functions $H$, $K$, $L$, $M$ are smooth and $K^2+HL>0$.
	\label{lemma1}
\end{lemma}

\noindent In particular, if $N>0$, the harmonic functions can be expressed in terms of spacetime invariants:
\begin{align}
	H  = \frac{f}{N} \; , \qquad \qquad &\qquad \qquad L=\frac{fg(W, W)+ 2g(V, W)\Psi-f\Psi^2}{N} \; , \nonumber\\
	K = \frac{f\Psi - g(V, W)}{N} \; , \quad   &M = \frac{g(W,W)g(V, W)-3f\Psi g(W,W)-3\Psi^2g(V, W)+f\Psi^3}{2N} \; . \label{eq_harmonicinvariant}
\end{align}
This shows that the harmonic functions $K$, $L$, $M$ are only defined up to a gauge transformation $\Psi\to\Psi+c$ where $c$ is a constant.  This allows us to deduce the following important result.

\begin{lemma}\label{lemma_harmonicwelldefined}
	 The harmonic functions $H, K, M, L$ are well-defined and smooth on  
		$\llangle \mathcal{M} \rrangle\setminus \mathcal{F}$, and they are preserved by the two Killing fields $V$, $W$. 
\end{lemma}

\begin{proof}
	Assumption \ref{assumption2} \ref{assumption_span} implies that if $f(p)=0$ at some $p\in \llangle \mathcal{M} \rrangle$ then $W\cdot V\neq 0$ at $p$, otherwise 
	there would not exist a timelike linear combination of the Killing fields at $p$. It follows that the zeros of the invariant $N=f^2|W|^2 +(V\cdot W)^2$ in 
	$\llangle \mathcal{M} \rrangle$ coincide with the zeros of $W$. Hence, for any $p \in  \llangle \mathcal{M} \rrangle$ such that $W_p\neq 0$ we have $N>0$. Therefore, by Lemma \ref{lemma1}, the associated harmonic functions are  well-defined and smooth
	at every  point in  $\llangle \mathcal{M} \rrangle$ that is not a fixed point of $W$. The harmonic functions can be expressed in terms of invariants as (\ref{eq_harmonicinvariant}), and since these invariants are preserved by both Killing fields, so are the harmonic functions. 
\end{proof}

\noindent\textbf{Remark.}
From Assumption \ref{assumption2} \ref{assumption_span} it follows that $f\neq 0$ at a fixed point of $W$. Therefore, wherever $f=0$ in  $\llangle \mathcal{M} \rrangle$, we must have $N>0$, so the associated harmonic functions are 
	well-defined, even though the base is not. Therefore, from (\ref{eq_harmonicinvariant}) it can be seen 
	that the zeros of $f$ and $H$ must coincide in  $\llangle\mathcal{M} \rrangle$. These so-called `evanescent ergosurfaces' have 
	been analysed in great detail in \cite{niehoff_evanescent_2016}. These are smooth, timelike hypersurfaces outside the horizon on which the stationary Killing field 
	becomes null. This way black hole solutions can evade the uniqueness theorem of \cite{reall_higher_2004}, which assumed that the stationary Killing field is strictly timelike outside the horizon. \\

From (\ref{eq_complexforms}), it is immediate that the cartesian coordinates satisfy
\begin{equation}
	\td x^i=\iota_WX^{(i)} \; . \label{eq_xdef}
\end{equation}
The right-hand side $\iota_WX^{(i)}$ is a globally defined 1-form on spacetime, and closed everywhere as a consequence of $\mathcal{L}_WX^{(i)}=0$, and the fact that $X^{(i)}$ are 
closed \cite{gauntlett_all_2003}. Since $\llangle \mathcal{M} \rrangle$ is simply connected, equation 
(\ref{eq_xdef}) allows us to introduce globally defined functions $x^i$  on  $\llangle \mathcal{M} \rrangle$ (up to an additive constant) that coincide with the local cartesian coordinates of the Gibbons-Hawking base. From (\ref{eq_VX}) it follows that the functions $x^i$ are preserved by both $V$ and $W$.   In the following sections, we will determine the behaviour of $x^i$ in the asymptotically flat region and near the horizon.

\subsection{Asymptotic flatness} \label{ssec_AF}

In this section we use asymptotic flatness as stated in our definition (Assumption \ref{assumption1} \ref{assumption_AF}) to deduce the behaviour of the base metric $h$, 
the complex 2-forms $X^{(i)}$, and the $\mathbb{R}^3$  cartesian coordinates $x^i$ of the Gibbons-Hawking base, near spatial infinity.

First observe by Assumption \ref{assumption1} \ref{ass_VAF}  we may identify $u^I, I=1, 2,3,4$ as coordinates on the base space $B$.   Asymptotic flatness (Assumption \ref{assumption1} \ref{assumption_AF}) then implies that in the asymptotic end, 
\begin{align}
f &= 1+ O(R^{-\tau}) \; ,\\
\omega &=  O(R^{-\tau}) \td u^I \; ,\\
h &= (\delta_{IJ}+\mathcal{O}(R^{-\tau}))\td u^I \td u^J, \label{eq_hAE}
\end{align}
so, in particular, $(B, h)$ has an asymptotically euclidean end diffeomorphic to $\mathbb{R}^4\setminus B^4$.   The next result constrains the behaviour of the hyper-K\"ahler structure near infinity.

\begin{lemma}
On the asymptotically flat end the complex structures of $(B, h)$ can be written in cartesian coordinates  as
\begin{equation}
	X^{(i)} = \Omega_-^{(i)} + \mathcal{O}(R^{-\tau}) \; , \label{eq_XAF}
\end{equation}
where $\Omega_-^{(i)}$ are a standard basis of anti-self-dual 2-forms on $\mathbb{R}^4$,
\begin{align}
	\Omega_-^{(1)} &= \td u^1 \wedge \td u^4 + \td u^3 \wedge \td u^2 \; ,\nonumber \\
	\Omega_-^{(2)} &= \td u^1 \wedge \td u^3 + \td u^2 \wedge \td u^4 \; , \\
	\Omega_-^{(3)} &= \td u^1\wedge \td u^2 + \td u^4\wedge \td u^3 \; . \nonumber
\end{align} 
\end{lemma}
\begin{proof}
The 2-forms $X^{(i)}$ satisfy the quaternion algebra (\ref{eq_quaternion}), which in particular implies that $X^{(i)}_{ab}X^{(i) ab}= -4$ (no sum over $i$). 
Hence, from (\ref{eq_hAE}) we immediately deduce that in the cartesian coordinates on the asymptotic end $X^{(i)}_{IJ}= O(1)$.  Next, we use the fact 
that $X^{(i)}$ are parallel with respect to the Levi-Civita connection $\nabla^{(h)}$ defined by $h$.  By Assumption \ref{assumption1} \ref{assumption_falloff} the 
derivatives of the metric in cartesian coordinates are $\partial_I h_{JK} = O(R^{-\tau-1})$  and therefore the covariant derivative 
$\nabla_I^{(h)} X^{(i)}_{JK} = \partial_I X^{(i)}_{JK} + O(R^{-\tau-1})$, so we deduce that $X^{(i)}_{IJ} = \bar{X}_{IJ}+ O(R^{-\tau})$ 
where $\bar{X}_{IJ}$ are constants.  

Finally, we use that $X^{(i)}$ is ASD with respect to the base metric $h$.   To this end, let us decompose $\bar{X}^{(i)}= \bar{X}^{(i)}_+ + \bar{X}^{(i)}_-$  where $\bar{X}_\pm^{(i)}$ are the SD/ASD parts with respect to the euclidean metric on the asymptotic end, that is, $\star_\delta \bar{X}^{(i)}_\pm = \pm \bar{X}^{(i)}_\pm$.   Then
\begin{equation}
\star_h X^{(i)} = \star_\delta X^{(i)} + O(R^{-\tau}) = \star_\delta \bar{X}^{(i)} + O(R^{-\tau}) = \bar{X}_+^{(i)}- \bar{X}_-^{(i)} + O(R^{-\tau}) \; ,
\end{equation}
where in the first equality we have used (\ref{eq_hAE}) to write $\star_h$ in terms of $\star_\delta$ and $O(R^{-\tau})$ terms and $X_{IJ}^{(i)}=O(1)$. Hence, $\star_h X^{(i)}= - X^{(i)}$ implies that $\bar{X}^{(i)}_{+ IJ}=O(R^{-\tau})$ and therefore since they are constants $\bar{X}_{+ IJ}^{(i)}=0$.  We have therefore shown that 
\begin{equation}
X^{(i)}= \bar{X}_-^{(i)}+ O(R^{-\tau})   \; ,\label{eq_XX-}
\end{equation}
where $\bar{X}^{(i)}_{IJ}$ are constant components of  ASD 2-forms on $(\mathbb{R}^4, \delta)$.  Thus, the quaternion algebra (\ref{eq_quaternion}) implies that $\bar{X}_-^{(i)}$ obey the quaternion algebra with respect to the euclidean metric $\delta$. Therefore, by performing a constant $SO(3)$ rotation on $X^{(i)}$, we may always set $\bar{X}_-^{(i)}= \Omega^{(i)}_-$,  where $\Omega^{(i)}$ are a basis of ASD 2-forms on $\mathbb{R}^4$ given in the lemma.
\end{proof}

We have shown that asymptotic flatness implies that the base space is asymptotically euclidean and the hyper-K\"ahler structure is asymptotically that of euclidean space.   However, we also know that the hyper-K\"ahler structure is of Gibbons-Hawking form. Combining these facts we deduce the following.

\begin{lemma}\label{lem_AF}
The triholomorphic Killing field can be written as
\begin{equation}
W=\tfrac{1}{2}(J_{12}+J_{34})+ O(R^{1-\tau}) \; ,
\end{equation}
where $J_{IJ} := u^J\partial_I - u^I\partial_J$. The corresponding $\mathbb{R}^3$-cartesian coordinates $x^i$ defined by (\ref{eq_xdef}) are 
\begin{align}
	x^1  =& \frac{1}{2}(u^1 u^3 + u^2 u^4) + \mathcal{O}(R^{2-\tau}) \; ,\qquad x^2  = \frac{1}{2}(u^2 u^3 - u^1 u^4) + \mathcal{O}(R^{2-\tau}) \; ,\nonumber\\
	&x^3 = \frac{1}{4} \left((u^1)^2+(u^2)^2-(u^3)^2-(u^4)^2\right)+ \mathcal{O}(R^{2-\tau}) \; .  \label{R4R3}
\end{align}
In particular, $r = \sqrt{x^ix^i}=R^2/4 + \mathcal{O}(R^{2-\tau})$ and so $r\to \infty$ in the asymptotically flat end.  
\end{lemma}

\begin{proof}
The axial Killing field $W$ can be written as an $\mathbb{R}$-linear 
combination of rotational Killing fields of $\mathbb{R}^4$ up to $\mathcal{O}(R^{-\tau})$ corrections \cite{chrusciel_killing_2005}.  Therefore, without loss of generality (by rotating 
$u^I$ coordinates if necessary) we can write 
\begin{equation}
	W = \alpha (J_{12}+J_{34}) + \beta (J_{12}-J_{34}) + \mathcal{O}(R^{1-\tau}) \; ,
\end{equation}
where $\alpha,\beta$ constants. One can check that $J_{12}+J_{34}$ preserves 
$\Omega_-^{(i)}$, while $J_{12}-J_{34}$ rotates $\Omega_-^{(1)}$ and $\Omega_-^{(2)}$ into the one another. Hence, from (\ref{eq_XAF}),
$W$ is triholomorphic if and only if $\beta=0$, and requiring $4\pi$-periodicity of orbits of $W$ fixes $\alpha = 1/2$. 

For the cartesian coordinates of the Gibbons-Hawking metric we have
\begin{equation}
	\td x^i = \iota_W X^{(i)} = \frac{1}{2}\iota_{J_{12}}\Omega_-^{(i)} + \frac{1}{2}\iota_{J_{34}}\Omega_-^{(i)} + \mathcal{O}(R^{1-\tau}) \; ,
\end{equation}
and a straightforward computation yields (\ref{R4R3}) upon integration.
\end{proof}

From this we deduce a number of important corollaries.

\begin{corollary}\label{cor_AFGH} 
The asymptotic end of $(B, h)$ is diffeomorphic to an $S^1$-fibration over $\mathbb{R}^3 \backslash B^3$ where $S^1$ is the orbits of $W$ and $B^3$ is a 3-ball in $\mathbb{R}^3$.
\end{corollary}

\begin{corollary}\label{cor_AF_H}
On the asymptotic end the harmonic function associated to the Gibbons-Hawking base takes the form
\begin{equation}
	H = \frac{1}{r}+ \sum_{l\ge1, m} h_{lm}r^{-l-1}Y^m_l \; ,   \label{eq_AF_H}
\end{equation}
where $Y^m_l$ are the spherical harmonics and $h_{lm}$ are constants.
\end{corollary}

\begin{proof}
For the second corollary observe that the harmonic function of the Gibbons-Hawking metric (\ref{eq_GHmetric}) can be written invariantly on the base as $H=  h(W, W)^{-1}$ where recall in these coordinates $W=\partial_\psi$.  On the other hand, from Lemma \ref{lem_AF} we can compute the norm of $W$ and find $h(W, W)= \tfrac{1}{4} R^2(1+ O(R^{-\tau})) = r (1+O(r^{-\tau/2}))$, which establishes the claimed leading term. The form of the subleading terms follows from the fact $H$ is harmonic on $\mathbb{R}^3 \backslash B^3$.
\end{proof}

\noindent {\bf Remark.} Our Assumption \ref{assumption1} \ref{assumption_AF} implies  that the invariants $g(V,W)=-f^2 \iota_W \omega=  O(R^{1-\tau})$ and $g(W, W)= f^{-1} h(W,W)- f^2 (\iota_W\omega)^2= \tfrac{1}{4}R^2(1+ O(R^{-\tau}))$ and so in particular $N= \tfrac{1}{4}R^2(1+ O(R^{-\tau}))$.  Furthermore, it also implies the Maxwell field (\ref{eq_Ftimelike}) $F_{\mu\nu}= O(R^{-\tau-1})$ in cartesian coordinates.  From (\ref{eq_magneticpot}),  it follows that the magnetic potential is $\Psi=\Psi_0+ O(R^{1-\tau})$ where $\Psi_0$ is a constant.  Using (\ref{eq_harmonicinvariant}) a computation reveals that $\lim_{R\to\infty}K=0$,  $\lim_{R\to\infty}L=1$ and $M = -\frac 32 \Psi_0+ O(R^{1-\tau})=-\frac 32 \Psi_0+O(r^{1/2-\tau/2})$. Harmonicity on $\mathbb{R}^3$ then implies
\begin{equation}
L= 1+ O(r^{-1}) \; , \qquad K= O(r^{-1}) \; , \qquad M =m+ O(r^{-1}) \; ,  \label{eq_AF_KLM}
\end{equation}
for some constant $m$.\\

The leading term in $H=1/r + \mathcal{O}(r^{-2})$ gives euclidean space. To see this, first we integrate for the 1-form $\chi$ using (\ref{eq_chieq}) 
which gives, up to a gauge transformation,
\begin{equation}
\chi= ( \tilde \chi_0+ \cos \theta ) \td \phi + \mathcal{O}(r^{-1}) \; ,   \label{eq_chiflat}
\end{equation}
where $(r, \theta, \phi)$ are spherical polar coordinates on $\mathbb{R}^3$ and $\tilde\chi_0$ is an integration constant. Under a coordinate change 
$(\psi, \phi)\to (\psi + c\phi, \phi)$, the constant $\tilde\chi_0\to\tilde\chi_0-c$ so we can fix it to any value that we like.  It turns  out that a convenient choice, which we will make, is to  fix $\tilde\chi_0$ to be an odd integer.
In particular, $\tilde\chi_0=\pm 1$ removes the Dirac string singularity in $\chi$ on the lower (upper) half $z$-axis on $\mathbb{R}^3$.

To determine the identification lattice of the angular directions let us define new coordinates
\begin{equation}
\tilde\psi= \psi+ \tilde\chi_0 \phi \; , \qquad \tilde{\phi}= \phi \; ,   \label{eq_eulerinfinity}
\end{equation}
and $r= \tfrac{1}{4}R^2$. Then the Gibbons-Hawking metric to leading order in $R$ becomes
\begin{equation}
h_0= \td R^2+ \tfrac{1}{4} R^2 ( ( \td\tilde \psi+ \cos\theta \td \tilde\phi)^2+ \td \theta^2+ \sin^2\theta \td \tilde \phi^2) \; .
\end{equation}
This is isometric to $\mathbb{R}^4$ in spherical coordinates where the radial coordinate of $\mathbb{R}^4$ is given by 
$R$   and $(\theta, \tilde\phi, \tilde \psi)$ are Euler-angles of $S^3$ with their identification lattice generated by
\begin{align}
	P:\quad (\tilde\psi,\tilde \phi)\sim(\tilde \psi+4\pi, \tilde\phi) \; , && R:\quad(\tilde\psi, \tilde\phi)\sim(\tilde\psi +2\pi, \tilde \phi+2\pi) \; .
	\label{eq_lattice}	
\end{align}
In the original coordinates, with $\tilde{\chi}_0$ an odd integer, this identification is equivalent to 
\begin{align}
	P:\quad (\psi,\phi)\sim(\psi+4\pi,\phi) \; , && R:\quad(\psi,\phi)\sim(\psi, \phi+2\pi) \; ,\label{eq_lattice_original}
\end{align}
that is, the angles are independently periodic so $\phi$ can be thought of as the standard azimuthal angle on $\mathbb{R}^3$. In these coordinates 
it is thus manifest that the asymptotic end is diffeomorphic to an $S^1$-fibration over $\mathbb{R}^3\backslash B^3$ where $\psi$ is a coordinate on $S^1$ and $x^i$ are coordinates on $\mathbb{R}^3$.

\subsection{Near-horizon geometry}\label{ssec_horizon}

By Lemma \ref{lemma_harmonicwelldefined} the harmonic functions are smooth at generic points of $\llangle \mathcal{M} \rrangle$, so the only potentially singular behaviour 
occurs at the horizon and fixed points of $W$.  In this section we will investigate the constraints imposed by a smooth black hole horizon on a supersymmetric solution with axial symmetry satisfying our above assumptions. In particular, we will deduce that connected components of the horizon correspond to points in the $\mathbb{R}^3$ coordinates of the Gibbons-Hawking base and that the harmonic functions have at most simple poles at these points. We will heavily rely on the known classification of near-horizon geometries~\cite{reall_higher_2004} which uses Gaussian null coordinates adapted to the horizon. In particular, our strategy will be to impose that the spacetime near the horizon admits an axial Killing field $W$ that preserves the 2-forms $X^{(i)}$, and then deduce the cartesian coordinates $x^i$ of the Gibbons-Hawking base in terms of Gaussian null coordinates using (\ref{eq_xdef}).

The supersymmetric Killing field $V$ must be tangent to the event horizon $\mathcal{H}$ and therefore (\ref{eq_fV}) implies it must be null on the horizon and thus tangent to its null generators. Furthermore, (\ref{eq_fV}) implies the horizon is extremal, that is, $\td (g(V, V))=0$ at the horizon.  Thus, $\mathcal{H}$ is an extremal Killing horizon of $V$. Next, by our Assumption \ref{assumption1} \ref{assumption_H} each component of the horizon $\mathcal{H}$ has a smooth cross-section $H$ transverse to the orbits of $V$. Let $(y^A)$ be coordinates on $H$.
Then, following~\cite{reall_higher_2004}, in a neighbourhood of the horizon $\mathcal{H}$ we may introduce Gaussian null coordinates $(v, \lambda, y^A)$ adapted so $V=\partial_v$ where $U=\partial_\lambda$ 
is tangent to affine null geodesics transverse to the horizon that are normalised so $g(V, U)=1$ and synchronised so the horizon is at $\lambda = 0$ (note $U$ is past-directed and $\lambda>0$ is outside the horizon). It can be shown that in such a neighbourhood of the horizon the metric  takes the form
\begin{equation}
	g =  -\lambda^2 \Delta^2 \td v^2 + 2 \td v \td\lambda +2\lambda h_A\td v \td y^A +\gamma_{AB}\td y^A\td y^B \; , \label{eq_gaussianNullCoords}
\end{equation}
where  the metric components are smooth functions of $(\lambda, y^A)$ for sufficiently small  $\lambda$.  The quantities $\Delta, h_A, \gamma_{AB}$ define components of a function, 1-form and Riemannian metric on the 3d surfaces of constant $(v, \lambda)$ which include the horizon cross-sections $H$ ($v=\text{constant}, \lambda=0$).   We will denote any quantity evaluated at $\lambda=0$ by $\mathring \Delta := \Delta|_{\lambda=0}$, $\mathring h_A:= h_A|_{\lambda=0}$, $\mathring\gamma_{AB}:= \gamma_{AB}|_{\lambda=0}$ etc, so in particular $(H, \mathring \gamma)$ is a 3d Riemannian manifold.

As is typical for spacetimes containing extremal horizons,  the horizon corresponds to an asymptotically cylindrical end of the space orthogonal to the orbits of $V$.  We recall the argument for this in the present context.
	The orbit space metric is $\hat{g}_{\mu\nu}:= g_{\mu\nu}- \frac{V_\mu V_\nu}{g(V, V)}$ in Gaussian null coordinates is 
	\begin{equation}
		\hat{g}= \frac{(\td \lambda + \lambda h_A\td y^A)^2}{\Delta^2\lambda^2}+\gamma_{AB}\td y^A\td y^B \; .   \label{eq_AC}
	\end{equation}
	When we approach the horizon on a geodesic $y^A=\const$, the proper distance $\int ^0 \td \lambda/(\lambda \Delta)$ diverges at 
	least logarithmically, so it takes infinite proper distance to reach the horizon in the orbit space.  Thus, the horizon corresponds to an asymptotic end diffeomorphic  to $\mathbb{R}\times H$.

Now we consider the axial Killing field $W$ near the horizon.  It must be tangent to the event horizon $\mathcal{H}$ and therefore in Gaussian null coordinates $W^\lambda=0$ at $\lambda=0$. Furthermore, $[W, V]=0$ implies all components of $W$ are $v$-independent. Then, by evaluating $\mathcal{L}_W g=0$ on the horizon it follows that $\mathcal{L}_{\mathring W} \mathring\gamma=0$ where $\mathring W:= \mathring W^A \partial_{y^A}$, that is, $\mathring W$ is a Killing field of $(H, \mathring \gamma)$.   Therefore, by smoothness of $W$ at the horizon, we may write $W$ in Gaussian null coordinates as
\begin{equation}
W= \mathring W^A \partial_{y^A} + \lambda \tilde{W}_\lambda \partial_\lambda + O(\lambda) \partial_{y^A}  + W^v \partial_v \; ,   \label{eq_WnearH}
\end{equation}
where $\tilde W_\lambda$ is a smooth function at the horizon.
In fact, the $v$-component of $W$ does not feature in the subsequent calculations, and therefore we do not need its detailed form here. Note that $W$ cannot vanish identically on the horizon, because if it did, it then follows $W$ vanishes everywhere. This is because, if $W=0$ on $\mathcal{H}$ then all tangential derivatives of $W$ 
	also vanish on $\mathcal{H}$, hence by  Killing's equation all first derivatives of $W$ must vanish on $\mathcal{H}$, which implies that $W$ vanishes everywhere.

We now prove one of the main results of this section which relates the Gibbons-Hawking coordinates to Gaussian null coordinates.
\begin{lemma}
	The cartesian coordinates $x^i$ defined by (\ref{eq_xdef}) are constant on a connected component of the horizon. Furthermore, the euclidean distance from a connected component of the horizon $x^i=a^i$,  is
	$r:=| \bm x - \bm a|= c\lambda + \mathcal{O}(\lambda^2)$ for some positive constant $c$.
	\label{lemma_GNC_cart}
\end{lemma}

\begin{proof}
In the following we will use the notation 
and results of \cite{reall_higher_2004}.  In the neighbourhood of $\mathcal{H}$ the hyper-K\"ahler 2-forms can be written in Gaussian null coordinates as 
\begin{equation}
	X^{(i)}= \td\lambda \wedge Z^{(i)} + \lambda(h\wedge Z^{(i)}-\Delta \star_3Z^{(i)}) \; , \label{eq_complexstr}
\end{equation}
where $h = h_A\td y^A$, $\star_3$ is the Hodge-dual with respect to $\gamma=\gamma_{AB}\td y^A \td y^B$, and $Z^{(i)}=Z^{(i)}_A\td y^A$ satisfy
\begin{align}
	\star_3Z^{(i)}=\frac{1}{2}\epsilon_{ijk}Z^{(j)}\wedge Z^{(k)} \;, &&Z^{(i)}\cdot Z^{(j)}=\delta_{ij}   \; ,
\end{align}
where the inner product $\cdot$ is defined by $\gamma$.
Since $V$ is Killing and preserves $X^{(i)}$,  the $Z^{(i)}$ are all preserved by its flow. Furthermore, in~\cite{reall_higher_2004} it was shown that compactness of $H$ (Assumption \ref{assumption1} \ref{assumption_H}) implies that on each connected component of the horizon $\mathring\Delta$ is constant,  $\mathring h^A$ is a Killing vector of $(H, \mathring \gamma)$,  $\mathring h^2:=\mathring h^A \mathring h_A$ is constant,  $\mathcal{L}_{\mathring h} \mathring Z^{(i)}=0$,\footnote{
	Throughout this section indices $A, B...$ of $Z^{(i)}, h$ etc are raised and lowered with  $\gamma_{AB}$ (and its inverse).
} and in the neighbourhood of $\mathcal{H}$ \footnote{
	Compared to \cite{reall_higher_2004}, here we do not restrict these equations to the horizon. 
} 
\begin{align}
	\td h = \Delta \star_3h + \mathcal{O}(\lambda)\td y^A + \mathcal{O}(1)\td\lambda \; , && \td\star_3h = \mathcal{O}(\lambda)\td y^A + \mathcal{O}(1)\td\lambda \; .
\end{align}
Further analysis depends on whether $\Delta$ and $h$ vanish on the horizon. There are three cases to consider which we examine in detail below.

\paragraph{Case 1a: $\mathring{\Delta} \neq0$ and $\mathring{h} \neq0$.}  One 
can define the following functions and 1-forms on (a possibly smaller) neighbourhood of $\mathcal{H}$,
\begin{align}
	\hat x^{(i)}:= \frac{1}{\sqrt{h^2}}h\cdot Z^{(i)} \; , && \sigma_L^{(i)}:=\frac{\Delta^2+h^2}{\Delta}Z^{(i)}+\frac{1}{\Delta}\td\left(h\cdot Z^{(i)}\right) \; , \label{eq_sigmadef}
\end{align}
which satisfy
\begin{equation}
	\td\sigma_L^{(i)}=-\frac{1}{2}\epsilon_{ijk}\sigma_L^{(j)}\wedge\sigma_L^{(k)}+\mathcal{O}(\lambda)\td y^A\wedge \td y^B + \mathcal{O}(1)\td\lambda\wedge \td y^A  \label{eq_sigmaalg}
\end{equation}
and $\hat x^{(i)}\hat x^{(i)}=1$.    Now define vector fields $\xi_L^i$ by  $\langle \sigma_L^i , \xi_L^j\rangle =\delta_{ij}$ and $\xi_L^i(v)=0$ and $\xi_L^i(\lambda)=0$, so in particular $\mathring \xi_L^i$ are the dual vectors to $\mathring \sigma_L^i:= \mathring (\sigma_L^i)_A \td y^A$ (note by our definitions $\sigma_L^i$ has a $\lambda$-component but $\xi_L^i$ does not). The near-horizon analysis~\cite{reall_higher_2004} shows that $\mathring \xi_L^i$ are Killing vector fields of $(H, \mathring \gamma)$ that commute with $\mathring h$.
It follows that in the generic case ($\mathring h\neq0$), the geometry of $H$ is locally isometric to that of a squashed 
three-sphere $S^3$ and the Killing fields of $H$ are exactly $\mathring h$ and $\mathring \xi_L^{(i)}$, which generate a $U(1)\times SU(2)$ isometry.   Therefore,  we deduce that the axial Killing field $W$ restricted to the horizon $\mathring W$ must be an $\mathbb{R}$-linear combination of $\mathring h, \mathring \xi_L^i$.  Hence,  from (\ref{eq_WnearH}) we can write $W$ in some neighbourhood of the horizon as
\begin{equation}
	W = W_0 h + W_i \xi_L^{(i)} + \lambda W_\lambda(y) \partial_\lambda + \mathcal{O}(\lambda)\partial_{y^A} + \mathcal{O}(\lambda^2)\partial_\lambda \; ,\label{eq_Whorizon}
\end{equation}
where $W_0, W_i$ are constants, $W_\lambda$ is a function of $y^A$, and we have adjusted the subleading terms as necessary. 

We are now in a position to compute $\iota_W X^{(i)}$ near the horizon and hence use (\ref{eq_xdef}) to determine the Gibbons-Hawking coordinates $x^i$ in terms of Gaussian null coordinates. For this it is useful to use the following identities:
\begin{equation}
	\td \hat x^{(i)} = \epsilon_{ijk}\hat x^{(j)}\sigma_L^{(k)}+\mathcal{O}(\lambda)\td y^A + \mathcal{O}(1)\td\lambda  \label{eq_dxi_sigma}
\end{equation}
and 
\begin{equation}
	h = \frac{\Delta\sqrt{h^2}}{\Delta^2 + h^2}\hat x^{(i)}\sigma_L^{(i)} + \mathcal{O}(\lambda)\td y^A+O(1) \td \lambda \;. \label{eq_h_sigma}
\end{equation}
Then,  using definitions (\ref{eq_sigmadef}), the expression for the complex structures (\ref{eq_complexstr}), 
and the form of $W$ in (\ref{eq_Whorizon}), we find after a tedious calculation  that
\begin{align}
	\iota_WX^{(i)}=&-\td\left(W_0\sqrt{h^2}\lambda \hat x^{(i)}\right) +  W_j \left(-\frac{\Delta}{\Delta^2+h^2}\delta_{ij}\td\lambda+\frac{\sqrt{h^2}}{\Delta^2+h^2}\epsilon_{ipj}\hat x^{(p)}\td\lambda-\frac{\lambda\Delta}{\Delta^2+h^2}\epsilon_{ijk}\sigma_L^{(k)}\right)\nonumber \\
	&+\lambda W_\lambda\left(\frac{\Delta}{\Delta^2+h^2}\sigma_L^{(i)}-\frac{\sqrt{h^2}}{\Delta^2+h^2}\epsilon_{ijk}\hat x^{(j)}\sigma_L^{(k)}\right)+\mathcal{O}(\lambda^2)\td y^A + \mathcal{O}(\lambda)\td\lambda \; . \label{eq_dxiexpansion}
\end{align}
Thus, taking the exterior derivative of (\ref{eq_dxiexpansion}) we get
\begin{align}
	\td\iota_WX^{(i)} = &\left[\delta_{iq}\left(\frac{\sqrt{h^2}}{\Delta^2+h^2}W_j\hat x^{(j)}-\frac{\Delta}{\Delta^2+h^2}W_\lambda\right)\right.\nonumber\\
	&\qquad+\epsilon_{iqj}\left(-\frac{\Delta}{\Delta^2+h^2}W_j -\frac{\sqrt{h^2}}{\Delta^2+h^2}W_\lambda\hat x^{(j)}\right)\nonumber\\
	&\qquad\qquad\left.-\frac{\sqrt{h^2}}{\Delta^2+h^2}W_q\hat x^{(i)}+\mathcal{O}(\lambda)\right]\left(\sigma_L^{(q)}\wedge \td\lambda\right) + \mathcal{O}(\lambda)\td y^A\wedge \td y^B  \;.
\end{align}
Now, recall that triholomorphicity of $W$ implies $\td \iota_W X^{(i)}=0$. 
The $\sigma_L^{(q)}\wedge \td\lambda$  are linearly independent on the horizon, and by continuity, also on some neighbourhood of the horizon, so the coefficient of these terms must vanish for 
all $i, q$. Contracting the coefficient of these terms with $\delta_{iq}$ and $\frac{1}{2}\epsilon_{iqm}$ and requiring them to vanish at $\lambda=0$, then yields the following linear system of equations for $W_i$, 
\begin{equation}
	A_{mj}W_j:=\left(\frac{2h^2}{3 \Delta} \hat x^{(j)}\hat x^{(m)}+\Delta\delta_{jm}+\frac{1}{2}\sqrt{h^2}\epsilon_{ijm}\hat x^{(i)}\right)_{\lambda=0}W_j=0  \; ,
\end{equation}
with $W_\lambda=\tfrac{2}{3} (\sqrt{h^2} \hat{x}^{(j)})_{\lambda=0} W_j$.
Since the determinant of the matrix
\begin{equation}
	\det A = \frac{(2 \mathring h^2+3\mathring \Delta^2)(\mathring h^2+4\mathring \Delta^2)}{12\mathring \Delta}\neq 0 \; ,
\end{equation}
it follows that $W_i=0$, which also implies that $W_\lambda=0$. 

Therefore, we have shown that  $W=W_0 h +\mathcal{O}(\lambda) \partial_{y^A}+\mathcal{O}(\lambda^2) \partial_\lambda $ for some constant
 $W_0\neq 0$. 
Thus substituting back into (\ref{eq_dxiexpansion})  the definition for the cartesian coordinates (\ref{eq_xdef}) yields
\begin{equation}
	\td x^i = -\td\left(W_0\sqrt{h^2}\lambda \hat x^{(i)}\right) + \mathcal{O}(\lambda^2)\td y^A + \mathcal{O}(\lambda)\td\lambda \; .  \label{eq_dxi}
\end{equation}
At $\lambda =0$, $\td x^i\propto \td \lambda$, therefore $x^i$ are constant on the horizon, and the cross-section $H$ of each connected component of the horizon corresponds to a 
single point in $\mathbb{R}^3$. Integrating (\ref{eq_dxi}) and taking $a^i$ to correspond to a connected component of the horizon yields 
\begin{equation}
	x^i-a^i = -W_0\sqrt{\mathring h^2}\lambda \hat x^{(i)}+ \mathcal{O}(\lambda^2) \; ,
\end{equation}
and the euclidean distance on $\mathbb{R}^3$ is
\begin{equation}
	r = \sqrt{(x^i-a^i)( x^i-a^i)} = |W_0|\sqrt{\mathring h^2}\lambda+\mathcal{O}(\lambda^2) \; .\label{eq_rlambda}
\end{equation}

\paragraph{Case 1b:  $\mathring h=0$.} In this case one must have $\mathring \Delta \neq 0$ and $h= \mathcal{O}(\lambda) \td y^A$.
We again define $\sigma_L^{(i)}$ as in (\ref{eq_sigmadef}) which satisfy (\ref{eq_sigmaalg}), and dual vectors $\xi_L^{(i)}$ that satisfy $\xi_L^{(i)}(v)= \xi^{(i)}_L(\lambda)=0$. Thus, in this case we can write
\begin{equation}
Z^{(i)}= \frac{\sigma^{(i)}_L}{\Delta}+ \mathcal{O}(\lambda) \td y^A+ \mathcal{O}(1) \td \lambda  \; .
\end{equation}
 Then, using (\ref{eq_complexstr}), we find 
\begin{align}
	X^{(i)}&=\td\left(\frac{\lambda\sigma_L^{(i)}}{\Delta}\right)+\mathcal{O}(\lambda^2)\td y^A\wedge \td y^B + \mathcal{O}(\lambda)\td\lambda\wedge \td y^A.
\end{align}
The horizon geometry $\mathring \gamma= \mathring \sigma_L^{(i)}\mathring \sigma_L^{(i)} /\mathring \Delta^2$ where we again define $\mathring\sigma_L^{(i)} := (\mathring\sigma_L^{(i)})_A \td y^A$.   It follows that the horizon geometry is  isometric to a round $S^3$ or a lens space. The dual vectors $\mathring \xi_L^{(i)}$ are now all Killing vector fields of $(H, \mathring{\gamma})$, and thus we can write the axial Killing field on the horizon as  $\mathring W = W_i  \mathring \xi_L^{(i)}  + \mathring{W}_R$ where $W_i \in \mathbb{R}$ and $\mathring{W}_R$ is a left-invariant vector field so $\mathcal{L}_{\mathring W_R} \mathring \sigma_L^{(i)}=0$. Therefore, by (\ref{eq_WnearH}),  in a neighbourhood of the horizon we can write
\begin{equation}
W =\mathring W_R +W_i \mathring \xi_L^{(i)} + \lambda W_\lambda(y) \partial_\lambda + \mathcal{O}(\lambda)\partial_{y^A} + \mathcal{O}(\lambda^2)\partial_\lambda \; .
\end{equation} 
Then a short computation using the above gives
\begin{align}
	\iota_W X^{(i)}=-\td \left(\frac{W^i_R}{\Delta}\lambda \right) + \frac{\lambda W_\lambda}{\Delta}\sigma_L^{(i)} -W_j \left(\frac{\delta_{ij}\td\lambda}{\Delta} + \frac{\lambda}{\Delta}\epsilon_{ijk}\sigma^{(k)}\right)+\mathcal{O}(\lambda^2)\td y^A + \mathcal{O}(\lambda)\td \lambda \; ,
	\label{eq_dxi2}
\end{align}
where $W^i_R:= \iota_{\mathring W_R} \mathring \sigma_L^{(i)}$. It follows that
\begin{equation}
\td \iota_{W} X^{(i)} = \sigma_L^{(k)} \wedge \td \lambda \left[ \frac{W_j \epsilon_{ijk}}{\Delta} - \frac{\delta_{ik} W_\lambda}{\Delta} +\mathcal{O}(\lambda) \right] + \mathcal{O}(\lambda) \td y^A \wedge \td y^B \; ,
\end{equation}
and therefore we deduce that triholomorphicity of $W$ implies $W_i=W_\lambda=0$. Therefore,  from (\ref{eq_xdef}) we again
we find that $\td x^i \propto \td \lambda$ at the horizon, so each connected component of the  horizon is a point say $a^i$ in $\mathbb{R}^3$. By 
integrating (\ref{eq_dxi2}) we find
\begin{align}
	x^i-a^i = -\frac{W_R^i}{\mathring\Delta}\lambda  + \mathcal{O}(\lambda^2) \; , \qquad r = \left|\frac{W_0}{\mathring\Delta}\right|\lambda + \mathcal{O}(\lambda^2) \; ,
\end{align}
where $W_0^2 := W^i_R W^i_R$ is a constant (since $\td W^i_R= \epsilon_{ijk} W^j_R \mathring \sigma_L^{(k)}$), which must be non-zero to avoid $W$ vanishing identically on the horizon.

\paragraph{Case 2: $\mathring\Delta=0$.}\label{case3} In this case $Z^{(i)}$ can be written in terms of coordinates $z^i$ as~\cite{reall_higher_2004} 
\begin{equation}
	Z^{(i)}=K(z^j)\td z^i+\mathcal{O}(\lambda)\td z^i\; , \qquad 
	h = \td \log K +\mathcal{O}(\lambda)\td z^i \; ,
\end{equation}
with  $K = L/\sqrt{z^iz^i}$ for some constant $L$.
Let us introduce standard spherical polar coordinates $(z^i)\to (R, \theta, \phi )$, and define $\psi:=\log R$. The near-horizon data is then
\begin{equation}
	\mathring \gamma=L^2(\td\psi^2+\td\theta^2+\sin^2\theta \td\phi^2)\; , \qquad \mathring h=  - \td\psi\; ,
\end{equation}
and 
\begin{equation}
	X^{(i)}= L \td \left( \lambda ( \hat x^{(i)} \td\psi+ \td \hat x^{(i)})\right)+\mathcal{O}(\lambda^2)\td y^A\wedge \td y^B + \mathcal{O}(\lambda)\td\lambda\wedge \td y^A \; ,
\end{equation}
where $\hat x^{(i)}(\theta, \phi):= z^i/R$.
The horizon cross-section $H$ is locally isometric to $S^1\times S^2$, and its independent Killing fields are $\partial_\psi$ and the standard 
Killing fields of $S^2$.  Therefore, without loss of generality we can always adapt the coordinates on $S^2$ so that $\mathring W= W_\psi \partial_\psi+ W_\phi \partial_\phi$, where $W_\psi, W_\phi$ are constants, and hence write $W$ in the neighbourhood of the horizon (\ref{eq_WnearH}) as
\begin{equation}
	W = W_\psi \partial_\psi + W_\phi \partial_\phi + \lambda W_\lambda(\theta,\phi,\psi) \partial_\lambda +\mathcal{O}(\lambda)\partial_{y^A} + \mathcal{O}(\lambda^2)\partial_\lambda \; ,
\end{equation}
which  yields
\begin{align}
	\iota_WX^{(i)}&= -\td\left(LW_\psi\lambda \hat x^{(i)}\right)+LW_\phi(\lambda \td\psi-\td\lambda)\partial_\phi \hat x^{(i)}+\nonumber\\
	&+L\lambda W_\lambda(\hat x^{(i)}\td\psi+\td\hat x^{(i)})+\mathcal{O}(\lambda^2)\td y^A + \mathcal{O}(\lambda)\td \lambda \; . \label{eq_dxi3}
\end{align}
It follows that
\begin{equation}
\td \iota_W X^{(i)} =L \left(  W_\phi  \partial_\phi \hat{x}^i+  W_\lambda \hat{x}^i  + O(\lambda)\right) \td \lambda\wedge \td \psi+ \dots \; ,
\end{equation}
where $\dots$ represent terms not proportional to $\td \lambda \wedge \td\psi$. Therefore, we deduce that $\iota_WX^{(i)}$ is closed (to leading order) if and only if 
$W_\phi=W_\lambda = 0$ (to see this note $\hat{x}^i \partial_\phi \hat{x}^i=0$). As before, in order for $W$ not to vanish identically, $W_\psi$ must be non-zero. From (\ref{eq_xdef}) we also find that $x^i$ are constant on the horizon, 
and (\ref{eq_dxi3}) can be integrated to get 
\begin{align}
	x^i -a^i= -LW_\psi \lambda \hat x^{(i)}+\mathcal{O}(\lambda^2) \; , \qquad r = |LW_\psi |\lambda+\mathcal{O}(\lambda^2) \; .
\end{align} 
\end{proof}

\noindent \textbf{Remark.} In Gaussian null coordinates (\ref{eq_gaussianNullCoords}) one can check 
$(\mathcal{L}_W g)_{\lambda\lambda} = 2 \partial_\lambda W^v$ so that $W^v=W^v(y)$.  If one assumes 
that $W$ is tangent to the cross-section $H$ of the horizon it therefore follows that $W^v=0$. Then, 
$(\mathcal{L}_W g)_{\lambda A}= 0$ implies $\partial_\lambda W^A=0$ and in turn  $(\mathcal{L}_W g)_{v\lambda}=0$ 
implies $\partial_\lambda W^\lambda=0$.  By assumption, $W$ is tangent to the horizon and hence $W^\lambda|_{\lambda=0}=0$ 
as in (\ref{eq_WnearH}), so we deduce that $W^\lambda=0$. This shows that in Gaussian null coordinates, any Killing vector 
that is tangent to the cross-section takes the form $W= W^A(y) \partial_{y^A}$ in the whole neighbourhood of the horizon 
(not just on the horizon), that is, $W$ is tangent to constant $(v, \lambda)$ surfaces even for $\lambda \neq 0$. 
Therefore, if one assumes that $W$ is tangent to $H$ (which we have not) the proof of Lemma \ref{lemma_GNC_cart} 
simplifies somewhat. \\

The proof of Lemma \ref{lemma_GNC_cart} also reveals the following important property of the axial Killing field.
\begin{corollary}\label{cor_Wnonzero}
The axial Killing field $W$ has no zeroes on the horizon.
\end{corollary}

We are  now in a position to determine the precise singular 
 behaviour of the  harmonic functions near a horizon.
\begin{lemma}
	The associated 
	harmonic functions in a neighbourhood of a connected component of the horizon $x^i=a^i$ are of the form
	\begin{equation}
		H =  \frac{h}{|\bm x - \bm a|}+ \tilde H \; , \quad K=  \frac{k}{|\bm x - \bm a|}+ \tilde K \; , \quad L =  \frac{l}{|\bm x - \bm a|}+ \tilde L \; , 
	 \quad M =  \frac{m}{|\bm x - \bm a|}+ \tilde M  \; ,  \label{eq_harmonichorizon}
	\end{equation} where $h,k, l, m$ are
	(possibly zero) constants and $\tilde H, \tilde K, \tilde L, \tilde M$ are harmonic functions regular at $a^i$.
	\label{lemma_harmonic_horizon}
\end{lemma}
\begin{proof}
In Gaussian null coordinates the invariants $g(V, V)= -\lambda^2 \Delta^2$, $g(V, W)= \lambda \mathring W \cdot \mathring h + \mathcal{O}(\lambda^2)$ and hence (\ref{eq_Ndef}) gives
\begin{equation}
N = \lambda^2 (N_0+ \mathcal{O}(\lambda)) \; , \qquad N_0 :=  | \mathring W|^2\mathring \Delta^2  +  (\mathring W \cdot \mathring h)^2  \; ,
\end{equation}
where we have used (\ref{eq_WnearH}), together with the fact that $\tilde{W}_\lambda= \mathcal{O}(\lambda)$ which follows from the proof of Lemma \ref{lemma_GNC_cart}.  Crucially, since $W$ has no zeroes on the horizon (Corollary \ref{cor_Wnonzero})  the function $N_0$ is strictly positive on the horizon for all types of near-horizon geometry (see three cases in proof of Lemma \ref{lemma_GNC_cart}). Therefore, the harmonic functions (\ref{eq_harmonicinvariant}) near the horizon take the form
\begin{align}
H &= \frac{ \Delta}{N_0  \lambda }+\mathcal{O}(1) \; , \qquad L = \frac{\Delta |W|^2+2 (W \cdot h) \Psi- \Delta \Psi^2}{N_0 \lambda}+\mathcal{O}(1)\; ,  \\
K&= \frac{\Delta \Psi- W\cdot h}{N_0 \lambda}+ \mathcal{O}(1) \;, \qquad M= \frac{ |W|^2 ( W \cdot h)-3 \Delta \Psi |W|^2- 3 \Psi^2 (W\cdot h)+\Delta \Psi^3}{2N_0 \lambda}+\mathcal{O}(1)  \; .\nonumber
\end{align}
Therefore, by Lemma \ref{lemma_GNC_cart} we deduce that $| \bm x - \bm a | H = \mathcal{O}(1)$ as $x^i \to  a^i$ and similarly for the other harmonic functions.  The claim now follows by standard harmonic function theory.
\end{proof}

\subsection{Orbit space and general form of harmonic functions}
\label{ssec_orbit}

The orbit space is defined by
\begin{equation}
\hat\Sigma : = \llangle \mathcal{M} \rrangle/[\mathbb{R}\times U(1)]  \cong \Sigma /U(1) \; ,   \label{eq_Sigmahat}
\end{equation}
where the $\mathbb{R}\times U(1)$ action is defined by the flow of the Killing fields $V, W$ and the second equality follows from Remark \ref{remark_orbit}. 
By Corollary \ref{cor_AFGH} we deduce that the orbit space $\hat{\Sigma}$ has an end diffeomorphic to $\mathbb{R}^3 \backslash B^3$ on which the 
Gibbons-Hawking cartesian coordinates $x^i$ are a global chart.  We now turn to a detailed study of the orbit space.

An extensive analysis of the structure of such orbit spaces has been performed in \cite{hollands_further_2011}. In general, $\hat \Sigma$ is a simply connected topological 
space with a boundary $\partial\hat\Sigma= \hat H \cup_{i=1}^l S^2_i \cup S^2_\infty$, where $\hat H$ is the orbit space of the event horizon, $S^2_\infty$ denotes the asymptotic boundary and  $S^2_i$ correspond to fixed points of the $U(1)$ action (i.e. zeroes of $W$ corresponding to `bolts'~ \cite{gibbons_classification_1979}). The interior of the orbit space is the 
union of three kinds of points $\hat\Sigma=\hat L \cup \hat E \cup \hat F$, corresponding to regular orbits (trivial isotropy), exceptional orbits (discrete isotropy) and fixed points ($U(1)$ isotropy), respectively. 
$\hat L$ is open in $\hat \Sigma$ and has a structure of a manifold, internal fixed points of $\hat\Sigma$ are isolated, and $\hat E$ are 
smooth arcs ending on either fixed points or $\hat H$ (they cannot form closed loops \cite{fintushel_circle_1977}).

\begin{lemma} \label{lemma_orbit}
The interior of the orbit space $\hat{\Sigma} = \hat L  \cup \hat F$ where $\hat L$ corresponds to regular orbits and $\hat F$ to a 
finite number of isolated fixed points, i.e., the $U(1)$-action has no exceptional orbits. Furthermore, its boundary $\partial \hat{\Sigma}= \hat{H}\cup S^2_\infty$, 
i.e. there are no fixed points corresponding to bolts.
\end{lemma}

\begin{proof}
First recall that Assumption \ref{assumption2} \ref{assumption_span} implies $f\neq0$ at the zeros of $W$, so the fixed points correspond to internal points of the base 
manifold $B$. Let us define $\widetilde{W}^\flat$ to be the metric dual of $W$ with respect to $h$. As a consequence of triholomorphicity, 
$\td \widetilde{W}^\flat$ is self-dual in $B$ (see e.g. \cite{gibbons_hidden_1988}), hence each of the fixed points of $W$ corresponds to a `nut' 
of type $(\pm1, \pm1)$ and  `bolts' are not possible  \cite{gibbons_classification_1979}. As a result, fixed points of $W$ must be isolated in $B$ and 
hence correspond to internal fixed points in $\hat{\Sigma}$. $W$ has no zeros in a neighbourhood of horizon components (Corollary \ref{cor_Wnonzero}) or 
at spatial infinity (Lemma \ref{lem_AF}), hence by Assumption \ref{assumption1} \ref{assumption_cauchy} $\hat F$ is contained in a compact set. $\hat F$ is also 
closed in $\hat\Sigma$, thus it must be finite. Furthermore, since fixed points correspond to a `nut' of type $(\pm1, \pm1)$,  there cannot be 
any arc of exceptional orbits ending on them (see~\cite{hollands_further_2011}). Thus, if there are arcs of exceptional orbits, those arcs must end on 
$\hat H$. Assume for contradiction that there is an arc of exceptional orbit ending on a horizon component $\hat H_i$. $H_i$ cannot have $S^2\times S^1$ 
topology, because in that case all points of $H_i$ have the same isotropy group (see Case 3 in Section \ref{case3} and \cite{reall_higher_2004}). It follows 
that $\mathring\Delta\neq 0$ for $H_i$, therefore there must be exceptional orbits with $f>0$ in their neighbourhood. In the base of such a neighbourhood 
we can use Gibbons-Hawking coordinates $(\psi, x^i)$, which excludes the possibility of an exceptional orbit since on such a chart the period of $\psi$ is fixed. (A more 
detailed and technical argument is given in Appendix \ref{app_exceptional}).
\end{proof}

The spacetime invariants $f, \Psi, N, x^i$ that we have constructed are preserved by the Killing fields $V, W$ and therefore descend to functions on the orbit space.\footnote{By abuse of notation, we denote such invariant functions on $\mathcal{M}$ and their corresponding functions on the orbit 
space by the same letter.}   It follows by  Lemma \ref{lemma_orbit} that fixed points in Gibbons-Hawking coordinates $x^i$ correspond to points in $\mathbb{R}^3$. We shall now prove that $x^i$ can be used as global coordinates on $\hat L$, so that in particular $\hat L$ is diffeomorphic to $\mathbb{R}^3$ with 
a finite set of points removed corresponding to the image of fixed points in $\hat F$ and horizon components $\hat{\mathcal{H}}_i$ (recall a horizon component in Gibbons-Hawking coordinates also corresponds to a point in $\mathbb{R}^3$ by Lemma \ref{lemma_GNC_cart}).

\begin{lemma}
	The function\footnote{
		For compact notation we will denote the vector of functions $x^i$ by $\bm x$.
	} $\bm x: \hat L\to \mathbb{R}^3\setminus \bm x(\hat{\mathcal{H}}\cup \hat F)$ is a diffeomorphism.\label{lemma_bijection}
\end{lemma}
	
\begin{proof}

Let us start by showing that $\bm x$ is a local diffeomorphism on $\hat L$. From Lemma \ref{Lemma_GH} it follows that 
$N$ is preserved by $V, W$, so it descends to the orbit space, furthermore $N>0$ on $\hat L$. Recall, that $\hat L$ is a manifold. The algebraic relations (\ref{eq_XX}) and definition (\ref{eq_xdef}) and (\ref{eq_Ndef}) imply that on the spacetime 
\begin{align}
	g^{-1}(\td x^i, \td x^j) = W^\mu W^\nu X^{(i)}_{\mu \rho }X^{(j)}_\nu{}^\rho = W^\mu W^\nu\delta_{ij}(f^2g_{\mu\nu}+V_\mu V_\nu) = N\delta_{ij} \; ,
\end{align}
hence $\td x^i$ are linearly independent in $T^*_q\mathcal{M}$ for any $q\in \mathcal{M}$ where $N(q)>0$, which is the case in $\llangle\mathcal{M}\rrangle$ on regular orbits. 
Since $\iota_W \td x^i = \iota_V \td x^i =0$, $\td x^i$ are also linearly independent in $T^*_p\hat L$  for all $p\in \hat L$. Therefore, $\bm x$ 
is a local diffeomorphism on $\hat L$. For it to be a diffeomorphism onto $\mathbb{R}^3\setminus \bm x(\hat{\mathcal{H}}\cup \hat F)$, we need to 
show surjectivity and global injectivity.

First, we shall prove surjectivity. Let us define the dual vectors $e_i$ on $\hat L$, i.e. 
$\td x^i(e_j)=\delta_j^i$. Let $\bm a:=\bm x(p)\in \mathbb{R}^3$ for some $p\in \hat L$, and 
$\bm y\in \mathbb{R}^3\setminus \bm x(\hat F\cup \hat{\mathcal{H}})$ an arbitrary point. Consider a path in 
$\mathbb{R}^3 \setminus \bm x(\hat F \cup \hat{\mathcal{H}})$ from $\bm a$ to $\bm y$ that is a union of line segments, such that straight 
continuation of any line segment stays in $\mathbb{R}^3 \setminus \bm x(\hat F \cup \hat{\mathcal{H}})$ (i.e. there is no fixed point or 
horizon mapped to the continuation of the segments). This can be done using two segments due to the fact that $|\hat F|$ and the number 
of horizon components are finite (Assumption \ref{assumption1} \ref{assumption_cauchy} and Lemma \ref{lemma_orbit}). For the segment ending on $\bm a$ let the vector tangent to it be $u^i \bm e_i$ with $\{ \bm e_i\}$ being the 
standard basis of $\mathbb{R}^3$. Then consider the maximal integral curve $\gamma$ of $U=u^i e_i$ 
starting at $p$ in $\hat L$.  Assume for contradiction that it is incomplete, which means that  $x^i(\gamma(t)) = a^i + t u^i $ is 
bounded. However, due to the Escape Lemma (see e.g. \cite{lee_introduction_2012}), the image of $\gamma$ cannot be contained in any compact subset 
of $\hat L$, therefore it must approach the asymptotically flat end (recall by construction $\gamma$ does not approach a horizon or fixed point and our Assumption \ref{assumption1} \ref{assumption_cauchy}). 
Thus, $x^i$ is bounded along $\gamma$ as we approach the asymptotically flat end, which
by Lemma \ref{lem_AF} is a contradiction. Therefore, $U=u^i e_i$ must be a complete vector field. One can similarly show that the vector field $V=v^i e_i$,  where $v^i \mathbf{e}_i$ is tangent to the line in $\mathbb{R}^3$ that ends at $\bm {y}$, is complete.  This shows that starting at $p \in  \hat L$, we can follow the integral curves of $U$ and $V$ to
reach a point $q \in\hat L$ such that $\bm x(q)=\bm y$. But $\bm y$ is arbitrary in $\mathbb{R}^3 \setminus \bm x(\hat F \cup \hat{\mathcal{H}})$  and hence
 $\bm x$ is surjective.

We next show that $\bm x$ is injective. For contradiction, let us assume 
that $p\neq q\in \hat L$ and $\bm x(p)=\bm x(q)=: \bm x_0$. As above, let us choose a straight line through $\bm x_0$ in 
$\mathbb{R}^3\setminus \bm x(\hat F\cup \hat{\mathcal{H}})$ with tangent vector $U=u^i \bm e_i$, and let $\gamma_p(t)$ and $\gamma_q(t)$ denote 
the two integral curves of $U$ in $\hat L$ starting at $\gamma_p(0)=p$ and $\gamma_q(0)=q$. The two curves are disjoint by the 
uniqueness of integral curves. The straight line in $\mathbb{R}^3$ does not go 
through any fixed points or horizon components, hence by using the argument of the previous paragraph $\gamma_p$ and $\gamma_q$ are complete. We claim that $\gamma_p, \gamma_q$ must enter the asymptotically flat end of $\hat L$. For contradiction,  suppose the contrary, so that these curves are contained in a compact set $K\subset \hat L$. Then by continuity of $\bm x$ the image $\bm x(K)$ is a compact subset of $\mathbb{R}^3$. On the other hand, by completeness of the curves $\bm x(\gamma_p)= \bm x(\gamma_q)=\{\bm x_0 + \bm u t, t\ge0\}$ is unbounded, so cannot be contained in a compact subset of $\mathbb{R}^3$. Therefore, we have a contradiction, so  $\gamma_p, \gamma_q$ must enter the asymptotically flat end of $\hat L$ as $t\to\infty$. This means that for any large enough $|\bm x|$, there exist two distinct 
points with the same $\bm x$ value. This violates asymptotic flatness, since by Corollary \ref{cor_AFGH} the $\bm x$ are global coordinates on the asymptotically flat end of $\mathbb{R}^3$ (thus injective). Therefore, we have obtained a contradiction, and we deduce that  $\bm x$ is globally injective and hence a diffeomorphism.
\end{proof}

\noindent \textbf{Remark.}  $\hat \Sigma \cup_i\{\hat{\mathcal{H}}_i\}$ is in bijection with $\mathbb{R}^3$  (here we are adding each horizon 
component as a single point). This follows from continuity of $\bm x$ on $\mathcal{M}$ and 
injectivity on $\hat L$.

\begin{corollary}
	The set $\{p\in\llangle\mathcal{M}\rrangle: f(p)\neq0 \}$ is dense in $\llangle\mathcal{M}\rrangle$.\label{cor_fdense}
\end{corollary}
\begin{proof}
Since $H$ is harmonic on $\mathbb{R}^3$, it is also real-analytic in $x^i$ (on its domain). It follows that if $H=0$ on some open set of 
$\mathbb{R}^3$,  it is zero everywhere, and so by (\ref{eq_harmonicinvariant}) $f$ vanishes identically, which cannot happen (e.g. by asymptotic flatness). Therefore, the set 
$\{\bm x\in\mathbb{R}^3: f(\bm x)\neq0 \}$ is dense in $\mathbb{R}^3$. By Lemma \ref{lemma_bijection}, the $x^i$ are global coordinates on the orbit space $\hat\Sigma$, so 
$\{p\in \hat\Sigma: f(p)\neq0 \}$ is also dense in $\hat\Sigma$, and since the quotient map is open, $\{p\in\llangle\mathcal{M}\rrangle: f(p)\neq0 \}$ 
is also dense in $\llangle\mathcal{M}\rrangle$, as claimed.
\end{proof}
 
We now determine the behaviour of the harmonic functions $H,K,L,M$ at a fixed point.
\begin{lemma}\label{lemma_harmonic_fixedpt}
	Let $p\in \hat\Sigma$ be a fixed point of $W$ as above. $H, K, L, M$ have (at most) simple poles at $\bm x(p)\in \mathbb{R}^3$.\label{lemma_fixedpt}
\end{lemma}
\begin{proof}
 By the remark below Lemma \ref{lemma_harmonicwelldefined} $f$ is non-zero on some neighbourhood of $p$ in $\hat\Sigma$ and $N>0$ on this neighbourhood except at $p$.  Therefore, from (\ref{eq_harmonicinvariant}) we 
see that $H$ is also non-zero on some neighbourhood of $p$ in $\hat \Sigma$. By Lemma \ref{lemma_bijection}, $\bm x: \hat\Sigma\to \mathbb{R}^3$ is surjective to some neighbourhood of $\bm x(p)$ and therefore $H$ is a harmonic function non-vanishing on a neighbourhood of $\bm x(p)$ in $\mathbb{R}^3$. Using B\^{o}cher's theorem for harmonic functions on 
$\mathbb{R}^3$ (see e.g. \cite{axler_harmonic_2001}), we see that $H$ has a simple pole at $\bm x(p)$ on $\mathbb{R}^3$. By 
(\ref{eq_harmonicinvariant}), it follows that all other harmonic functions $K,L,M$ have (at most) simple poles at $\bm x(p)$.
\end{proof}

We can now put  together the results we have obtained so far to completely fix the functional form of the harmonic functions for any solution satisfying our assumptions.

\begin{theorem}\label{theorem_harmonicfunctions}
	Any solution $(\mathcal{M}, g, F)$ of $D=5$ minimal supergravity satisfying Assumption \ref{assumption1} and \ref{assumption2} must 
have a Gibbons-Hawking base (wherever $f \neq 0$) and the associated harmonic functions $H,K,L,M$ are of `multi-centred' form
\begin{align}
	H =  \sum_{i=1}^N \frac{h_i}{r_i} \; , && K =  \sum_{i=1}^N \frac{k_i}{r_i} \; , && L = 1+ \sum_{i=1}^N \frac{l_i}{r_i} \; , &&M = m+ \sum_{i=1}^N \frac{m_i}{r_i}  \; , 
\end{align}
where $m\in \mathbb{R}$ is a constant, $r_i := |\bm x-\bm a_i|$ and the centres $\bm a_i\in \mathbb{R}^3$ are the coordinates of fixed points of $W$ or connected horizon components, and  $h_i, k_i, l_i, m_i$  are constants satisfying
\begin{equation}
\sum_{i=1}^N h_i =1  \; . \label{eq_sumh}
\end{equation}
\label{thm_necessary}
\end{theorem}

\begin{proof}
Lemmas \ref{lemma_harmonicwelldefined}, \ref{lemma_harmonic_horizon}, \ref{lemma_bijection},  \ref{lemma_harmonic_fixedpt} imply that  the harmonic functions can be written as
\begin{equation}
H= \sum_{i=1}^N \frac{h_i}{r_i} +  \tilde{H}
\end{equation}
for some constants $h_i$, where $\tilde{H}$ is a harmonic function that is regular everywhere in $\mathbb{R}^3$. The finiteness of the number of centres follows from Assumption \ref{assumption1} 
\ref{assumption_cauchy} and Lemma \ref{lemma_orbit}. On the other hand, asymptotic flatness implies Corollary \ref{cor_AF_H}, which implies (\ref{eq_sumh}) and that $\tilde{H}\to 0$ in the asymptotically flat end.  Therefore, $\tilde{H}$ is a bounded everywhere regular harmonic function on $\mathbb{R}^3$ and hence must be a constant, and this constant  vanishes using (\ref{eq_AF_H}) again. Thus, $H$ takes the claimed form.  An identical argument works for the other harmonic functions $K, L, M$ using (\ref{eq_AF_KLM})  giving the claimed form. 
\end{proof}

A consequence of Theorem \ref{thm_necessary} is that 
a solution in this class is determined by choosing $N$ points on $\mathbb{R}^3$ corresponding to  the simple poles of the harmonic functions, and assigning weights to 
each of the poles. 
However, it is not guaranteed that all such solutions correspond to a solution that is smooth in the DOC and at the horizon. In the next section we will determine the necessary and sufficient criteria for this.

\section{Smoothness of multi-centred solutions}\label{sec_Sufficient}

In this section we will determine the conditions required for smoothness of the solution in Theorem \ref{thm_necessary} at the horizon and the fixed points. In each case 
the strategy is the same: we locally expand the harmonic functions in terms of spherical harmonics around a centre.

\subsection{Regularity and topology of the horizon}\label{ssec_horizonRegularity}

As we showed in Lemma \ref{lemma_GNC_cart}, a  connected component of the horizon corresponds to a simple pole in $\mathbb{R}^3$ of the harmonic functions associated to the Gibbons-Hawking base space. Without loss of generality we can take a horizon component at the origin of $\mathbb{R}^3$, so the harmonic functions take the form
\begin{equation}
H= \frac{h_{-1}}{r}+ h_0+ \tilde{H} \; ,   \label{eq_harmonicexpansion}
\end{equation}
where $h_{-1}, h_0$ are constant, $\tilde{H}$ is a harmonic function which is smooth (in fact analytic) and vanishes at $r=0$, and similarly for $K, L, M$. It will sometimes be useful to expand $\tilde{H}= \sum_ {l\geq 1, |m| \leq l} h_{lm} r^l Y^m_l$ where $Y^m_l(\theta, \phi)$ are the spherical harmonics on $S^2$ and $h_{lm}$ are constants.  It then follows from (\ref{eq_chieq}) and (\ref{eq_xi}) that the 1-forms take the form, up to a gradient,
\begin{align}
	\chi &= (h_{-1}\cos\theta+\chi_0)\td\phi +\tilde \chi \; ,   \label{eq_chi_nearH}\\
	\xi &= (-k_{-1}\cos\theta+\xi_0) \td\phi +\tilde \xi \; , \label{eq_xi_nearH}
	\end{align}
where we have used the identity $\star_3 \td (\cos\theta \td \phi)= \td (r^{-1})$, $\chi_0, \xi_0$ are constants and $\tilde{\chi}, \tilde{\xi}$ are 1-forms that satisfy $\star_3 \td \tilde{\chi}= \td \tilde H$ and $\star_3 \td \tilde{\xi}= \td \tilde K$. Therefore,  in particular, $\tilde{\chi}, \tilde{\xi}$ must be smooth 1-forms on $\mathbb{R}^3$.   Upon expanding the harmonic functions in spherical harmonics  we find that, up to a gradient,
\begin{equation}
\tilde\chi= \sum_{\substack{ l \geq 1\\ |m| \leq l}} \frac{h_{lm}  r^{l+1}}{l+1} \star_2 \td Y^m_l  \; ,  \label{eq_chi_nearH_polar}
\end{equation}
where $\star_2$ is the Hodge star operator for the metric $\td \Omega^2$ on the unit $S^2$, and similarly for $\tilde\xi$.

In order to determine the other 1-form $\hat\omega$  we need to solve (\ref{eq_omegahat}), which is a bit more complicated.  We can decompose this, up to a gradient, as
\begin{equation}
\hat\omega = (\omega_0 +\omega_{-1} \cos \theta) \td \phi +\tilde\omega\; , \qquad \tilde\omega:= \hat{\omega}_{\text{sing}} + \hat{\omega}_{\text{reg}} \;, \label{eq_omegahat_nearH}
\end{equation}
where $\omega_0$, $\omega_{-1}$ are constants and $\hat{\omega}_{\text{sing}}$, $\hat{\omega}_{\text{reg}}$ are 1-forms defined by
\begin{align}
\omega_{-1} &:= h_0m_{-1}-m_0h_{-1}+\tfrac{3}{2}(k_0l_{-1}-l_0k_{-1})\; , \label{eq_omegam1}  \\
\star_3 \td \hat{\omega}_{\text{sing}} &= \frac{1}{r} \td F- F\td \left(\frac{1}{r} \right)\; , \qquad F:= h_{-1}\tilde{M} - m_{-1} \tilde{H}+ \tfrac{3}{2} ( k_{-1} \tilde{L}- l_{-1} \tilde{K}) \; ,  \label{eq_doms}\\
\star_3 \td \hat{\omega}_{\text{reg}}&= \td \left(h_{0}\tilde{M} - m_{0} \tilde{H}+ \tfrac{3}{2} ( k_{0} \tilde{L}- l_{0} \tilde{K}) \right) + \tilde{H} \td \tilde{M}- \tilde{M} \td \tilde{H}+ \tfrac{3}{2} ( \tilde{K}\td \tilde{L}- \tilde{L}\td \tilde{K})   \; .  \label{eq_domr}
\end{align}
In particular,  $\hat{\omega}_{\text{reg}}$ is determined by the regular parts of the harmonic functions and therefore  must be a smooth 1-form on $\mathbb{R}^3$. On the other hand,  $\hat{\omega}_{\text{sing}}$ receives contributions from the singular parts of the harmonic functions and thus requires a little more care.  In fact, by expanding the harmonic functions in spherical harmonics, $F= \sum_{l\geq 1, |m| \leq l} f_{lm} r^l Y^m_l$ for constants $f_{lm}$, one can derive the explicit expression (again up to a gradient),
\begin{equation}
\hat{\omega}_{\text{sing}}= \sum_{\substack{ l \geq 1\\ |m| \leq l}}  \frac{f_{lm} r^l }{l} \star_2 \td Y^m_l  \; . \label{eq_omegasing}
\end{equation}
 In particular, $\hat{\omega}_{\text{sing}}$ is a smooth 1-form on $S^2$ for each fixed value of $r$, since the spherical harmonics $Y^m_l$ are smooth on $S^2$, and vanishes at $r=0$.

We now have all the ingredients to construct the spacetime metric and gauge field near the horizon.  In fact,  since the first two  orders in the $r$-expansions of the harmonic functions are $\phi$-independent,  the analysis is essentially identical to that in the case of solutions with biaxial symmetry~\cite{breunholder_moduli_2019}.

Using (\ref{eq_omegapsi}), (\ref{eq_f}) and (\ref{eq_Ndef}) it follows that near the horizon the invariants take the form
\begin{align}
  N^{-1} &= \frac{\alpha_0}{r^2} + \frac{\alpha_1}{r}+\mathcal{O}(1) \; , \qquad f= \frac{h_{-1}}{\alpha_0} r +\frac{h_0\alpha_0- h_{-1}\alpha_1}{\alpha_0^2} r^2+\mathcal{O}(r^3)\; , \label{eq_inv1_nearH} \\
  g_{\psi\psi} &= \beta_0+r\beta_1+\mathcal{O}(r^2) \; , \qquad  g_{t\psi} = r(\gamma_0+r\gamma_1)+\mathcal{O}(r^3) \; ,  \label{eq_inv2_nearH}
\end{align}
where  $\alpha_i$, $\beta_i$, $\gamma_i$ are constants and the error terms are analytic in $r$ and smooth on $S^2$ (since they depend on the spherical harmonics). Since $N>0$ in the DOC away from fixed points and $W=\partial_\psi$ is spacelike it follows that $\alpha_0>0$ and $\beta_0>0$ respectively. In fact, using the explicit expressions for these constants it turns out that these inequalities are equivalent to the single condition $\alpha_0^2\beta_0>0$  which reads~\cite{breunholder_moduli_2019}
\begin{equation}
	-h^2_{-1}m^2_{-1}-3h_{-1}k_{-1}l_{-1}m_{-1}+h_{-1}l^3_{-1}-2k_{-1}^3m_{-1}+\frac{3}{4}k_{-1}^2l_{-1}^2>0 \; . \label{ineq_horizon}
\end{equation}
In fact this is not only necessary, but also sufficient for the existence of a smooth horizon away from the axes $\theta =0,\pi$. 
This is revealed by performing the coordinate change $(t, \psi, r, \theta, \phi) \to (v, \psi', r, \theta, \phi')$ defined by 
\begin{align}
	\td t&= \td v+\left(\frac{A_0}{r^2}+\frac{A_1}{r}\right)\td r \; , &\td\psi&=\td\psi'+\frac{B_0}{r}\td r-\chi_0\td\phi' \; , &\td\phi&=\td\phi' \; ,
	\label{eq_coordTf}
\end{align}
with 
\begin{align}
	&A_0^2=\beta_0\alpha_0^2 \; , \qquad\qquad B_0 = -\frac{A_0\gamma_0}{\beta_0} \; , \nonumber \\
	A_1 = \frac{\alpha_0\beta_0}{2A_0}&\left(B_0^2\beta_1+2B_0A_0\gamma_1+\alpha_1-\frac{2h_{-1}(h_0\alpha_0-h_{-1}\alpha_1)}{\alpha_0^3}A_0^2\right) \; .   \label{eq_A0B0A1}
\end{align}
We emphasise that the single condition (\ref{ineq_horizon}) (which is equivalent to $\alpha_0>0, \beta_0>0$) is sufficient for this coordinate change to exist.  This coordinate change is  the same as in the case with biaxial symmetry~\cite{breunholder_moduli_2019}.  The metric in the new chart can be written as
\begin{align}
g  =&- f^2 (\td v+ \hat{\omega}')^2 + 2 g_{t\psi}  (\td v + \hat{\omega}') ( \td \psi'+h_{-1} \cos\theta \td \phi'+\tilde{\chi}')  \nonumber \\
&+ 2 g_{vr}( \td v+\hat\omega')\td r+ g_{rr} \td r^2 + 2 g_{\psi' r} ( \td \psi'+h_{-1} \cos\theta \td \phi'+\tilde{\chi}')  \td r \label{eq_g_nearH} \\
&+  g_{\psi\psi} ( \td \psi'+h_{-1} \cos\theta \td \phi'+\tilde{\chi}') ^2+ r^2 N^{-1} (\td\theta^2 + \sin^2\theta \td \phi'^2) \; , \nonumber 
\end{align}
where $\tilde{\chi}'$, $\hat\omega'$ denote the  1-forms $\tilde\chi$, $\hat\omega$ with $\phi$ replaced by $\phi'$ and
\begin{align}
g_{vr} &=   - f^2 \left(\frac{A_0}{r^2}+\frac{A_1}{r} +\frac{\omega_\psi B_0}{r}\right), \qquad g_{\psi'r}= \frac{N B_0}{f^2r} + g_{t\psi} \left(\frac{A_0}{r^2}+\frac{A_1}{r} +\frac{\omega_\psi B_0}{r}\right) \; , \\
g_{rr}&= \frac{1}{N}  +\frac{N B_0^2}{f^2 r^2}- f^2 \left(\frac{A_0}{r^2}+\frac{A_1}{r} +\frac{\omega_\psi B_0}{r}\right)^2  \;.
\end{align}
Using the expansion of the invariants (\ref{eq_inv1_nearH}), (\ref{eq_inv2_nearH}) and the form of the coordinate change (\ref{eq_A0B0A1}), it follows that 
\begin{equation}
g_{vr} =\pm  \frac{1}{\sqrt{\beta_0}}+ \mathcal O (r), \qquad g_{rr}= \mathcal O (1), \qquad g_{\psi'r} = \mathcal O(1) \; ,
\end{equation}
where the error terms are analytic in $r$ and smooth on $S^2$. Therefore, we deduce from (\ref{eq_inv1_nearH}), (\ref{eq_inv2_nearH}), (\ref{eq_chi_nearH_polar}), (\ref{eq_omegahat_nearH}), (\ref{eq_omegasing}), that the spacetime metric  (\ref{eq_g_nearH}) and its inverse are analytic in $r$ at $r=0$ and can be analytically extended to $r\leq 0$. The hypersurface $r=0$ is a Killing horizon of $V=\partial_v$ and the metric induced on the cross-section of the horizon $v=\text{const}, r=0$ is
\begin{equation}
	\beta_0(\td \psi'+ h_{-1} \cos\theta \td \phi')^2+\alpha_0 (\td\theta^2 + \sin^2\theta \td \phi'^2)  \; .\label{eq_Hmetric}
\end{equation} Furthermore, it can be shown that the Maxwell field is also analytic at $r=0$ and the near-horizon limit of the solution takes the same form as in the biaxisymmetric case.

We will now turn to analysing regularity at the axes $\theta=0, \pi$ including where these intersect the horizon at $r=0$.  By inspecting the horizon metric (\ref{eq_Hmetric}) it is clear  the vector fields that vanish at the axes  are
\begin{equation}
K_\pm = \partial_{\phi'}  \mp h_{-1} \partial_{\psi'}  \; ,
\end{equation}
in particular,  $K_+=0$ at $\theta=0$ and $K_-=0$ at $\theta=\pi$.  Therefore, smoothness of the spacetime metric at the axis $\theta=0, \pi$ requires
\begin{equation}
	g_{\mu \rho} K_\pm^\rho =0   \quad \text{ at } \theta = 0, \pi, \text{ respectively.}\label{eq_axisregularity} 
\end{equation}
For $\mu= v$  this condition is equivalent to 
\begin{equation}
0= f^2( \hat{\omega}_\phi+\omega_\psi( \chi_\phi- \chi_0 \mp h_{-1}) )= f^2 (\pm \omega_{-1} +\omega_0)   \quad \text{ at } \theta = 0, \pi, \text{ respectively, }
\end{equation}
where  we have used (\ref{eq_coordTf}) and the second equality follows from (\ref{eq_chi_nearH}), (\ref{eq_omegahat_nearH}) and the fact that $\tilde{\chi}_\phi=\tilde{\omega}_\phi=0$ at $\theta=0, \pi$ for any $r$ (this is because  $\partial_\phi=0$ at $\theta=0, \pi$ and $\tilde\chi_\phi= \iota_{\partial_\phi} \tilde\chi$ where $\tilde\chi$ is a smooth 1-form on $S^2$ and similarly for $\tilde\omega_\phi$). Therefore, since $\omega_{-1}, \omega_0$ are constants the condition $g_{v \rho} K_\pm^\rho=0$ at $\theta=0, \pi$ is equivalent to 
\begin{equation}
\omega_{-1}=\omega_0=0  \; .  \label{eq_Haxisreg}
\end{equation} 
It can be similarly shown that (\ref{eq_Haxisreg}) are sufficient for (\ref{eq_axisregularity}) to hold for all other components $\mu$. Therefore, (\ref{eq_axisregularity}) is equivalent to (\ref{eq_Haxisreg}).

To verify smoothness at the axes we also need to check that all higher order terms in the expansion around $\theta=0$ and $\theta=\pi$ are suitably smooth. We will return to this point below. First it is convenient to perform a global analysis of the horizon geometry in order to deduce the possible horizon topologies. In fact, we will now show that asymptotic flatness imposes global constraints that restrict the horizon topology as in the biaxisymmetric case.

\begin{lemma} 
\label{lemma_Htop}For a multi-centred solution as given in Theorem \ref{theorem_harmonicfunctions} the topology of cross-sections of each connected component of the event horizon  is $S^3, S^2 \times S^1$ or a lens space $L(p, 1)$.
\end{lemma}

\begin{proof} The analysis splits into two cases depending on if $h_{-1}$ vanishes. First suppose $h_{-1}\neq 0$. It is convenient to define coordinates adapted to the vectors that vanish on the axes, that is,  $K_\pm =\partial_{\phi^\pm}$ where
\begin{equation}
\phi^\pm:=\frac{1}{2} \left(\phi' \mp \frac{\psi'}{h_{-1}} \right) \; .
\label{eq_coordTf2}
\end{equation}
By asymptotic flatness the coordinates $\psi$, $\phi$ satisfy (\ref{eq_lattice_original}), which is equivalent to $(\theta, \tilde\phi, \tilde\psi)$, 
defined by (\ref{eq_eulerinfinity}),  being Euler angles on the $S^3$ at spatial infinity. Using the coordinate change (\ref{eq_coordTf}) this is 
equivalent to the identifications on the $\phi^\pm$ coordinates
\begin{align}
	P: \,\,\,(\phi^+, \phi^-)&\sim\left(\phi^+-2\pi\frac{1}{h_{-1}}, \phi^-+2\pi\frac{1}{h_{-1}}\right) \; ,\nonumber\\
	R:\,\,\,(\phi^+, \phi^-)&\sim\left(\phi^++2\pi\frac{h_{-1}-\chi_0}{2h_{-1}}, \phi^-+2\pi\frac{h_{-1}+\chi_0}{2h_{-1}}\right) \; .   \label{eq_PR_phipm}
\end{align}
Recall that the identification lattice of $L(p,q)$ is generated by
\begin{align}
	Q: \,\,\,(\phi^+, \phi^-)&\sim\left(\phi^++2\pi\frac{1}{p}, \phi^-+2\pi\frac{q}{p}\right) \; ,\nonumber\\
	S:\,\,\,(\phi^+, \phi^-)&\sim\left(\phi^+, \phi^-+2\pi\right) \; .     \label{eq_QS_phipm}
\end{align}
The requirement that (\ref{eq_Hmetric}) extends to a smooth metric on a compact manifold means that the lattice generated by $\{P,R\}$ must be 
the same as the one generated by $\{Q, S\}$. It can be shown that this  condition is  equivalent to $h_{-1} = \pm p$, $\chi_0\equiv h_{-1} \mod 2$ 
and $q\equiv-1 \mod p$. In particular, notice that $h_{-1}$ and $\chi_0$ are required to be integers with the same parity. Therefore, the
allowed topologies are $L(\pm h_{-1}, -1)\cong L(|h_{-1}|,  1)$ or $S^3$ for $h_{-1}=\pm1$.

In the case  $h_{-1}=0$ the horizon geometry (\ref{eq_Hmetric}) extends to a smooth metric on a compact manifold if and only if $\psi'$ and $\phi'$ are 
independently periodic and the periodicity of $\phi'$ is $2\pi$. The horizon topology in this case is  $S^1\times S^2$.  Again, by asymptotic 
flatness $\psi$ and $\phi$ are independently periodic with periodicities $4\pi$, $2\pi$, respectively (\ref{eq_lattice_original}), which 
in terms of the coordinates (\ref{eq_coordTf}) is equivalent to 
\begin{align}
P: (\psi', \phi') &\sim (\psi'+4\pi, \phi')  \; , \nonumber \\
R: (\psi', \phi') &\sim  ( \psi'+ 2\pi\chi_0 , \phi' +2\pi) \; . 
\end{align}
Thus,  in order for $\psi'$ and $\phi'$ to be independently periodic, $\chi_0$ must be an even integer. 
\end{proof}

Now we have the global geometry of the horizon, we can calculate its area, which yields
\begin{equation}
	A_H=16\pi^2\sqrt{-h^2_{-1}m^2_{-1}-3h_{-1}k_{-1}l_{-1}m_{-1}+h_{-1}l^3_{-1}-2k_{-1}^3m_{-1}+\frac{3}{4}k_{-1}^2l_{-1}^2} \; .
\end{equation}
The quantity inside the square-root is always positive due to (\ref{ineq_horizon}).

We now return to verifying smoothness at the axes $\theta=0, \pi$.  For definiteness, we focus on the $\theta =0$ axis, although the argument for $\theta=\pi$ is identical. First consider the case $h_{-1}\neq 0$ and introduce coordinates $\phi^\pm$ adapted to the vectors $K_\pm$ that vanish on the axes as in (\ref{eq_coordTf2}).  In particular, $K_+=\partial_{\phi^+}$ vanishes at $\theta=0$ and from the periodicities (\ref{eq_QS_phipm}) it is easy to see that $\phi^+$ must be $2\pi$-periodic for fixed $\phi^-$.   Now,  inverting (\ref{eq_coordTf2}) we have $\phi'= \phi^++\phi^-$ and $\psi' = h_{-1}( \phi^- - \phi^+)$ which allows us to easily write the full metric (\ref{eq_g_nearH}) in terms of $\phi^\pm$.  In particular, the explicit dependence of the metric components on $\phi^\pm$ comes from the dependence on $\phi'$ of the higher order terms in $r$, which in turn arises from the $\phi$-dependence of the spherical harmonics
\begin{align}
	Y_l^m(\theta, \phi) =  c_{lm} P_l^m(\cos\theta)e^{im\phi'} = c_{lm} P_l^m(\cos\theta)e^{im(\phi^++\phi^-)}  = Y^m_l(\theta, \phi^+) e^{im \phi^-}  \; .
\end{align}
Therefore, since $Y^m_l$ are smooth on $S^2$ they will be smooth at the pole $\theta=0$ of the sphere parameterised by $(\theta, \phi^+)$  (recall $\phi^+$ is $2\pi$-periodic for fixed $\phi^-$). Furthermore, the 1-form
\begin{equation}
 \td \psi'+h_{-1} \cos\theta \td \phi' = h_{-1} \left(  (1+\cos\theta)  \td \phi^- + (-1+\cos\theta) \td \phi^+ \right)
\end{equation}
is manifestly smooth at the axis $\theta=0$.  Hence, from the $r$-expansion of the functions (\ref{eq_inv1_nearH}), (\ref{eq_inv2_nearH})  and  the 1-forms (\ref{eq_chi_nearH_polar}), (\ref{eq_omegahat_nearH}), (\ref{eq_omegasing}) together with (\ref{eq_Haxisreg}), it now follows that the full metric is smooth at the axis $\theta=0$, at least on a neighbourhood of the horizon $r=0$ (the domain of convergence of the expansion of the harmonic functions into spherical harmonics). This establishes that the metric is smooth at the axis including up to where it intersects the horizon if and only if (\ref{eq_Haxisreg}) holds.   It can be easily seen that the Maxwell field is also smooth at the axes including at the intersection with the horizon (this essentially follows from the fact that the only new type of term is from $\td \xi= k_{-1} \sin\theta \td \theta \wedge \td \phi+\td \tilde\xi$ using (\ref{eq_xi_nearH})).

Finally, let us consider the case $h_{-1}=0$. This is simpler since as observed above $\psi', \phi'$ are independently periodic and correspond 
to the angle coordinates on the $S^1$ and $S^2$ factors of the horizon respectively. Therefore, the metric (\ref{eq_g_nearH}) is smooth on the $S^2$ 
parameterised by  $(\theta, \phi')$ in a neighbourhood of $r=0$ if and only if (\ref{eq_Haxisreg}), since the dependence of the higher order terms 
in $r$ on $(\theta, \phi')$ is through the spherical harmonics.

To summarise, we have shown that the solution is smooth at a horizon, including where it intersects the axis, if and only if the coefficients in the expansions (\ref{eq_harmonicexpansion}), (\ref{eq_chi_nearH}) etc, satisfy (\ref{ineq_horizon}), (\ref{eq_Haxisreg}),
\begin{equation}
h_{-1} \in \mathbb{Z}, \qquad   \chi_0+h_{-1} \in 2\mathbb{Z}  \; ,   \label{eq_hchicond_hor}
\end{equation}
and the horizon topology is $S^1 \times S^2$ if $h_{-1}=0$, $S^3$ if $h_{-1}= \pm 1$ and $L( |h_{-1}|, 1)$ otherwise.

\subsection{Smoothness at the fixed points}\label{ssec_fixed}

In this section we will derive the necessary and sufficient criteria for the solution to be smooth around a fixed point of the axial Killing field $W$.  As shown in Lemma \ref{lemma_harmonic_fixedpt} a fixed point corresponds to a simple pole of the harmonic functions $H, K, L, M$, so without loss of generality, we will take this to be at the origin. We then 
 expand these functions in the  same way that we did for a horizon in equation (\ref{eq_harmonicexpansion}).  It follows that the 1-forms $\chi, \xi, \omega$ can also be written in the same form as for a horizon, namely these are given by equations (\ref{eq_chi_nearH}-\ref{eq_omegasing}).
 
 As argued earlier, at a fixed point we must have $f\neq 0$ and therefore the base metric is well-defined at least on a neighbourhood of such points.  It follows that the base metric $h$, the 1-form $\omega$ and the function $f$ must be smooth at the fixed points.
Let us first consider the base metric near a fixed point at $r=0$.  It is convenient to introduce coordinates 
\begin{align}
	r&=\frac{1}{4}R^2\; , &\psi'&=\psi+\chi_0\phi\; , &\phi'&=h_{-1}\phi \; ,
\end{align}
so that the base metric takes the form
\begin{equation}
	h=G\left(\td R^2+\frac{1}{4}R^2\left[\td\theta^2+  \sin^2\theta \frac{\td\phi'^2}{h_{-1}^2}+\frac{1}{G^2}(\td\psi'+ \cos\theta \td\phi'+\tilde\chi)^2\right]\right),
	\label{eq_GHmetric_nearFP}
\end{equation}
where we have defined $G:=rH=h_{-1}+\mathcal{O}(R^2)$ and $\tilde\chi = \mathcal O (R^4)$ is defined by the regular part $\tilde H$ of the harmonic function which can be written as (\ref{eq_chi_nearH_polar}).  In order to avoid a curvature singularity of the base metric as $R\to 0$ we must require
\begin{equation}
	h_{-1}=\pm 1 \;,  \label{eq_hcond} 
\end{equation}
in which case (\ref{eq_GHmetric_nearFP}) approaches a locally flat metric on $\pm \mathbb{R}^4$.  Furthermore, to avoid any conical singularities 
at $R=0$ the angles $(\theta,\phi', \psi')$ must be Euler angles on $S^3$ with identifications as in (\ref{eq_lattice}), in which case the base 
metric approaches the flat metric on $\pm \mathbb{R}^4$ (note $G=1$ and $\tilde\chi=0$ corresponds to the flat metric on $\mathbb{R}^4$).  On the 
other hand, by asymptotic flatness, $(\psi, \phi)$ must obey identifications (\ref{eq_lattice_original}), which implies that $(\psi', \phi')$ have 
the correct identification lattice (\ref{eq_lattice}) if and only if $\chi_0$ is an odd integer.

To verify the base metric is smooth at the fixed point  requires control of the higher order terms as $r\to 0$.  In particular, we must check that the metric is smooth at the origin of $\mathbb{R}^4$. For this,  it is easiest to use  cartesian coordinates $(u_I)_{I=1,2,3,4}$ on $\mathbb{R}^4$, which are introduced as follows. We first define the double-polar coordinates\footnote{
	$\phi^\pm$ in this section are defined differently to those in Section \ref{ssec_horizonRegularity}.
}
\begin{align}
	X_+&=R\cos\frac{\theta}{2}\; , & X_-&=R\sin\frac{\theta}{2}\; , & \phi^{\pm}&=\frac{1}{2}(\psi'\pm\phi') \; ,\label{EulerToR4coords}
\end{align}
in terms of which the leading order metric is 
\begin{equation}
h \sim \pm \left( \td X_+^2 + X_+^2 (\td \phi^+)^2 +  \td X_-^2 + X_-^2 (\td \phi^-)^2 \right)   \; .
\end{equation}
Hence, cartesian coordinates on $\mathbb{R}^4$ are given by 
\begin{align}
	u_1&=X_+\cos\phi^+, &u_2&=X_+\sin\phi^+,& u_3&=X_-\cos\phi^-, & u_4&=X_-\sin\phi^- \; .
\end{align}
It is helpful for the analysis to note that the following $\mathbb{R}^3$-functions are smooth on $\mathbb{R}^4$:
\begin{align}
	r = \frac{1}{4}(u_1^2+u_2^2+u_3^2+u_4^2)\;, &\qquad r\cos\theta = \frac{1}{4}(u_1^2+u_2^2-u_3^2-u_4^2)\;, \nonumber\\
	r\sin\theta \exp(i\phi) &= \frac{1}{2}(u_1+ ih_{-1}u_2)(u_3- ih_{-1}u_4)  \; .\label{eq_regularfns}
\end{align}
Furthermore, noting that the $\mathbb{R}^3$ cartesian coordinates satisfy $x^1+ i x^2= r\sin \theta e^{i\phi}$ and $x^3= r \cos \theta$, 
it immediately follows that $x^i$ are smooth functions on $\mathbb{R}^4$, and more generally any smooth function $f(x)$ on $\mathbb{R}^3$ 
is also a smooth function on $\mathbb{R}^4$. In particular, any regular harmonic function on $\mathbb{R}^3$ is automatically smooth on 
$\mathbb{R}^4$. Therefore, the function $G=h_{-1}+ r (h_0+\tilde H)$ defined above is smooth on $\mathbb{R}^4$. It is also helpful to note 
that the 1-form
\begin{equation}
	r(\td \psi'+\cos\theta d\phi')=\frac{1}{2}(X_+^2\td\phi^++X_-^2\td\phi^-)
\end{equation}
is smooth on $\mathbb{R}^4$ (convert to cartesian coordinates $u_I$), and the 1-form $\tilde{\chi}$ which is defined by 
$\star_3 \td \tilde\chi= \td \tilde H$ where $\tilde{H}$ is the regular part of the harmonic function in (\ref{eq_harmonicexpansion}) 
must also be smooth on $\mathbb{R}^4$ (up to a gauge transformation). Now, we can write the base metric (\ref{eq_GHmetric_nearFP}) as 
\begin{align}
h 
&= G \td u_I \td u_I - \frac{(h_{-1}+G) (h_0+ \tilde H)}{G}  r^2(\td\psi'+ \cos\theta \td\phi')^2  
+ \frac{2}{G} r(\td\psi'+ \cos\theta \td\phi') \tilde \chi+ \frac{r }{G} \tilde \chi^2 \; ,
\end{align}
where  we used (\ref{eq_harmonicexpansion}) and the definition $G= r H$.  It is now manifest that the base metric is smooth at the origin 
of $\mathbb{R}^4$, since all functions and 1-forms that we have written it in terms of are smooth by the above comments.

We now turn to smoothness of the function $f$ on the base.  Using (\ref{eq_harmonicexpansion}), we have the useful identity
\begin{equation}
	(rH)^{-1}=h_{-1}-r H_0+r^2 G_1 \; ,
	\label{Hinvtrick}
\end{equation}
where $H_0 := h_0 + r \tilde H$ and $G_1:= \tilde H_0^2/(h_{-1}+r H_0)$ is a smooth function on $\mathbb{R}^4$. Then, using (\ref{eq_f}) and (\ref{Hinvtrick}), it is easy to see that
\begin{equation}
f^{-1} = \frac{l_{-1}+ k_{-1}^2 h_{-1}}{r} +  l_0- h_0 k_{-1}^2+ 2 h_{-1} k_{-1} k_0+ \mathcal{O}(r) \; ,
\end{equation}
where the error terms are smooth on $\mathbb{R}^4$.  Therefore, since we must have $f\neq 0$ at a fixed point, smoothness requires
\begin{align}
&l_{-1} = - h_{-1} k_{-1}^2  \; ,  \label{eq_fcond1} 
\\
&h_{-1}( l_0- h_0 k_{-1}^2+ 2 h_{-1} k_{-1} k_0) >0  \; ,\label{eq_fcond2}
\end{align} 
where the sign of the inequality is chosen to ensure the spacetime metric has the correct signature at the fixed point.
We deduce that these are the necessary and sufficient conditions for $f$ to be smooth at the fixed point.

Let us now look at the 1-form $\omega$ which decomposes as (\ref{eq_omega}).   The invariant $g(V, W)= - f^2 \omega_\psi$, together with the 
fact that $f\neq 0$ at the fixed point,  implies that $\omega_\psi$ must be a smooth function that vanishes at the fixed point.   By expanding 
(\ref{eq_omegapsi})   near $r=0$ and using (\ref{eq_fcond1}) one finds
\begin{equation}
\omega_\psi= \frac{ m_{-1}- \tfrac{1}{2} k_{-1}^3}{r} + \mathcal O (1) \; ,
\end{equation}
where the error terms are smooth.  Thus, in particular we must require
\begin{equation}
m_{-1} = \tfrac{1}{2} k_{-1}^3   \; .  \label{eq_ompsicond1}
\end{equation}
Then, using (\ref{Hinvtrick}) together with (\ref{eq_fcond1}), (\ref{eq_ompsicond1}),  we can write (\ref{eq_omegapsi}) as
\begin{equation}
\omega_\psi=  h_{-1} (- \omega_{-1}+ F) +r  \tilde G_1 \; ,
\end{equation}
where $\omega_{-1}$ is the  constant defined in (\ref{eq_omegam1}),  $F$ is the harmonic function defined in (\ref{eq_doms}), and $\tilde G_1$ is 
a smooth function on $\mathbb{R}^4$.   Therefore, since $F$ vanishes at $r=0$, the vanishing of $\omega_\psi$ at $r=0$ occurs iff
\begin{equation}
\omega_{-1}=0   \; .  \label{eq_omegam1cond}
\end{equation}
Thus,  together with the form for $\hat\omega$ given in (\ref{eq_omegahat_nearH}), we can write
\begin{equation}
\omega=  \underbrace{h_{-1} F ( \td \psi'+ \cos \theta \td \phi')+ \hat \omega_{\text{sing}}  }_{=: \alpha} +\omega_0 h_{-1} \td \phi'+ \underbrace{ \tilde{G}_1 r(\td \psi'+\cos\theta \td\phi')+ \omega_\psi \tilde \chi+ \hat \omega_{\text{reg}}}_{\text{smooth on $\mathbb{R}^4$}}  \; , \label{eq_alphadef}
\end{equation}
where the fact that the last three terms are smooth immediately follows from our above analysis.  Therefore, smoothness of $\omega$ reduces to that 
of $\alpha$.  From (\ref{eq_omegasing}) and the fact that $F$ is harmonic it automatically follows that both terms of $\alpha$ are smooth at $r>0$, 
but smoothness at $r=0$ remains to be checked. To this end, a short computation reveals that
\begin{align}
\star_\delta \td \alpha  &= - \frac{h_{-1}}{r} \star_3 \td F+ r (\td\psi'+\cos \theta \td \phi') \wedge (  h_{-1}F  \star_3 \td (\cos\theta \td \phi') + \star_3 \td \hat \omega_{\text{sing}}  ) \\
&= - \frac{h_{-1}}{r} \star_3 \td F +  (\td\psi'+\cos \theta \td \phi') \wedge \td F  \; \nonumber ,
\end{align}
where $\star_\delta$ denotes the Hodge dual with respect to the flat metric $\delta = r (\td \psi' + \cos\theta \td \phi')^2+ r^{-1} \td x^i \td x^i$  on 
$\mathbb{R}^4$ with orientation $\epsilon_{\psi' 123}=1$, and the second line is obtained using (\ref{eq_doms}). Then, using the fact that $F$ is harmonic, 
we deduce
\begin{equation}
\td \star_\delta \td \alpha=0 \; .
\end{equation}
Therefore, $\alpha$ must be  a smooth 1-form on $\mathbb{R}^4$ (up to a gauge transformation). In Appendix \ref{app_omega} we show that $\alpha$ is a smooth 
1-form on $\mathbb{R}^4$ by an explicit calculation.  We deduce that  $\omega$ is smooth on $\mathbb{R}^4$ if and only if the constant $\omega_0=0$.

We finish this section with the analysis of the Maxwell field, which takes the form
\begin{equation}
	F=\frac{\sqrt{3}}{2}\td\left[f(\td t+\omega)-\frac{K}{H}(\td\psi'+\cos\theta \td\phi')+h_{-1}k_{-1}\cos\theta \td\phi'-\frac{K}{H}\tilde\chi+\tilde\xi\right] \; ,
\end{equation}
where we have used (\ref{eq_Maxwell}) and (\ref{eq_xi_nearH}).  The 1-form $\tilde \xi$ is smooth on $\mathbb{R}^4$ by the same argument for $\tilde \chi$, and by the above analysis we also know that $f (\td t+\omega)$ and $(K/H)\tilde \chi$ are smooth. Using 
(\ref{Hinvtrick}), the middle two terms in the gauge field can be rewritten as
\begin{align}
	&-\frac{K}{H}(\td\psi'+\cos\theta \td\phi')+h_{-1}k_{-1}\cos\theta \td\phi' = -k_{-1}h_{-1}\td\psi'  \\ &\qquad \qquad \qquad  \qquad + \left( k_{-1}(H_0-r G_1)-  (k_0 +  \tilde K) (r H)^{-1} \right) r (\td\psi'+\cos\theta \td\phi') \; , \nonumber
\end{align}
where the terms on the second line are manifestly smooth. Thus, the only non-smooth term is pure gauge, so the Maxwell field is indeed smooth on $\mathbb{R}^4$. This concludes our analysis 
at a fixed point of $W$.

To summarise, we have shown that the solution at a fixed point $r=0$ is smooth if and only if the parameters satisfy (\ref{eq_hcond}), (\ref{eq_fcond1}), (\ref{eq_fcond2}), (\ref{eq_ompsicond1}), (\ref{eq_omegam1cond}), $\omega_0=0$ and 
\begin{equation}
\chi_0 \in 2\mathbb{Z}+1\; .   \label{eq_chi0om0cond}
\end{equation}
The spacetime in a neighbourhood of a such point is then diffeomorphic to $\mathbb{R}^5$.

\section{General black hole and soliton solutions}\label{sec_Summary}

\subsection{Classification theorem}
We are now ready to deduce the main result of this paper which provides a  classification of black hole and solitons spacetimes that satisfy our assumptions.

\begin{theorem}\label{thm_classification}
	Any asymptotically flat, supersymmetric black hole or soliton solution $(\mathcal{M}, g, F)$ of $D=5$ minimal supergravity, with an axial symmetry, satisfying Assumption \ref{assumption1} and \ref{assumption2} must 
have a Gibbons-Hawking base (wherever $f \neq 0$) with  `multi-centred' harmonic functions
	\begin{align}
	H =  \sum_{i=1}^N \frac{h_i}{r_i} \; , && K =  \sum_{i=1}^N \frac{k_i}{r_i} \; , && L = 1+ \sum_{i=1}^N \frac{l_i}{r_i} \; , &&M =m+ \sum_{i=1}^N  \frac{m_i}{r_i}  \; , \label{eq_thm_harmonic}
\end{align}
where  $r_i := |\bm x-\bm a_i|$, and $\bm a_i=(x_i, y_i, z_i) \in \mathbb{R}^3$ correspond to fixed points of the axial Killing field or the connected  components of the horizon, 
and the 1-forms can be written as
\begin{align}
\chi &= \sum_{i=1}^N \left(\chi_0^i+ \frac{h_i(z-z_i)}{r_i}   \right) \td \phi_i \; , \qquad \xi =-  \sum_{i=1}^N   \frac{k_i(z-z_i)}{r_i}  \td \phi_i \; , \label{chixi_sol}\\ 
		\hat\omega & = \sum_{\substack{i,j=1 \\ i \neq j}}^N \left(h_im_j+\frac{3}{2}k_il_j\right)\beta_{ij}\; ,\label{omega_sol_expl}
	\end{align}
	where 
	\begin{align}
	\td \phi_i &:= \frac{(x-x_i)\td y- (y-y_i)\td x}{(x-x_i)^2+ (y-y_i)^2} \; , \\
		\beta_{ij}&:= \left(\frac{(\bm x-\bm a_i)\cdot (\bm a_i-\bm a_j)}{|\bm a_i-\bm a_j|r_i}-\frac{(\bm x-\bm a_j)\cdot (\bm a_i-\bm a_j)}{|\bm a_i-\bm a_j|r_j}-\frac{(\bm x-\bm a_i)\cdot (\bm x-\bm a_j)}{r_ir_j}+1\right)\nonumber \\
		&\qquad\qquad\times\frac{((\bm a_i-\bm a_j)\times(\bm x-\bm a_j))\cdot \td\bm x}{|(\bm a_i-\bm a_j)\times(\bm x-\bm a_j)|^2} \; . \label{eq_beta_sol}
	\end{align}
The parameters $h_i, k_i, l_i, m_i$  must satisfy 
\begin{align}
	&\sum_{i=1}^N h_i = 1 \; , \label{eq_thm_AF_h}\\
	&m = -\frac{3}{2}\sum_{i=1}^N k_i \; , \label{eq_thm_AF_m}
\end{align}
and for each $i=1, \dots, N$, \begin{align}
		h_i m +\frac{3}{2}k_i + \sum_{\substack{j=1 \\j\neq i}}^N\frac{h_i m_j -m_i h_j +\frac{3}{2}(k_il_j-k_jl_i)}{|\bm a_i -\bm a_j|}=0 \; . \label{eq_thm_ctr}
	\end{align}
	Moreover, if $\bm a_i$ is a fixed point  $\chi_0^i \in 2\mathbb{Z}+1$,
	\begin{align}
		h_i = \pm 1 \; , && l_i = -h_ik_i^2 \; , && m_i = \frac{1}{2}k_i^3 \; , \label{eq_thm_fix}
	\end{align}
	\begin{align}
		h_i + \sum_{\substack{j=1 \\j\neq i}}^N\frac{2k_ik_j-h_i(h_jk_i^2-l_j)}{|\bm a_i-\bm a_j|}>0 \; ,  \label{eq_thm_fh}
	\end{align}
	whereas if $\bm a_i$ is a horizon component $h_i \in \mathbb{Z}$, $\chi^i_0 +h_i  \in 2\mathbb{Z}$  and
	\begin{align}
		-h_i^2m_i^2-3h_ik_il_im_i+h_il_i^3-2k_i^2m_i+\frac{3}{4}k_i^2l_i^2>0 \; , \label{eq_thm_hor}
	\end{align}
	and the horizon topology is $S^1\times S^2$ if $h_i = 0$, $S^3$ if $h_i = \pm 1$ and a lens space
	$L(h_i, 1)$ otherwise.	Finally, the harmonic functions must satisfy
	\begin{align}
		K^2+HL >0 \label{eq_thm_N}
	\end{align}
 for all $\bm x \in \mathbb{R}^3\setminus\{\bm a_1, \dots, \bm a_N\}$.
 \end{theorem}

\begin{proof}
 Theorem \ref{theorem_harmonicfunctions} shows that the solution must have a Gibbons-Hawking base with harmonic functions of multi-centred form (\ref{eq_thm_harmonic}).   The 1-forms are determined as follows.  First note that $\beta_i:= (z-z_i) \td \phi_i/r_i $ obeys $\star_3 \td \beta_i = \td (1/r_i)$. Therefore, solving (\ref{eq_chieq}), (\ref{eq_xi}) immediately gives that $\chi, \xi$ can be written in the claimed form, where $\chi_0^i$ are integration constants introduced for convenience.  To determine $\hat\omega$ it is convenient to follow~\cite{bena_black_2008} (see also~\cite{dunajski_einstein-maxwell_2007, tomizawa_supersymmetric_2016}) and define the 1-forms (\ref{eq_beta_sol}) which are a solution to 
 \begin{align}
		\star_3\td\beta_{ij}&=\frac{1}{r_i}\td\left(\frac{1}{r_j}\right)-\frac{1}{r_j}\td\left(\frac{1}{r_i}\right)+\frac{1}{r_{ij}}\td\left(\frac{1}{r_i}-\frac{1}{r_j}\right) \; , \label{eq_beta}
	\end{align}
	 where $r_{ij} := | \bm a_i - \bm a_j|$. 
	Note that, in contrast to $\beta_i$, the 1-forms $\beta_{ij}$ are free of Dirac string singularities, indeed $\beta_{ij}$ are smooth 1-forms on $\mathbb{R}^3$ except at the corresponding centres $\bm a_i, \bm a_j$.
	Then, we can solve (\ref{eq_omegahat}) and write 
	\begin{equation}
	\hat\omega = \sum_{\substack{i,j=1 \\ i\neq j}}^N \left(h_im_j+\frac{3}{2}k_il_j\right)\beta_{ij}   + \sum_{i=1}^N  \omega_{-1}^{i} \beta_i   \; , \label{eq_omega_beta}
	\end{equation}
	where we have defined the constants
	\begin{equation}
	\omega_{-1}^i:= -h_i m -\frac{3}{2}k_i - \sum_{\substack{j=1 \\j\neq i}}^N\frac{h_i m_j -m_i h_j +\frac{3}{2}(k_il_j-k_jl_i)}{r_{ij}}  \; .  \label{eq_omm1i}
	\end{equation} 
	The functional form of the solution is now fully fixed.   
 
 Next, by expanding the harmonic functions around a centre $\bm x= \bm a_i$  they take the form (\ref{eq_harmonicexpansion})  where the coefficients are given by 
 \begin{align}
 h_{-1}&= h_i \; , \qquad h_0 = \sum_{\substack{j=1 \\ j \neq i}}^N \frac{h_j}{r_{ij}} \; , \qquad  k_{-1}= k_i \; , \qquad k_0 = \sum_{\substack{j=1 \\ j \neq i}}^N  \frac{k_j}{r_{ij}} \; , \\
  l_{-1} &= l_i \; , \qquad l_0 =1+ \sum_{\substack{j=1 \\ j \neq i}}^N  \frac{h_j}{r_{ij}} \; , \qquad   m_{-1}= m_i \; , \qquad m_0 =m+ \sum_{\substack{j=1 \\ j \neq i}}^N  \frac{m_j}{r_{ij}}   \; .
 \end{align}
Necessary and sufficient conditions for smoothness at the event horizon and fixed points were determined in Section \ref{ssec_horizonRegularity} and 
\ref{ssec_fixed} in terms of these coefficients.  For smoothness at  fixed points, we must require (\ref{eq_hcond}), (\ref{eq_fcond1}), (\ref{eq_fcond2}), 
(\ref{eq_ompsicond1})  and (\ref{eq_chi0om0cond}), which upon use of the above coefficients give (\ref{eq_thm_fix}), (\ref{eq_thm_fh}) and that 
$\chi_0^i$ must be an odd integer.  For smoothness at the horizon, we require   (\ref{ineq_horizon}), (\ref{eq_hchicond_hor}), 
which upon use of the above coefficients gives (\ref{eq_thm_hor}), $h_i \in \mathbb{Z}$ and $\chi_0^i + h_i$ must be an even integer. 
The horizon topology is determined in Lemma \ref{lemma_Htop}. Notice in particular that $h_i$ and $\chi_0^i$ have the same parity at each centre.   

The remaining conditions for smoothness at both fixed points and horizons are $\omega_{-1}=0$ and 
$\omega_0=0$, which are defined by (\ref{eq_omegahat_nearH}) and (\ref{eq_omegam1}). Upon using the above coefficients $\omega_{-1}=0$ is equivalent to 
the constraint (\ref{eq_thm_ctr}) for each $i=1, \dots, N$, which in fact is the same as $\omega_{-1}^i=0$ (see definition (\ref{eq_omm1i})) so that the 
1-form $\hat \omega$ now has the claimed form (\ref{omega_sol_expl}).  Furthermore, from the explicit form of $\beta_{ij}$ one can check that $\omega_0=0$ 
for each centre (this is due to the aforementioned fact that $\beta_{ij}$ has no string singularities).  Therefore, all necessary and sufficient conditions 
for smoothness at fixed points and horizons are now satisfied.

Let us now look at the asymptotics as $r\to\infty$. As shown in Theorem \ref{theorem_harmonicfunctions}, asymptotic flatness requires the parameter constraint (\ref{eq_thm_AF_h}). Summing (\ref{eq_thm_ctr}) over all centres and using (\ref{eq_thm_AF_h}) we find (\ref{eq_thm_AF_m}). Using this together with (\ref{eq_omegapsi}) and (\ref{eq_f}), one can easily check that $\omega_\psi = \ord(r^{-1})$ and $f=1+\ord(r^{-1})$ as $r\to \infty$. Now we expand the 1-forms. Notice that 
    \begin{equation}
	       \chi= \left( \sum_{i=1}^N \chi_0^i   + \cos \theta\right) \td \phi  + \mathcal O(r^{-2})\td\bm x  \; ,
	\end{equation}
    so upon comparison to (\ref{eq_chiflat}) we deduce $\tilde\chi_0= \sum_{i=1}^N \chi_0^i $ where (\ref{eq_eulerinfinity}) gives the coordinates that are manifestly asymptotically flat. Furthermore, recall that $\omega_{-1}^i=0$ in (\ref{eq_omega_beta}), so  from the explicit expression for $\beta_{ij}$ we obtain that $\hat\omega = \mathcal O (r^{-2})\td \bm x$. Thus, the solution is asymptotically flat.
    
    Now, recall that we chose a gauge where $\tilde\chi_0$ is an odd integer, so we must check that this is compatible with the smoothness constraints at all centres. Indeed, we have $\tilde\chi_0= \sum_{i=1}^N \chi_i \equiv \sum_{i=1}^N h_i=1 \mod 2$, where in the second equality we used that $\chi_i$ and $h_i$ are of the same parity and in the final equality (\ref{eq_thm_AF_h}). Therefore, $\tilde\chi_0$ is automatically odd  which means that the smoothness constraints for $\chi$ at all centres can be satisfied simultaneously.

Finally,  (\ref{eq_thm_N}) is required by Lemma \ref{lemma1} and is equivalent to smoothness of the solution in the DOC away from the fixed points.
\end{proof}

\noindent \textbf{Remarks.} 
\begin{enumerate}
\item It is not known if the conditions (\ref{eq_thm_AF_h})-(\ref{eq_thm_N}) in the above theorem are also sufficient for the solution to have a 
	globally hyperbolic DOC, although it is clear that other assumptions in (\ref{assumption1}) and (\ref{assumption2}) are indeed satisfied. 
	In particular, it is not known if they are even sufficient for stable causality $g^{tt}<0$ (which is a consequence of global hyperbolicity). 
	In \cite{avila_one_2018} the authors argue that for solitons (\ref{eq_thm_N}) implies stable causality and support this with numerical evidence. 
	In Appendix \ref{appendix_Numerics} we present a numerical analysis of three-centred solutions and find that those configurations which satisfy 
	(\ref{eq_thm_AF_h})-(\ref{eq_thm_N}) together with positive mass $M>0$ (given by (\ref{eq_conserved})) are indeed stably causal.  This adds to previous 
	evidence for biaxially symmetric solutions that stable causality is a consequence of these conditions~\cite{kunduri_black_2014, kunduri_supersymmetric_2014, breunholder_supersymmetric_2019, breunholder_moduli_2019}.
	
\item The explicit form of the 1-forms (\ref{chixi_sol}) possess $N$ Dirac string singularities parallel to the $z$-axis and is a gauge choice.
While these cannot be removed simultaneously, by a local coordinate transformation each string can be rotated into any  direction. The spacetime can be covered  by a family of charts in which the strings are arranged to be between every other consecutive centre (for some arbitrary ordering of the centres),  similarly to the multi-centred gravitational instantons~\cite{gibbons_gravitational_1978}.

\end{enumerate}

\subsection{Conserved charges and fluxes}

Conserved charges associated to Killing fields can be determined by Komar integrals. Therefore the mass $M$ and angular momentum $J_\psi$ can be computed using Komar integrals with respect to the Killing vectors $\partial_t, \partial_\psi$ respectively. However, as we will show in the next section, generically $\partial_\phi$ is not a Killing vector. In order to compute the angular momentum $J_\phi$ associated to this, one can compare the asymptotic form of the metric to that of a localised source with known charges \cite{emparan_black_2008}. In order for the metric to be in the right form asymptotically one must align the $z$-axis in $\mathbb{R}^3$ with the direction of $\bm D$ defined as 
\cite{berglund_supergravity_2006}
\begin{equation}
	\bm D  = \sum_{i=1}^N\left(k_i -h_i\sum_{j=1}^N k_j\right)\bm a_i \; .
\end{equation} 
Note that this does not depend on the choice of origin in $\mathbb{R}^3$ due to (\ref{eq_thm_AF_h}).
One then finds that the conserved charges of the solution in Theorem \ref{thm_classification} are 
\begin{align}
	M &= 3\pi \left(\left(\sum_{i=1}^N k_i\right)^2+\sum_{i=1}^N l_i\right)\; , \qquad Q = 2\sqrt{3}\pi\left(\left(\sum_{i=1}^N k_i\right)^2+\sum_{i=1}^n l_i\right)\; , \nonumber\\
	J_\psi &= 2\pi \left(\left(\sum_{i=1}^N k_i \right)^3+\frac{3}{2}\sum_{i,j=1}^N k_il_j+\sum_{i=1}^N m_i\right)\; ,\qquad J_\phi = 3\pi |\bm D|\; . \label{eq_conserved}
\end{align}
As expected, the solutions saturate the BPS bound with $M= \sqrt{3} Q/2$.  Notice these take the same form as in the biaxially symmetric case~\cite{breunholder_moduli_2019}.

Each curve in $\mathbb{R}^3$ that connects two centres lifts to a non-contractible 2-cycle in spacetime ending on fixed points of $W$ or a 
horizon component. These 2-cycles either smoothly cap off at the fixed points as the length of orbits of $W$ goes to zero, or end on the horizon. 
They have topology of $S^2$ (between fixed points), a disc (between a fixed point and a horizon), or a tube (between horizon components). The flux through 
a 2-cycle $C_{ij}$ between centres $\bm a_i$ and $\bm a_j$ is defined as 
\begin{equation}
	\Pi_{ij} := \frac{1}{4\pi}\int_{C_{ij}} F =  \lim_{r_i\to 0}A_\psi - \lim_{r_j\to 0}A_\psi \; ,
\end{equation}
where we used Stokes' theorem. Using the explicit form of the Maxwell field (\ref{eq_Maxwell}), and taking the limit yields
\begin{align}
	\lim_{r_i\to 0} A_\psi =
	 \begin{cases}
		-\frac{\sqrt{3}}{2}\frac{k_i}{h_i} \; ,& \text{if } i \text{ corresponds to a fixed point,}\\
		\frac{\sqrt{3}}{2}\frac{h_im_i+\tfrac{1}{2}k_i l_i}{k_i^2+h_il_i} \; ,& \text{if } i \text{ corresponds to a horizon.}
	\end{cases}
\end{align}

\subsection{Symmetries} \label{ssec_symmetries}

Theorem \ref{thm_classification} shows that there is no requirement that forces the centres $\bm a_i \in \mathbb{R}^3$ to be collinear on $\mathbb{R}^3$. Naturally, one would expect 
that a solution with centres in generic positions (although still satisfying the constraints of the theorem) has less symmetry than those with collinear centres. 
In this section we will show that this is indeed the case. We will only consider Killing fields of  $(\mathcal{M}, g)$ that commute with the supersymmetric Killing field $V$, that is, symmetries which are also symmetries of the base space (wherever it is defined). 

Due to asymptotic flatness any Killing field of  $(\mathcal{M}, g)$ must approach a linear combination of those of Minkowski at asymptotic infinity 
\cite{chrusciel_killing_2005}, and in Lemma \ref{lem_AF} we showed that $W$ generates an isoclinic\footnote{
	That is, $W$ generates a rotation around two orthogonal axes with the same angle.
} rotation at infinity, i.e. $W = \frac{1}{2}(J_{12}+J_{34}) + \mathcal{O}(R^{-1})$, where $J_{IJ}$ were defined in Lemma \ref{lem_AF}. Therefore, 
excluding boosts (as they do not commute with $V$) and $V$ itself, the algebra of the Killing fields of $(\mathcal{M}, g)$ must be isomorphic to some subalgebra of 
the lie algebra of the four-dimensional euclidean group $E(4)$ that contains $J_{12}+J_{34}$. Such algebras must contain one of the following~\cite{fushchich_continuous_1986}:
\begin{enumerate} [label=(\roman*)]
	\item a Killing field that commutes with $W$, e.g. a subalgebra $\langle J_{12}, J_{34} \rangle$, \label{symm_commuting}
	\item an $SU(2)$ subalgebra $\langle J_{12}+J_{34}, J_{13}-J_{24}, J_{14}+J_{23}\rangle$, \label{symm_SU2}
	\item an $E(2)$ subalgebra $\langle J_{12}+J_{34}, P_1, P_2\rangle$,\label{symm_E2}
\end{enumerate}
where the  translations $P_I=\partial_I$ in the coordinates of Lemma \ref{lem_AF}. We will now consider each case in turn.

For case \ref{symm_commuting} let us denote the additional Killing field by $\xi$, which by assumption
commutes with $V$ and $W$. The orbit space $\hat\Sigma$ inherits a metric from the spacetime $q_{\mu\nu} := g_{\mu\nu}- G^{AB}g_{A\mu}g_{B\nu}$, where $G^{AB}$ is the inverse of matrix $G_{AB}$ in (\ref{eq_Ndef}), and indices $\{A, B\}$ correspond to $\{t, \psi\}$. We find that this orbit space metric is
\begin{equation}
	q= \frac{1}{N} \td x^i \td x^i \; , \label{eq_OSmetric}
\end{equation}
which is non-singular on the region $N>0$ (i.e. in the DOC away from fixed points).
Now, since $\xi$ preserves both 
the spacetime metric and Killing fields $V$ and $W$, it follows that $\mathcal{L}_\xi q=0$ and $\mathcal{L}_\xi N=0$, i.e. the orbit 
space has a Killing field that preserves $N$. From the explicit form of the orbit space metric (\ref{eq_OSmetric}) it follows that $\xi$ 
is a Killing field of the euclidean metric on $\mathbb{R}^3$. From (\ref{eq_harmonicinvariant}) it follows that the harmonic functions 
$H, K, L, M$ are invariant under $\xi$.   Now, 1-parameter subgroups of the isometry group of euclidean space $\mathbb{R}^3$ are either 
closed (rotation) or unbounded (translation with possibly a rotation).  However, since $H$ is invariant under this  subgroup and $H\to 0$ 
at infinity by asymptotic flatness, it must be that $\xi$ generates a rotation.  Thus, $\xi$ is an axial Killing field of $\mathbb{R}^3$ and 
hence the centres $\bm a_i$ must be collinear.  In this case the full spacetime solution has  $\mathbb{R}\times U(1)^2$ symmetry  and 
corresponds to the biaxially symmetric case\cite{breunholder_moduli_2019}. Therefore, if the centres are not collinear the abelian isometry 
group of the solution cannot be larger than $\mathbb{R}\times U(1)$.

In case \ref{symm_SU2} the $U(1)$ isometry generated by $W$ is a subgroup of the $SU(2)$ isometry. The complex structures of the hyper-K\"ahler base space must carry a real 
three-dimensional representation of any subgroup of its isometry group, which for $SU(2)$ must be either the trivial or adjoint 
representation. The latter is excluded by the triholomorphicity of $W$, therefore the whole $SU(2)$ symmetry is necessarily triholomorphic. Furthermore, 
from the asymptotic form of the Killing fields it follows that around spatial infinity the $SU(2)$ action has three-dimensional orbits. Hyper-K\"ahler 
manifolds with cohomogeneity-1 triholomorphic $SU(2)$ symmetry belong to the BGPP class of solutions \cite{belinskii_asymptotically_1978}. In Appendix 
\ref{app_symmetries} we show that the only multi-centred asymptotically euclidean BGPP solution is the trivial flat metric on $\mathbb{R}^4$.  Therefore, 
case (ii) cannot happen.

In case \ref{symm_E2} two of the generators corresponding to translations commute, therefore there exists a linear combination $K$ of them which is 
triholomorphic \cite{gibbons_hidden_1988}. Since $W$ has closed orbits it must correspond to the rotation in $E(2)$ and hence $W, K$ are linearly 
independent triholomorphic Killing fields. It follows that $[W, K]$ is also triholomorphic and by the algebra of $E(2)$ it must be a third 
independent Killing field. Therefore, the whole $E(2)$ symmetry must be triholomorphic. Similarly to case \ref{symm_SU2}, the $E(2)$ action 
has generically three-dimensional orbits at spatial infinity. In Appendix \ref{app_symmetries} we derive the most general  hyper-K\"ahler metric 
with cohomogeneity-1 triholomorphic $E(2)$ symmetry and show that the only multi-centred Gibbons-Hawking metrics in this class is the trivial flat metric on 
$\mathbb{R}^4$. Therefore, case (iii) also cannot happen.

Therefore, we have shown that our general solution in Theorem \ref{thm_classification} generically possesses an abelian isometry group $\mathbb{R}\times U(1)$ 
and not larger. Furthermore, this is enhanced to $\mathbb{R}\times U(1)^2$ precisely if the centres are collinear in $\mathbb{R}^3$.
To our knowledge the construction outlined in this paper provides the first examples of smooth, asymptotically flat black hole solutions in higher dimensions 
with exactly a single axial symmetry (on top of a stationary symmetry), confirming the conjecture of Reall at least for supersymmetric black 
holes~\cite{reall_higher_2004}.
In Appendix \ref{appendix_Numerics} we present a numerical study of three-centred solutions and find that the parameter constraints 
(\ref{eq_thm_AF_h}-\ref{eq_thm_hor}) can be easily satisfied even for non-collinear centres. 
Furthermore, (\ref{eq_thm_N}) can be proven to hold for the case of a black lens $L(3,1)$. For the other horizon topologies 
we found a large set of parameters that  numerically satisfy the parameter constraints and for which (\ref{eq_thm_N}) holds.  Therefore, we expect that there 
is a vast moduli space of asymptotically flat supersymmetric black holes with exactly a single axial symmetry.

\section{Discussion}\label{sec_Discussion}

In this work we have presented a classification of asymptotically flat, supersymmetric black holes in five-dimensional minimal supergravity. Our main 
assumption is that there is an axial symmetry that `commutes' with the remaining supersymmetry, i.e. it preserves the Killing spinor. It is an interesting 
question whether there are black hole solutions if this assumption is relaxed. A  natural possibility would be to only require that the axial symmetry
commutes with the supersymmetric Killing field, but that the Killing spinor is not preserved by its flow. Then, 
the complex structures on the base space (spinor bilinears) are no longer preserved by the axial symmetry, i.e. the axial Killing field is not triholomorphic. In this case the hyper-K\"ahler base is not a Gibbons-Hawking metric, but  instead can be obtained 
by solving the $SU(\infty)$ Toda equation \cite{bakas_toda_1997}. There exist constructions of supersymmetric solutions without a triholomorphic axial 
symmetry \cite{bena_bubbles_2007}, however no smooth black hole solution of this class is known. 

There have been a number of previous attempts to construct supersymmetric black holes with exactly a single axial symmetry, however in all these cases 
the spacetime metric itself or the matter fields are  not smooth at the horizon~\cite{bena_one_2004, horowitz_how_2005, bena_black_2006, bena_sliding_2006, bena_coiffured_2014, candlish_smoothness_2010}. 
It is worth noting that even the Majumdar-Papapetrou static multi-black hole spacetimes do not have smooth horizons in five (or higher) dimensions~\cite{welch_smoothness_1995, candlish_smoothness_2007, gowdigere_smoothness_2019, gowdigere_smoothness_2014}.
In contrast to our construction, a  common feature of all these solutions is that they break the $U(1)$ rotational
symmetry that corresponds to the triholomorphic Killing field $\partial_\psi$ in our setting.  Generally, it is expected that breaking a rotational symmetry leads to non-smooth horizons, essentially because coordinate changes of the type (\ref{eq_coordTf}) necessarily introduce logarithmically divergent terms in the corresponding angular coordinate ($\psi$ in this case), 
which would cause infinite oscillation in the metric and or matter fields near the horizon~\cite{horowitz_how_2005, candlish_smoothness_2010}. Our construction avoids this problem because we break the symmetry in a direction which is not rotational\footnote{
	By rotational here we mean that the near-horizon geometry has non-zero angular momentum associated to that direction, i.e. in our coordinates $J_\psi\neq 0$, while $J_\phi=0$.}  near the horizon (i.e. the angular coordinate $\phi$ does not diverge at the horizon by (\ref{eq_coordTf})). Furthermore, the classification of near-horizon geometries~\cite{reall_higher_2004} shows that the rotational symmetry near the horizon is always triholomorphic, which suggests that regular black hole solutions without a triholomorphic rotational symmetry may not exist. However, to properly check this, one would need to carefully 
analyse the near-horizon geometry in coordinates corresponding to the aforementioned Toda system.

For the class of solutions studied in this work, it is not obvious how large the moduli space of black holes is. Theorem \ref{thm_classification} imposes equations and 
inequalities not only on the parameters of the harmonic functions, but also on a combination of harmonic functions (\ref{eq_thm_N}) which must be satisfied 
everywhere on $\mathbb{R}^3$. It would significantly advance our understanding of the moduli space of black holes if an equivalent condition were known 
on the parameters of harmonic functions. The positivity of the ADM mass has been conjectured to be sufficient (together with the other parameter constraints 
listed in Theorem \ref{thm_classification}) \cite{breunholder_supersymmetric_2019}, however our numerical results show that this is unfortunately not the case (see Appendix 
\ref{appendix_Numerics}), even for the biaxially symmetric case.  This  prevents us from having a totally explicit description of the black hole moduli space in this theory.

Finally, as mentioned in the introduction, the rigidity theorem  does not apply to supersymmetric black holes since the stationary Killing field is null on 
the horizon.  It is therefore possible that there are supersymmetric black holes with no axial symmetry at all.  Obtaining a complete 
classification of supersymmetric black holes in this theory would require analysis of this case too. This would for example include the Majumdar-Papapetrou 
multi-black holes, which have been shown to be the most general asymptotically flat, static, supersymmetric black holes in this theory (although the horizon 
has low regularity)~\cite{lucietti_all_2021}. However, the analysis of  general supersymmetric solutions with no axial symmetry likely would require new 
techniques. We leave this as an interesting problem for the future. \\

\noindent \textbf{Acknowledgements.}  DK is supported by an EPSRC studentship. JL is supported by a Leverhulme Research Project Grant.
\\

\noindent \textbf{Competing interests.} The authors have no relevant financial or non-financial interests to disclose.
\\

\noindent \textbf{Data availability.} The datasets generated or analysed during the current study are available from the corresponding author on reasonable request.

\appendix

\section{Killing spinor is preserved by triholomorphic symmetry} \label{app_KS}

In this section we show that the Killing spinor is preserved by a triholomorphic Killing field that commutes with the supersymmetric Killing field. 
We will work on regions where $f\neq 0$, which are dense in $\mathcal{M}$ (Corollary \ref{cor_fdense}), and by continuity of the 
Lie-derivative we argue that it must be preserved on the whole spacetime.

Let us choose a pseudo-orthonormal co-frame $e^\mu$, where $\mu=0, i$, and $i=1,2,3,4$,
\begin{align}
	e^0 := -\frac{V^\flat}{f} = f(\td t +\omega) \; , \qquad e^i := \frac{\tilde e^{\tilde i}}{\sqrt{f}} \; ,
\end{align}
where $\tilde e^{\tilde i}$, with $\tilde i=i$, are an orthonormal co-frame of the base metric $h$,\footnote{
	In this section we distinguish quantities on the base space with a tilde.} with spin connection 1-forms $\tilde\omega^{\tilde i}{}_{\tilde j}$ defined by  $\td\tilde{e}^{\tilde i} = -\tilde{\omega}^{\tilde i}{}_{\tilde j}\wedge \tilde{e}^{\tilde j}$. 
Now, by Lemma \ref{Lemma_GH} the axial Killing field $W$ preserves the base data $(f, \omega, h)$, so we are free to choose our co-frame such that $\mathcal{L}_W \tilde e ^{\tilde i}=\mathcal{L}_W  e ^i=0$, and $\mathcal{L}_We^0=0$ is automatic.  The  spin connection 1-forms  $\omega^\mu_{~\nu}$ of the spacetime are given  by\footnote{
	The spin connection should not be confused with the 1-form $\omega$ defined on the base.}
\begin{align}
	\omega^0{}_i = f_i e^0+\frac{1}{2}G_{ij}e^j \; , \qquad \omega^i{}_j = \tilde{\omega}^{\tilde i}{}_{\tilde j}+\frac{1}{2}G_{ij}e^0+\frac{1}{2}(f_ie^j-f_je^i) \; ,
	\label{eq_spinconnection}
\end{align}
where $f_ie^i:=f^{-1}\td f$ and $\tfrac{1}{2} G_{ij}e^i\wedge e^j:=f\td\omega$.

Let $\{E_\mu\}$ be the dual frame of vector fields to $\{e^\mu\}$ 
and expand the Killing field $W$ in this frame,
\begin{equation}
	W = (\iota_W e^0) E_0 + (\iota_W e^i) E_i \; .
\end{equation}
The corresponding metric dual with respect to $g$ can be written as
\begin{align}
	W^\flat = -(\iota_W e^0) e^0 + (\iota_W e^i) e^i=  -(\iota_W e^0) e^0 + \frac{\widetilde W^\flat}{f} \; ,
\end{align}
where $\widetilde W^\flat$ is the metric dual of $W$ with respect to the base metric $h$.

By definition \cite{fatibene_geometric_1996}, the Lie-derivative of the Killing spinor by $W$ is
\begin{align}
	\mathcal{L}_W\epsilon &:= \nabla_W \epsilon -\frac{1}{4}\td W^\flat \cdot \epsilon \\
	& = W(\epsilon) -\frac{1}{4}(\iota_W\omega_{\mu\nu})\gamma^\mu\gamma^\nu\epsilon - \frac{1}{8} (\td W^\flat)_{\mu\nu}\gamma^\mu\gamma^\nu\epsilon \; ,
	\label{eq_LieKS}
\end{align}
where the second line follows from the expression for the spinorial Levi-Civita connection $\nabla_X\epsilon = X(\epsilon) -\frac{1}{4}(\iota_W\omega_{\mu\nu})\gamma^\mu\gamma^\nu\epsilon$, and for the Clifford algebra we use conventions $\gamma^\mu\gamma^\nu+\gamma^\nu\gamma^\mu = -2g^{\mu\nu}$. Using (\ref{eq_spinconnection}) and 
$0=\mathcal{L}_W e^\mu = \iota_W\td e^\mu + \td \iota_W e^\mu$, a calculation yields
\begin{equation}
	\td W^\flat = -2(\iota_W\omega_{0i})e^0\wedge e^i - (\iota_W e^j) f_i e^i\wedge e^j - \frac{1}{2}(\iota_W e^0)G_{ij}e^i\wedge e^j + \frac{1}{f}\td \widetilde{W}^\flat \; .
\end{equation}
Substituting into (\ref{eq_LieKS}), and using again (\ref{eq_spinconnection}) for $\omega_{ij}$, we get
\begin{align}
	\mathcal{L}_W\epsilon =  \widetilde{D}_{\widetilde{W}}\epsilon -\frac{1}{4f}\td \widetilde{W}^\flat\cdot\epsilon \; ,
	\label{eq_KSLiemixed}
\end{align}
where $\widetilde D$ is the Levi-Civita connection of the base metric $h$, and $\widetilde{W}=\pi_*W$  the projection of $W$ to the base 
(see the proof of Lemma \ref{Lemma_GH}). We can expand $\td \widetilde{W}^\flat$ on the base tetrad as 
\begin{equation}
	\td \widetilde{W}^\flat = (\td \widetilde{W}^\flat)_{\tilde i \tilde j}\tilde e^{\tilde i} \wedge \tilde e^{\tilde j} = f (\td \widetilde{W}^\flat)_{\tilde i \tilde j} e^i\wedge e^j \; .
\end{equation}
This factor of $f$ cancels the one in the second term of (\ref{eq_KSLiemixed}) and we get that the Lie-derivative of the Killing spinor on the spacetime is the same 
as on the base\footnote{Since we are in an orthonormal frame the gamma matrices will have the exact same form 
as the spatial gamma matrices of the spacetime, i.e. $\gamma^i = \gamma^{\tilde i}$.},
\begin{align}
	\mathcal{L}_W\epsilon = \widetilde{D}_{\widetilde{W}} \epsilon -\frac{1}{8}(\td \widetilde{W}^\flat)_{\tilde i \tilde j}\gamma^{\tilde i}\gamma^{\tilde j}\epsilon = \widetilde {\mathcal{L}}_{\widetilde{W}}\epsilon \; .
	\label{eq_LieKSbase}
\end{align}
The Killing spinor $\epsilon$ on the r.h.s can be interpreted as a spinor on the base and takes the form~\cite{gauntlett_all_2003} 
\begin{equation}
	\epsilon = \sqrt{f}\eta \; ,
\end{equation}
where $\widetilde{D}\eta = 0$. It follows (since $f$ is also preserved by $W$) that the first term of (\ref{eq_LieKSbase}) vanishes. Now we use 
the fact that $\widetilde W$ is triholomorphic to deduce that $\td \widetilde W^\flat$ is self-dual \cite{gibbons_hidden_1988}, and since
$\gamma^{[\tilde i}\gamma^{\tilde j]} \epsilon$ is anti-self-dual \cite{gauntlett_all_2003} on the base, the second term of 
(\ref{eq_LieKSbase}) also vanishes. Thus, the Killing spinor is preserved by a 
triholomorphic Killing field that also commutes with the supersymmetric Killing field, as claimed.

\section{On the non-existence of exceptional orbits} \label{app_exceptional}

Here we give a more detailed argument against the existence of exceptional orbits in the DOC. As discussed in the proof of Lemma \ref{lemma_orbit}, such an 
exceptional orbit can only end on a horizon component with spherical topology, which has a neighbourhood where $f>0$. It follows that there exists an open 
region $U_0\subset\Sigma$ with a Gibbons-Hawking coordinate chart such that $U_0\cap E\neq \emptyset$. By making $U_0$ smaller if necessary, we may assume 
without loss of generality that all points in $ U_0 \cap E$ have the same isotropy group $\mathbb{Z}_n$ for some integer $n>1$, i.e. the parameter $\psi$ 
of orbits of $W$ is $4\pi/n$-periodic (assuming regular orbits have $4\pi$ periodicity). 

Let $\Omega_0$ be the $\psi=0$ hypersurface in $U_0$. Each orbit of 
$W$ intersects $\Omega_0$ at most once, since $x^i$ are constants on orbits of $W$, and $\bm x$ is injective on $\Omega_0$. It follows from the periodicities of 
the orbits that the flowout from $\Omega_0$ along $W$ is injective for flow parameters $(-2\pi/n, 2\pi/n)$, so $U_0$ can be chosen such that it is 
diffeomorphic to $(-2\pi/n, 2\pi/n)\times\Omega_0$ by the Flowout Theorem (see e.g. \cite{lee_introduction_2012}). 

Let us denote by 
$\Psi:\mathbb{R}\times \Sigma\to \Sigma$ the flow of $W$, and let $\Omega_{2\pi}:=\Psi(2\pi, \Omega_0)$. By the same arguments $U_{2\pi}:=\Psi(2\pi, U_0)$
is diffeomorphic to $(-2\pi/n, 2\pi/n)\times\Omega_{2\pi}$. Recall that exceptional orbits are smooth arcs in $\hat \Sigma$, so using $x^i$ as coordinates 
in $\Omega_0$ and $\Omega_{2\pi}$, they can be viewed as smooth curves $\gamma$ in some open set of $\mathbb{R}^3$, i.e. 
$U_0\cap E\cong  (-2\pi/n, 2\pi/n) \times \gamma$. Using periodicity of the regular and exceptional orbits, one can show that the set 
$U_0\cap U_{2\pi}\cong S\times \gamma$, where $S=S^1\setminus\{p\}$ for $n$ even or $S=S^1\setminus\{p, q\}$ for $n$ odd, and $p$ and $q$ 
are antipodal points of an $S^1$ orbit. Both $U_0$ and $U_{2\pi}$ are open in $\Sigma$ and so is their intersection, but $S\times \gamma$ 
is not (it is a two-dimensional submanifold of $\Sigma$), which is a contradiction.

\section{Smoothness of  \texorpdfstring{$\omega$}{TEXT} at fixed points} \label{app_omega}

In this section we show explicitly that the 1-form $\alpha$ defined in (\ref{eq_alphadef}) is smooth on $\mathbb{R}^4$ around a fixed point $r=0$.  
The harmonic function $F$ defined in (\ref{eq_doms}) can locally be expanded as
\begin{equation}
	F=\sum_{\substack{k=1 \\-k \le m\le k}}^\infty r^k f_{km}P_k^m(\cos\theta)e^{im\phi}\; ,
	\label{omegapsiform}
\end{equation}
where $P_k^m(x)$ are the associated Legendre functions and $f_{km}$ some complex coefficients. For reference, the associated Legendre functions  for $m\geq 0$ can be written as
\begin{equation}
	P_k^m(x)= (-)^m\frac{(k+m)!}{(k-m)!}P^{-m}_k(x) = (-)^m(1-x^2)^{m/2}\frac{\td^m}{\td x^m}P_k(x)\; ,
\end{equation}
where $P_k(x)$ are the Legendre polynomials.
Therefore, for any $|m| \leq k$ they are a product of $(\sin\theta)^{|m|}$ and a polynomial of $x:=\cos\theta$ of order $k-|m|$.

 Integrating (\ref{eq_doms}), or using (\ref{eq_omegasing}), yields\footnote{Here we are omitting the argument of associated Legendre functions which we always take to be $x=\cos\theta$, and dot denotes the derivative with respect to $x$.}
\begin{align}
	\hat\omega_\text{sing}&=\sum_{\substack{k=1 \\-k \le m\le k}}^\infty \frac{r^{k}}{k}f_{km}\left[(\sin\theta)^2 \dot P_k^me^{im\phi}\td\phi +\frac{im}{\sin\theta}P_k^me^{im\phi}\td\theta\right]\; .
	\label{eq_omega_sol}
\end{align}
The 1-form $\alpha$ defined by (\ref{eq_alphadef}) in this gauge  is given by
\begin{equation}
	\alpha = \sum_{\substack{k=1 \\-k \le m\le k}}^\infty r^kf_{km}e^{im\phi}\left[ h_{-1}P_k^m(\td \psi' + \cos\theta\td \phi') + \frac{(\sin\theta)^2}{kh_{-1}} \dot P_k^m{}\td\phi' +\frac{im}{k\sin\theta}P_k^m\td\theta\right]\; .
\end{equation}
The smoothness of axisymmetric terms ($m=0$) has been checked in \cite{breunholder_moduli_2019}, but for completeness we also include it here. The axisymmetric 
terms for a given $k$, using (\ref{EulerToR4coords}) and basic properties for Legendre polynomials, can be written as 
\begin{align}
	&h_{-1}f_{k0}r^k \left[P_k(\td\psi' + \cos\theta\td\phi') + \frac{(\sin\theta)^2}{k}\dot P_k\td\phi'\right] =\nonumber\\
	&=h_{-1}f_{k0}r^k\left[(P_k+ P_{k-1})\td\phi^++(P_k- P_{k-1})\td\phi^-\right]\;, 
\end{align}
where we have used the identity $(1\pm x) P_k\pm k^{-1}(1-x^2) \dot{P}_k= P_k \pm P_{k-1}$.
From the recursion formula for the Legendre polynomials it follows that $r^k(P_k\pm P_{k-1})\td\phi^\pm = r(1\pm\cos\theta)\tilde G_\pm\td\phi^\pm =\frac{1}{2} \tilde G_\pm X_\pm^2\td\phi^\pm$ 
for some smooth function $\tilde G_\pm$ of $X_\pm^2$ and therefore these terms are indeed smooth on $\mathbb{R}^4$~\cite{breunholder_moduli_2019}.

We now consider the non-axisymmetric terms  in $\alpha$ for a given $k$ and $m\neq 0$. We focus on terms with $m>0$ ($P_l^m$ and $P_l^{-m}$ only differs by a constant factor, so the analysis is essentially identical for $m<0$), which can be written using (\ref{eq_hcond}) as 
\begin{align}
	&h_{-1}r^k\left(\frac{(\sin\theta)^2}{k}\dot P_k^m+(1+\cos\theta)P_k^m\right)e^{im\phi}\td\phi^+\nonumber\\
	&+h_{-1}r^k\left(-\frac{(\sin\theta)^2}{k}\dot P_k^m+(1-\cos\theta)P_k^m\right)e^{im\phi}\td\phi^-\nonumber\\
	&+r^k\frac{im}{k\sin\theta}P_k^me^{im\phi}\td\theta\;. \label{eq_alphakm}
\end{align}
For the last line we use the recursion formula
\begin{equation}
	\frac{1}{\sqrt{1-x^2}}P_k^m=\frac{-1}{2m}\left[P_{k-1}^{m+1}+(k+m-1)(k+m)P_{k-1}^{m-1}\right]\;.
\end{equation}
The term containing $P_{k-1}^{m+1}$ can be written as
\begin{align}
	&r^kP_{k-1}^{m+1}e^{im\phi}\td\theta\sim(r\sin\theta e^{i\phi})^{m}\left[(r\cos\theta)^{k-m-2}+\dots\right] \left[X_+^2\td (X_-^2)-X_-^2\td (X_+)^2\right]\;,
\end{align}
where $\dots$ represents lower order terms in $\cos\theta$. Since this term is explicitly smooth, we will omit it in the further analysis.

In the first two lines of (\ref{eq_alphakm}) we can use the identity 
\begin{equation}
	(1-x^2)\dot P_k^m+kxP_k^m=(k+m)P_{k-1}^m\;,
\end{equation}
and after omitting the smooth factor $\frac{1}{k}(r\exp(i\phi)\sin\theta)^{m-1}$, we get 
\begin{align}
	&h_{-1}r^{k-m+1}\sin\theta \left((k+m)\tilde P_{k-1}^m+k\tilde P_k^m\right)e^{i\phi}\td\phi^+\nonumber\\
	&+h_{-1}r^{k-m+1}\sin\theta \left(-(k+m)\tilde P_{k-1}^m+k\tilde P_k^m\right)e^{i\phi}\td\phi^-\nonumber\\
	&-\frac{i}{2}r^{k-m+1}(k+m-1)(k+m)\tilde P_{k-1}^{m-1}e^{i\phi}\td\theta\;,
	\label{omegaremnant2}
\end{align}
where we defined $\tilde P_k^m:=\frac{P_k^m}{(\sin\theta)^m}=(-)^m \frac{\td^m}{\td x^m}P_k$. 
Let us now look at the real part of (\ref{omegaremnant2}) (the analysis of imaginary part gives the same result), which can be written as
\begin{align}
	&\frac{h_{-1}}{2}r^{k-m}(u_1u_3+u_2u_4)[(k+m)\tilde P^m_{k-1}+k\tilde P^m_k]\frac{u_1\td u_2-u_2\td u_1}{X_+^2} \nonumber\\
	&+\frac{h_{-1}}{2}r^{k-m}(u_1u_3+u_2u_4)[-(k+m)\tilde P^m_{k-1}+k\tilde P^m_k]\frac{u_3\td u_4-u_4\td u_3}{X_-^2} \label{eq_omegaremnant_real}\\ 
	&+\frac{h_{-1}}{4}r^{k-m}(u_2u_3-u_4u_1)(k+m-1)(k+m)\tilde P^{m-1}_{k-1}\left(-\frac{u_1\td u_1+u_2\td u_2}{X_+^2}+\frac{u_3\td u_3+u_4\td u_4}{X_-^2}\right).\nonumber
\end{align}
By standard properties of the Legendre polynomials one can show that we can write the following polynomials as
\begin{align}
	k\tilde P_k^m\pm(k+m)\tilde P_{k-1}^m+\frac{1}{2}(k+m-1)(k+m)\tilde P_{k-1}^{m-1} = (1\pm x)Q^\pm_{km}\; 
	\label{omegacond}
\end{align}
with some polynomials $Q^\pm_{km}(x)$ of order $k-m-1$. Using this property with the upper (lower) sign for the $\td u_1$ and $\td u_2$ ($\td u_3$ 
and $\td u_4$) terms, one can check that (\ref{eq_omegaremnant_real}) is smooth. In detail, we can rearrange the $\td u_1$ terms as
\begin{align}
	\frac{h_{-1}}{2}\frac{r^{k-m}}{X_+^2}&\bigg\{(u_1^2u_4-u_1u_2u_3)\left[k\tilde P_k^m+(k+m)\tilde P_{k-1}^m+\frac{1}{2}(k+m-1)(k+m)\tilde P_{k-1}^{m-1}\right]\nonumber\\
	&-\left[k\tilde P_k^m+(k+m)\tilde P_{k-1}^m\right](u_1^2+u_2^2)u_4\bigg\}\td u_1 = \\
	=\frac{h_{-1}}{2}&\left[\frac{1}{2}(u_1^2u_4-u_1u_2u_3)r^{k-m-1}Q^+_{km}-r^{k-m}\left(k\tilde P_k^m+(k+m)\tilde P_{k-1}^m\right)u_4\right]\td u_1\;,\nonumber
\end{align}
where we used (\ref{omegacond}) and $r(1+x)=r(1+\cos\theta)=\frac{1}{2}X_+^2$. Since $Q^+_{km}(x)$ is of order $k-m-1$, the last line 
is explicitly smooth. The argument for the $\td u_{A=2,3,4}$ terms is identical. This proves our claim that $\alpha$ is smooth on $\mathbb{R}^4$.

\section{Three-centred solutions}\label{appendix_Numerics}

The simplest examples with a single axial symmetry are three-centred solutions. We will focus on single black hole solutions, where without loss of 
generality, we will choose coordinates such that the origin corresponds to the black hole. We can use the redundancy $\Psi\to\Psi+c$ in the harmonic 
functions so that they have the following form:
\begin{align}
	H &=\frac{h_0}{r}+\frac{h_1}{r_1}+\frac{h_2}{r_2} \; ,&& K= \frac{k_1}{r_1}+\frac{k_2}{r_2} \; , \nonumber \\
	L &= 1+\frac{l_0}{r}-\frac{h_1k_1^2}{r_1}-\frac{h_2k_2^2}{r_2} \; ,&& M=-\frac{3}{2}(k_1+k_2)+\frac{m_0}{r}+\frac{k_1^3}{2r_1}+\frac{k_2^3}{2r_2} \; .\label{def_3params}
\end{align}
The remaining constraints of Theorem \ref{thm_classification} on the parameters are
\begin{align}
	&h_1^2=h_2^2=1 \; ,\qquad\qquad\qquad  \qquad h_0+h_1+h_2=1 \; ,\label{eq_3_h}\\
	&\frac{h_0k_i^3-2m_0h_i-3k_il_0}{|\bm a_i|}+(-1)^i\frac{(k_2 h_1-k_1h_2)^3}{|\bm a_1-\bm a_2|}+3(k_1+k_2)h_i-3k_i=0 \; , \text{ for } i=1, 2 \; ,\label{omegapsi0_3}\\
	&h_0l_0^3-h_0^2m_0^2>0 \; , \label{horizon_3}\\
	&h_i -\frac{h_i(h_0k_i^2-l_0)}{|\bm a_i|}-h_1h_2\frac{(k_1h_2-k_2h_1)^2}{|\bm a_1-\bm a_2|}>0 \; , \text{ for } i=1, 2 \; .\label{hifi0_3}
\end{align}
There are three different choices for the parameters $h_i$:
\begin{enumerate}[label=(\roman*)]
	\item $h_0=1 \; ,\; h_1 = 1 \; ,\; h_2 = -1 \; ,$  \label{3ctr_1}
	\item $h_0 = -1 \; ,\; h_1 = h_2 =1 \; ,$   \label{3ctr_2}
	\item  $h_0 = 3 \; ,\; h_1=h_2=-1 \; .$ \label{3ctr_3}
\end{enumerate}
Cases \ref{3ctr_1} and \ref{3ctr_2} correspond to a horizon topology $S^3$, while \ref{3ctr_3} is a black lens $L(3,1)$. 
There are two further conditions that have to be satisfied everywhere on $\mathbb{R}^3$ for smoothness and stable causality:
\begin{align}
	&N^{-1}\equiv K^2+HL>0 \label{smoothness_3} \; , \\
	&g^{tt}<0  \; .\label{causal_3}
\end{align}
In \cite{breunholder_supersymmetric_2019} it has been conjectured that in the $U(1)^2$-symmetric case (\ref{eq_3_h}-\ref{hifi0_3}) together with 
the positivity of the total mass implies that (\ref{smoothness_3}-\ref{causal_3}) are automatically satisfied. Hence, for the analysis presented here, 
we add the requirement of positivity of mass $M>0$, which in terms of the parameters is
\begin{equation}
	l_0-h_1k_1^2-h_2k_2^2+(k_1+k_2)^2>0 \; .\label{MADM_3}
\end{equation}
For the $U(1)^2$-symmetric black lens 
this conjecture was analytically proven for (\ref{smoothness_3}), while for the other two cases numerical evidence was presented. As mentioned in the remarks below Theorem \ref{thm_classification}, 
 it has been argued that for soliton solutions (\ref{causal_3}) follows from (\ref{smoothness_3})~\cite{avila_one_2018}. 

For the numerical checks, sets of parameters are constructed as follows. Without loss of generality,  we fix coordinates for the centres 
$\bm a_i = (x_i, y_i, z_i)$ such that $x_1=y_1=y_2=0$. For each case \ref{3ctr_1}-\ref{3ctr_3}, we first choose the parameters $k_1, k_2, z_1, z_2, x_2$ randomly and then determine $l_0$ and $m_0$ as a solution of the linear system of equations (\ref{omegapsi0_3})\footnote{
	The linear equations are not invertible for some non-generic values of $k_1,k_2$, so these cases have been treated separately.}. Then those sets of parameters that fail to satisfy 
(\ref{horizon_3}-\ref{hifi0_3}) or (\ref{MADM_3}) are discarded. Finally, (\ref{smoothness_3}-\ref{causal_3}) are checked 
numerically at random points. For concreteness, some examples satisfying all conditions (\ref{eq_3_h}-\ref{MADM_3}) are 
shown in Table \ref{tab_3ctr}. We will now consider the three cases in turn.

\begin{table}
	\begin{center}
\begin{tabular}{|c|c|c|c|c|c|c|c|}
	\hline
	$h_0$ & $h_1$ & $h_2$ & $k_1$ & $k_2$ & $z_1$ & $z_2$ & $x_2$ \\
	\hline\hline
	1 & 1 & -1 & -3.467 & -7.256 & 0.2422 & -5.083 & 1.606 \\ \hline
	1 & 1 & -1 & 1.112 & -5.219 & 1.479 & 0.3712 & 1.810 \\ \hline
	1 & 1 & -1 & -1.009 & 9.371 & 8.756 & 9.130 & -9.022 \\ \hline
	-1 & 1 & 1 & -8.155 & -9.288 & 1.562 & 7.167 & -5.651 \\ \hline
	-1 & 1 & 1 & -3.061 & -5.430 & 0.02954 & -2.200 & 7.453 \\ \hline
	-1 & 1 & 1 & 6.501 & 7.217 & 2.931 & 0.6217 & -6.930 \\ \hline
	3 & -1 & -1 & 5.049 & 7.799 & 0.2238 & -8.775 & 3.644 \\ \hline
	3 & -1 & -1 & -9.477 & -8.117 & 7.736 & -0.3053 & 0.4992 \\ \hline
	3 & -1 & -1 & 7.571 & 5.287 & 9.393 & 0.3150 & -1.048 \\ \hline
\end{tabular}\caption{Examples of parameters satisfying the constraint equations and inequalities (\ref{eq_3_h}-\ref{MADM_3}). 
For all cases $x_1=y_1=y_2=0$, and $l_0,m_0$ are determined by (\ref{omegapsi0_3}).}\label{tab_3ctr}
\end{center}
\end{table}

For case \ref{3ctr_3}, the black lens, the proof that (\ref{smoothness_3}) is automatically satisfied carries over to the case with non-collinear centres as follows. Using 
(\ref{hifi0_3}) for $N^{-1}$ yields
\begin{align}
	N^{-1}>(k_1-k_2)^2&\left(\frac{|\bm a_1|}{|\bm a_1-\bm a_2|r r_1}+ \frac{|\bm a_2|}{|\bm a_1-\bm a_2|r r_2}-\frac{1}{r_1r_2}\right)\nonumber\\
	&+\frac{3 l_0}{r^2} +\left(\frac{3}{r}+\frac{|\bm a_1|}{r r_1}+\frac{|\bm a_2|}{rr_2}-\frac{1}{r_1}-\frac{1}{r_2}\right) \; .
\end{align}
The first term is non-negative by Ptolemy's inequality\footnote{
	When the centres are not co-planar, a lower estimate can be used by looking at a co-planar arrangement, keeping 
	$r_1$, $r_2$, $r$, $|\bm a_1|$ and $|\bm a_2|$ the same, and increasing $|\bm a_1-\bm a_2|$.}. 
The second term is positive due to (\ref{horizon_3}), and the last term is non-negative by the triangle inequality. The stable causality 
condition (\ref{causal_3}) was checked numerically. In the numerical checks we have looked at $10^{4}$ random points in a radius $R$ of the origin 
such that $R>3 \max\{| \bm a_1|, | \bm a_2|\}$, chosen either uniformly, or centred around the origin, closer to the singular points. We found that 
(\ref{causal_3}) was satisfied by all $10^4$ set of parameters that satisfy (\ref{eq_3_h}-\ref{hifi0_3}) and (\ref{MADM_3}).

For case \ref{3ctr_1} and \ref{3ctr_2}, we can only check (\ref{smoothness_3}-\ref{causal_3}) numerically on parameters that satisfy 
(\ref{eq_3_h})-(\ref{hifi0_3}) and (\ref{MADM_3}). The method of the numerical checks is identical to the one described in case \ref{3ctr_3}.
For each case we have found a large parameter space for which all the constraints are satisfied. We have found that for case \ref{3ctr_2} 
equations (\ref{eq_3_h})-(\ref{hifi0_3}) and (\ref{MADM_3}) are sufficient to guarantee (\ref{smoothness_3}-\ref{causal_3}) for the tested $10^4$ set of parameters. 
In contrast, rather surprisingly, for case \ref{3ctr_1} there are even collinear (therefore $U(1)^2$-symmetric) configurations that 
satisfy (\ref{eq_3_h})-(\ref{hifi0_3}) and (\ref{MADM_3}) but violate (\ref{smoothness_3}-\ref{causal_3}). This disproves the conjecture of 
\cite{breunholder_supersymmetric_2019}. When the centres are collinear, these configurations all have $k_2/k_1\sim \mathcal{O}(1)$, and even in that 
parameter range they appear with small probability when the parameters are chosen uniformly, which explains why they have not been found previously. 
None of these violating configurations have equal momenta, which is the relevant case when one would like to compare these black holes to the BMPV solution. Also, 
no configuration have been found which satisfy (\ref{smoothness_3}) but violate (\ref{causal_3}), inline with the conjecture of \cite{avila_one_2018} (although that conjecture is for solitons). 

We have 
also looked at soliton solutions, where instead of (\ref{horizon_3}), we require $l_0=0$, $m_0=0$, and 
\begin{equation}
	h_0\left(1- \frac{h_1k_1^2}{|\bm a_a|}-\frac{h_2k_2^2}{|\bm a_2|}\right) >0 \; .
\end{equation}
For these configurations, we have found that for all tested configurations satisfying (\ref{eq_3_h})-(\ref{hifi0_3}) and (\ref{MADM_3}), 
(\ref{smoothness_3}-\ref{causal_3}) are also satisfied.

\section{Hyper-K\"ahler metrics with triholomorphic \texorpdfstring{$SU(2)$ or $E(2)$}{TEXT} isometry} \label{app_symmetries}

In this section we determine the possible multi-centred Gibbons-Hawking solutions with triholomorphic cohomogeneity-1 $SU(2)$ or $E(2)$ symmetry. 
For the latter case we derive the general cohomogeneity-1 hyper-K\"ahler metric with triholomorphic $E(2)$ symmetry.

\subsection{\texorpdfstring{$SU(2)$}{TEXT}}

Hyper-K\"ahler manifolds with triholomorphic cohomogeneity-1 $SU(2)$ symmetry belong to the BGPP class~\cite{belinskii_asymptotically_1978}.  
In Gibbons-Hawking form adapted to a triholomorphic subgroup $U(1)\subset SU(2)$ the associated harmonic function can be written as 
\begin{equation}
	H(\bm x) = \left(\prod_{i=1}^3(\beta(\bm x)-\beta_i)\right)^{-1/2}\left(\sum_{i=1}^3\frac{x_i^2}{(\beta(\bm x)-\beta_i)^2}\right)^{-1}\;, \label{eq_Hbeta}
\end{equation}
where $\beta(\bm x)$ is an algebraic root of
\begin{equation}
	\sum_{i=1}^3\frac{x_i^2}{\beta-\beta_i}=C\;,	\label{eq_betadef}
\end{equation}
with $\beta_i$, $C$ constants~\cite{gibbons_gravitational_2003,dunajski_harmonic_2003}. First note that if $\beta_1=\beta_2=\beta_3$ then 
$H = 1/(\sqrt{C} r)$, and if we take $C=1$ the solution reduces to flat euclidean space.  We will therefore exclude this case.

A simple pole in $H$ can occur at 
$\bm x_0=(x_0, y_0, z_0)$ only if $\beta(\bm x_0)=\beta_i$ for some $i$.\footnote{If $C\neq 0$ then $H$ is clearly non-singular at all 
points $\beta(\bm x)\neq \beta_i$.  If $C=0$ and $\beta(\bm x)\neq \beta_i$ then $H$ can only be singular at the origin in which case $\beta$ 
has a direction dependent limit, so it cannot be a simple pole of $H$.} For definiteness, assume that $H$ is singular at $\beta(\bm x_0) = \beta_1$, 
so the analysis splits into the cases:
\begin{enumerate}[label=(\roman*)]
	\item $\beta_1=\beta_2\neq \beta_3$, \label{case_1beta}
	\item $\beta_i$ are distinct.	\label{case_2beta}
\end{enumerate}

Now consider case \ref{case_1beta} for which $H$ is necessarily axisymmetric. From (\ref{eq_Hbeta}) we can see that $H$ is singular at $\bm x_0$ iff $(x^2+y^2)/(\beta-\beta_1)\to 0$ as $\bm x\to \bm x_0$.  Then (\ref{eq_betadef}) implies that $z_0^2= C(\beta_1- \beta_3)\geq 0$. If $C(\beta_1-\beta_3)>0$, there are isolated singularities at $x_0 = y_0 = 0$, $z_0 = \pm \sqrt{C(\beta_1-\beta_3)}$,
and the corresponding harmonic function takes the form
\begin{equation}
	H(\bm x) = \frac{1}{2\sqrt{C}}\left(\frac{1}{|\bm x-\bm a|}\pm \frac{1}{|\bm x+\bm a|}\right)\;,
\end{equation} 
with $\bm a = (0, 0, \sqrt{C(\beta_1-\beta_3)})$. This family includes the Eguchi-Hanson metric (with $+$ sign and $C=1/4$). These configurations, however,
do not correspond to asymptotically euclidean metrics. If $C=0$, the harmonic function is given by $H = \const \frac{z}{r^3}$, so this solution does not 
have a simple pole.

Finally, let us look at case \ref{case_2beta}. From (\ref{eq_Hbeta}) it can be shown 
that if $H$ has a simple pole at a given $\bm x_0$ then
\begin{align}
	\lim_{\bm x\to \bm x_0}\frac{x^2}{(\beta(\bm x)-\beta_1)^{3/2}}=0 \; .
	\label{eq_polelimit}
\end{align}
Using (\ref{eq_polelimit}) in (\ref{eq_betadef}) implies that $(x_0, y_0, z_0)$ is on the curve defined by
\begin{align}
	\frac{y^2}{\beta_1-\beta_2}+ \frac{z^2}{\beta_1-\beta_3}=C\;, \qquad x=0  \; .	\label{eq_yzcond}
\end{align}
Let $S$ be the connected component of (\ref{eq_yzcond}) containing $\bm x_0$. We now show that $H$ is singular on $S$. Recall that 
having a simple pole requires $\lim_{\bm x\to \bm x_0} \beta=\beta_1$. Taking this limit in the direction along the curve (\ref{eq_yzcond}) and 
using the fact that the other root of (\ref{eq_betadef}) on (\ref{eq_yzcond}) is separated from $\beta_1$, it follows that $\beta=\beta_1$ on $S$. Let 
$(0, \tilde y, \tilde z)\in S$, and let us look at the limit
\begin{align}
	\lim _{x\to0}&\frac{x^2}{(\beta(x, \tilde y, \tilde z)-\beta_1)^{3/2}} = \lim _{x\to0}\frac{C-\frac{\tilde y^2}{\beta-\beta_2}-\frac{\tilde z^2}{\beta-\beta_3}}{\sqrt{\beta-\beta_1}} = \nonumber\\
	&=\lim _{x\to0}2\left(\frac{\tilde y^2}{(\beta-\beta_2)^2}+\frac{\tilde z^2}{(\beta-\beta_3)^2}\right)\sqrt{\beta-\beta_1} = 0\;,
\end{align}
where for the second line we used L'H\^opital's rule. Thus, by (\ref{eq_Hbeta}) $H$ diverges on $S$, hence the singularity of $H$ at 
$\bm x_0$ cannot be isolated.

\subsection{\texorpdfstring{$E(2)$}{TEXT}}
We will first derive the general form for cohomogeneity-1 hyper-K\"ahler metrics with triholomorphic $E(2)$ symmetry. Without loss of 
generality we can write this as
\begin{align}
	h =(\det h(\rho))\td\rho^2 + h_{ij}(\rho)\sigma_R^i\sigma_R^j\;, \label{eq_cohom1metric}
\end{align}
where $\sigma_R^i$, $i, j=1,2,3$, are left-invariant 1-forms on $E(2)$, satisfying 
\begin{align}
	\td\sigma_R^i=\frac{1}{2}c^i{}_{jk}\sigma_R^j\wedge\sigma_R^k\;, 
\end{align}
with $c^i{}_{jk}$ being the structure constants of the Lie algebra of $E(2)$. The Killing fields of (\ref{eq_cohom1metric}) are given by the right-invariant vector fields $L_i$ which satisfy $[L_i, L_j]= c^k_{~ij} L_k$. Let 
us define $t^{im}:=\tfrac{1}{2} \epsilon^{ijk}c^m{}_{jk}$. For $E(2)$, $t^{ij}$ is a symmetric, positive semi-definite matrix with one zero eigenvalue 
(VII$_0$ in the Bianchi classification).

By a change of basis $\sigma_R^i{}'= L^i{}_j\sigma_R^j$, where $L\in GL(3, \mathbb{R})$, one can simultaneously diagonalise $h$ and $t$ at a given 
$\rho=\rho_0$ as follows. Under such change $h' = (L^{-1})^T h L^{-1}$ and $t' = (\det L)^{-1}L t L^T$. Since $h$ is symmetric and positive definite, we 
can change basis such that $h' = 1$. As $t'$ is symmetric, we can diagonalise it by an orthogonal matrix so that $t'' = \diag (0, t_2, t_3)$ 
and $h'' = 1$. Finally, we can rescale each direction such that $t''' = \diag (0, 1, 1)$ and $h'''$ is diagonal.   

We now drop primes and assume that $h_{ij}(\rho_0)$ is diagonal and $t=\diag (0,1,1)$.  The latter are equivalent to the structure constants $c^3_{~12}=1 = c^2_{~31}$ with the rest vanishing. Now, following 
\cite{huggett_cohomogeneity-one_2017, dammerman_diagonalizing_2009}, the Einstein condition for the metric (\ref{eq_cohom1metric}) then gives at $\rho_0$
(omitting primes)
\begin{align}
	\dot h_{13}=\dot h_{12}=0\;,\qquad (h^{22}-h^{33})\dot h_{23}=0\;, \label{eq_Einsteincond}
\end{align}
where dot means derivative with respect to $\rho$. Generically, if no other symmetry is assumed, this implies that the first derivative of the off-diagonal 
metric components vanish. If there is additional symmetry and $h^{22}=h^{33}$, then $h$ is automatically diagonal \cite{dammerman_diagonalizing_2009}. By 
further differentiation of the Einstein condition, using real-analyticity\footnote{$h_{ij}$ and $K_{ij}:=\dot h_{ij}$ satisfy a system of first 
order ODEs (the Einstein equations) which are analytic in $(h_{ij}, K_{ij})$ if $\det h \neq0$, so by standard results,  there is a unique solution analytic at $\rho=\rho_0$ for given `initial' values at $\rho=\rho_0$.} in $\rho$, 
one can see that the metric remains diagonal for all $\rho$.

After diagonalisation, the metric can be put into the form 
\begin{align}
	h = \omega_1\omega_2\omega_3\td \rho^2 +\frac{\omega_2\omega_3}{\omega_1} (\sigma_R^1)^2 +\frac{\omega_1\omega_3}{\omega_2} (\sigma_R^2)^2 + \frac{\omega_1\omega_2}{\omega_3} (\sigma_R^3)^2\;,
\end{align}
for some functions $\omega_i(\rho)$. We can take a basis of anti-self-dual (ASD) 2-forms 
\begin{align}
	\Omega^1 &=\omega_2\omega_3 \td \rho \wedge \sigma_R^1 - \omega_1\sigma_R^2\wedge\sigma_R^3 \; , \nonumber \\
	\Omega^2 & =\omega_1\omega_3 \td \rho \wedge \sigma_R^2 - \omega_2\sigma_R^3\wedge\sigma_R^1 \; , \\
	\Omega^3 &= \omega_1\omega_2 \td \rho \wedge \sigma_R^3 - \omega_3\sigma_R^1\wedge\sigma_R^2    \; ,\nonumber
\end{align}
where we choose the orientation $\td\rho\wedge\sigma_R^1\wedge\sigma_R^2\wedge\sigma_R^3$. The requirement that $\Omega^i$ are closed 
is equivalent to the systems of ODE:
\begin{equation}
\dot{\omega}_1=0\; , \qquad \dot{\omega}_2=- \omega_1 \omega_3 \; , \qquad \dot{\omega}_3= -\omega_1 \omega_2   \; .
\end{equation}
In particular,  $\omega_1$ is a constant which must be non-vanishing, so without loss of generality we can assume $\omega_1>0$ 
(this can be arranged using the discrete symmetry $\omega_1\to -\omega_1, \omega_2\to -\omega_2$ of the above system).  We can 
now easily integrate for $\omega_{2,3}$, which gives us two qualitatively different classes of solutions. The first is 
$\omega_2= A e^{\pm \omega_1 \rho}$, $\omega_3=\mp A e^{\pm \omega_1 \rho}$, and it is easily checked that this gives euclidean space. 
The second class of solutions can be written as
\begin{align}
 \omega_2 = -A \sinh(\omega_1(\rho-\rho_0))\;, \qquad \omega_3 = A \cosh(\omega_1(\rho-\rho_0))\;,
\end{align}
with $A$ and $\rho_0$ arbitrary constants and without loss of generality we can set $\rho_0 = 0$.  Let us focus on this second non-trivial solution.

Now, by defining a new coordinate $\hat\rho:= - \omega_1 \rho$ and by rescaling $\hat\sigma_R^{2,3}= \sqrt{\omega_1} \sigma_R^{2,3}$, which does not 
change the structure constants,  we obtain the metric (dropping hats)
\begin{align}
h = K^2 \sinh \rho \cosh \rho\, (\td \rho^2+ (\sigma_R^1)^2) + \coth \rho\,  (\sigma_R^2)^2+ \tanh \rho \, (\sigma_R^3)^2  \; , \label{eq_hE2}
\end{align}
where $K^2:= A^2/ \omega_1 $. The ASD 2-forms are 
\begin{align}
	\Omega^1 &=- K^2\sinh \rho\cosh\rho\, \td \rho \wedge \sigma_R^1 - \sigma_R^2\wedge\sigma_R^3 \; , \nonumber\\
	\Omega^2 & =- K(\cosh \rho\,  \td \rho \wedge \sigma_R^2+ \sinh\rho\, \sigma_R^3\wedge\sigma_R^1 ) \; ,\\
	\Omega^3 &= - K( \sinh\rho\, \td \rho \wedge \sigma_R^3 +\cosh\rho\, \sigma_R^1\wedge\sigma_R^2)    \; , \nonumber
\end{align}
where note that the orientation is now $-\td\rho\wedge\sigma_R^1\wedge\sigma_R^2\wedge\sigma_R^3$. It is useful to note that we can always choose local 
coordinates $(Z, x,y)$  on $E(2)$ so
\begin{align}
	&\sigma_R^1 = \td Z\;, \qquad \sigma_R^2 = -\sin Z \td x + \cos Z \td y\;, \qquad \sigma_R^3 = \cos Z \td x +\sin Z \td y\;.
\end{align}
In these coordinates the Killing fields that generate the $E(2)$ algebra are
\begin{align}
	&L_1 = \partial_Z + x\partial_y - y\partial_x\;, \qquad L_2 = \partial_y\;, \qquad L_3 = \partial_x\;.
\end{align}
This completes the derivation of the general form of the metric which is given by (\ref{eq_hE2}). This 
is the analogue of the BGPP metric for $E(2)$-symmetry.

As $\rho \to\infty$ (\ref{eq_hE2}) approaches the metric
\begin{equation}
h_0 = \td R^2 + R^2 \td Z^2 + \td x^2+\td y^2,   \;  \label{eq_asymptotic_hE2}
\end{equation}
with $R := \frac{K}{2}\exp(\rho)$.  Even though (\ref{eq_asymptotic_hE2}) is (locally) isometric to the flat euclidean metric on $\mathbb{R}^4$,  it is 
only valid at $R\to\infty$, thus (\ref{eq_hE2}) is not asymptotically euclidean (one can also check that $R_{abcd}R^{abcd}\to 0$ iff $R\to\infty$).  This excludes (\ref{eq_hE2}) as a possible base for an asymptotically flat supersymmetric solution.

Even though (\ref{eq_hE2}) does not have the appropriate asymptotic behaviour, it is interesting to write it in Gibbons-Hawking form with respect to the 
Killing field $W' = \frac{1}{2}L_1$ (the analogue of an isoclinic Killing field $W$ for asymptotically euclidean metrics). The coordinates $(\psi, x^i)$ 
adapted to $W'$ are defined by $\td x^i= \iota_{W'} \Omega^i$ and $W'=\partial_\psi$. We find that they take a simpler form in polar 
coordinates $(x, y)\to (r, \theta)$, and a computation using the above ASD 2-forms reveals that
\begin{align}
	&x^1 = \frac{K^2\sinh^2\rho - r^2}{4}\;, \qquad x^2 = \frac{K}{2}r \sinh \rho \cos\phi\;, \\
	&x^3 = \frac{K}{2}r\cosh \rho \sin\phi\;, \qquad \psi = Z+\theta \; , 
\end{align}
where  $\phi:=Z-\theta$.
The associated harmonic function is given by $H=1/h(W',W')$, and we find
\begin{equation}
	H = \frac{16\sinh(2\rho)}{K^2(\cosh(4\rho)-1)+4r^2(\cosh(2\rho)+\cos(2\phi))}\;.
\end{equation}
This diverges on a parabola in $\mathbb{R}^3$ parametrised by $r$:
\begin{align}
	x^1 = -\frac{r^2}{4}\;, \qquad x^2 = 0\;, \qquad x^3=\pm \frac{K}{2}r\;,
\end{align}
and therefore its singularity is not isolated.

\bibliographystyle{sn-mathphys}
\bibliography{references}
%\printbibliography %Prints bibliography if biblatex used
\end{document}